\documentclass[11pt]{article}
\usepackage{arxiv}

\usepackage[T1]{fontenc}
\usepackage[utf8]{inputenc}% utf8 characters
\usepackage[english]{babel}% english language 

\usepackage{doi,graphicx,flafter,placeins,hyperref,amsfonts,amsmath,nicefrac,color,soul,longtable,colortbl,setspace,ifthen,xspace,url,pdflscape,amsthm,amssymb,mathtools,tikz,xcolor,rotating,tabularx,multirow,subcaption,enumerate,algorithm,algpseudocode,bm,etoolbox,mwe,booktabs}

\hypersetup{
	colorlinks=true,
	linkcolor=blue,
	filecolor=magenta,      
	urlcolor=black,
}

\usetikzlibrary{arrows,decorations,backgrounds,calc,fit,shapes,positioning,shapes.geometric} 
\tikzset{
	global scale/.style={
		scale=0.6,
		every node/.style={scale=0.6}
	}
}

\graphicspath{{./png_files/}}

\newtheorem{remark}{Remark}
\newtheorem{assumption}{Assumption}

\newtheorem{proposition}{Proposition}

\newtheorem{corollary}{Corollary}

\newtheoremstyle{nonitalic} % Name of the style
  {}                        % Space above
  {}                        % Space below
  {\normalfont}             % Body font (no italics)
  {}                        % Indent amount
  {\bfseries}               % Theorem head font (bold)
  {.}                       % Punctuation after theorem head
  { }                       % Space after theorem head
  {}                        % Theorem head spec (empty means 'Example X.')
\theoremstyle{nonitalic}
\newtheorem{example}{Example}
\theoremstyle{nonitalic}
\newtheorem{definition}{Definition}

\hypersetup{
pdftitle={Algorithmic Collusion And The Minimum Price Markov Game},
pdfsubject={cs.GT, cs.LG, econ.QM},
pdfauthor={Igor Sadoune, Andrea Lodi, Marcelin Joanis},
pdfkeywords={Algorithmic Game Theory, Multiagent Reinforcement Learning, Prisoner's Dilemma, Auction Game},
}

\title{Algorithmic Collusion And The Minimum Price Markov Game}

\author{
    \href{https://orcid.org/0000-0001-9332-9793}{\includegraphics[scale=0.06]{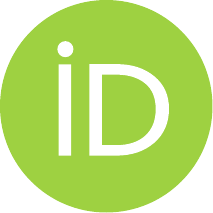}\hspace{1mm}Igor Sadoune} \\
    Department of Mathematics and Industrial Engineering\\
    Polytechnique Montreal\\
    CIRANO\\
    Montreal, QC, Canada \\
    \texttt{igor.sadoune@polymtl.ca} \\
    \And
    \href{https://orcid.org/0000-0003-3736-2761}{\includegraphics[scale=0.06]{orcid.pdf}\hspace{1mm}Marcelin Joanis} \\
    Department of Mathematics and Industrial Engineering\\
    Polytechnique Montreal\\
    CIRANO\\
    Montreal, QC, Canada \\
    \texttt{marcelin.joanis@polymtl.ca} \\
    \And
    \href{https://orcid.org/0000-0001-9269-633X}{\includegraphics[scale=0.06]{orcid.pdf}\hspace{1mm}Andrea Lodi} \\
    Jacobs Technion-Cornell Institute\\
    Cornell Tech and Technion - IIT\\
    New York, NY, USA \\
    \texttt{andrea.lodi@cornell.edu} \\
}

\date{}

\renewcommand{\headeright}{}
\renewcommand{\undertitle}{}
\renewcommand{\shorttitle}{}

\begin{document}
\maketitle

\begin{abstract}
This paper introduces the Minimum Price Markov Game (MPMG), a theoretical model that reasonably approximates real-world first-price markets following the minimum price rule, such as public auctions. The goal is to provide researchers and practitioners with a framework to study market fairness and regulation in both digitized and non-digitized public procurement processes, amid growing concerns about algorithmic collusion in online markets. Using multi-agent reinforcement learning-driven artificial agents, we demonstrate that (i) the MPMG is a reliable model for first-price market dynamics, (ii) the minimum price rule is generally resilient to non-engineered tacit coordination among rational actors, and (iii) when tacit coordination occurs, it relies heavily on self-reinforcing trends. These findings contribute to the ongoing debate about algorithmic pricing and its implications.
\end{abstract}


\textit{To replicate the study, follow}  \href{https://github.com/IgorSadoune/algorithmic-collusion-and-the-minimum-price-markov-game}{\color{blue}{replication repo (GitHub)}}. 

\textit{MPMG is implemented in a Python package; see}  \href{https://pypi.org/project/mpmg/0.2.3/}{\color{blue}{PyPI}} or \href{https://github.com/IgorSadoune/mpmg-python-package}{\color{blue}{GitHub}}.

\section{Introduction}
\label{sec:Introduction}

Concerns about algorithmic pricing have grown as technological advancements in machine learning, widespread data availability, and the digitization of the economy have created fertile grounds for such practices \cite{Gerlick2020, Gormsen2022, Ezrachi2020, Li2023}. Algorithmic pricing can potentially lead to collusion, as observed in digital markets such as online retail \cite{Calvano2020a, Assad2020} and electricity supply markets \cite{Rothkopf1999, Ye2019, Tellidou2007, Viehmann2021, Matsukawa2019}. These sectors demonstrate the potential for artificial decision-makers to learn collusive behaviors through dynamic pricing algorithms and optimal bidding strategies. This echoes the concerns about algorithmic collusion in public procurement, which are increasingly relevant, especially as digital platforms and automated decision-making tools become more prevalent in these settings. Since public procurement is a crucial mechanism for collective welfare, the vulnerability of this sector to collusive outcomes is particularly alarming. For example, it has been shown that bidding patterns were incompatible with a competitive equilibrium in Ukrainian E-procurement data \cite{Baranek2021}. Given these developments, questions arise. Is public procurement directly threatened by algorithmic pricing? And if so, under which conditions?

In contrast to digital markets, public procurement in most Western economies typically adheres to first-price sealed-bid auction rules, and more specifically, to the minimum price rule that determines the winner by selecting the lowest bid. This rule is supported by a longstanding body of literature that demonstrates the theoretical efficiency of first-price auctions and, by extension, the efficacy of the minimum price rule in maximizing collective welfare \cite{Laffont1997a, Lebrun2006, Bergemann2017}. Public procurement processes, generally not digitized, have made more data available to the public due to recent transparency laws, a response to a long history of traditional cartel collusion \cite{Chassin2010, Harrington2008, McAfee1992}. This initiative has been effective as it has helped regulators and researchers detect collusion in public markets \cite{Tas2017, Chassang2019}. However, the data can also be used to train artificial decision-makers to provide optimal bidding strategies. The potential for these strategies to lead to algorithmic collusion, whether accidental or deliberately learned, raises concerns for legal authorities because such collusion would not trigger antitrust laws \cite{Calvano2020b}. Indeed, accidental algorithmic collusion stems from profit-maximization motives, and proving the deliberate formation of algorithmic coordination remains challenging.

In this paper, we introduce the Minimum Price Markov Game (MPMG), a model designed to provide a framework for studying emergent behavior in the context of minimum price auctions, commonly used in real-world public procurement. We draw on Robert Axelrod's concept of ``the emergence of cooperation among egoist'' \cite{Axelrod1981}, exploring tacit coordination among isolated agents in the MPMG. Departing from the traditional game-theoretic Bayesian framework of first-price auctions, we consider a game of complete information based on the assumption of data availability. Indeed, public auction data typically includes the features of auctions and agents from past to ongoing contracts, as well as the associated bids \cite{Sadoune2024}. This allows us to base the MPMG on the Prisoner's Dilemma (PD), leveraging a structure that captures the essence of public procurement as a social dilemma—competitive by law but cooperative by nature. Indeed, the PD is well known for encapsulating the paradox of rationality central to our discussion on the study of emergent coordination among isolated agents in public markets. Furthermore, we test the MPMG using various state-of-the-art Multi-Agent Reinforcement Learning (MARL) methods and thereby show that the minimum price rule, in our setting, is generally robust but not impervious to tacit coordination among isolated decision-makers.

Our methodological choice is motivated by extensive literature that shows the effectiveness of combining game theory and reinforcement learning for modeling strategic interactions and optimizing behavior through dynamic adaptation among agents. We explore such literature in Section \ref{sec:Related Work} and show that our methodological framework is ideal for exploring emergent behaviors and potential tacit collusion in public procurement markets. In Section \ref{sec:The Minimum Price Markov Game}, we provide the mathematical model for the MPMG and discuss the benefits of the Markov property in our setting. Section \ref{sec:ComputerExperiments} presents the results of our computer experiments and concludes on their implications for tacit coordination in minimum price-driven auctions. Finally, Section \ref{sec:Discussion} provides interpretation, practical considerations and perspectives.

\section{Related Work}
\label{sec:Related Work}
The interplay among Markov Decision Process (MDP), auction design, and multi-agent learning has gained significant attention in recent research. Although the literature addressing algorithmic collusion within this framework is relatively sparse, our study draws inspiration from the growing body of work on integrating Stackelberg game models within multi-agent learning frameworks, MDP for auction designs, and tacit coordination in social dilemmas. By reviewing these interconnected areas, we aim to highlight the advancements in understanding strategic interactions, optimal bidding strategies, and tacit coordination among agents. This overview highlights the pivotal role of MDP and related methodologies in shaping contemporary economic and strategic decision-making processes.

\paragraph{MDP and auction design.} The MDP framework facilitates the exploration, understanding, and prediction of strategic interactions and behaviors within market designs. Examples include computing Nash equilibria in auction games \cite{Graf2023, Bichler2023} and optimizing bidding strategies in combinatorial auctions \cite{Brero2017, Beyeler2021}. Although the Markov game framework is influential, bandit algorithms have also proven to be exceptionally suitable for modeling optimal bidding strategies in online auction designs \cite{Hummel2016, Nguyen2020, Gao2021, Basu2022}. This is due to the simplicity and effectiveness of bandit algorithms, which operate in a stateless or single-state environment and make decisions independently of past actions without the need for an underlying state transition model.

\paragraph{MDP and Stackelberg game.} Stackelberg models, in which players move sequentially rather than simultaneously and therefore were historically used to study hierarchical or sequential leadership, are being increasingly used in various applications within the Markov game and multiagent learning framework \cite{Goktas2022}. The Stackelberg game is a staple in the analysis of various economic and strategic situations where the ability to move first or commit to a decision before one’s competitors can confer a strategic advantage. This includes industries with clear market leaders who can influence the market environment \cite{Fudenberg1991}. For example, Reinforcement Learning (RL) agents have been trained to learn the Stackelberg equilibrium in a Markov decision process representing economic mechanisms for assigning items to individualistic agents in sequential-pricing stages \cite{Brero2022, Brero2021a, Brero2021b}. 

\paragraph{Multiagent learning, social dilemmas, and tacit coordination.} In the exploration of tacit collusion and cooperative strategies within repeated game frameworks, a multitude of studies have contributed to a nuanced understanding of how agents can implicitly coordinate their actions under restricted communication conditions to optimize collective outcomes. The effectiveness of collusion in auction settings without explicit communication is a recurring theme \cite{Skrzypacz2004}, with similar phenomena observed in oligopoly models through Q-learning mechanisms where collusion emerges without direct interaction among agents \cite{Waltman2008,Kimbrough2005}. These findings are complemented by insights into multi-agent control using deep reinforcement learning, which further illustrates how complex cooperative behaviors can be engineered even in large-scale agent systems \cite{Gupta2017}. 

In scenarios where direct communication is either limited or non-existent, learning algorithms play a pivotal role in facilitating implicit cooperation and understanding competitive behaviors. For instance, learning methods that incorporate the anticipation of other agents' learning processes have shown promise in fostering cooperative strategies that are robust against exploitation \cite{Foerster2018a}. This is particularly evident in the Iterated Prisoner’s Dilemma, where strategies developed through evolutionary algorithms have shown to outperform traditional approaches \cite{Harper2017}. Indeed, the contradictory nature of these games makes their study with RL agents quite challenging, as demonstrated in \cite{Vassiliades2011}. Moreover, the development of frameworks that enforce cooperation and resist collusion highlights the potential for creating stable cooperative environments in public goods games and other multi-agent settings \cite{Li2019b, Chaudhuri2021}.

In this paper, we build on this literature and explore the potential for tacit coordination among isolated artificial agents in our social dilemma-based MPMG using multiagent learning. 

\section{The Minimum Price Markov Game}
\label{sec:The Minimum Price Markov Game}

In real-world multiagent systems, a common observation is that heterogeneous agents coexist in non-cooperative environments. For example, in heterogeneous markets, stronger agents typically enjoy a competitive advantage, making growth, or even survival, more difficult for weaker ones. This disparity in market power---whether due to factors like market share, technology, or cost structures---creates a stratification of prices, influencing price equilibrium. In the context of first-price public auctions, stronger agents tend to submit the lowest bids, leading to the expectation that they will win more frequently. This assumption is reasonable, as their ability to cut costs is positively correlated with their operational capacity.

In a game-theoretic setting, it is necessary to define a payoff function that applies to both heterogeneous and homogeneous markets without losing generality. The payoff structure should smoothly transition from homogeneous to heterogeneous markets, allowing for a mathematically consistent analysis of strategic behaviors across different numbers of players. Market heterogeneity should be reflected in differences in payoffs, win frequencies, and incentives to avoid price collusion, with these variations proportionate to the gaps in market strength.

Furthermore, in public procurement, there is a tension between the pursuit of personal gain and collective welfare. Firms strive to maximize their individual profits, while the state aims to allocate public funds efficiently and maintain market competition and fairness. From the firms' perspective, individually rational choices made without cooperation can lead to suboptimal outcomes for all participants, as maximal profits are only achievable through cooperative behavior. This conflict between individual and collective rationality is the core characteristic of social dilemmas. The minimum price rule, with its mathematical elegance, does not diminish this inherent conflict; rather, it preserves this tension fully, maintaining the balance between self-interest and group benefit.

It follows that a minimum price-ruled auction game formulation will fall in accordance with the principles of social dilemma games. Namely, it must display non-zero-sum interactions for which players can benefit from cooperating or suffer from a single defection. The dominant strategy (Nash Equilibrium), being collectively irrational, is thus Pareto inefficient, meaning that an efficient coordination scheme leads to the strategy profile associated with the non-Nash Pareto Optimal outcome. In other words, the role of coordination is to transition from the Nash Equilibrium to the profit-maximizing strategy profile.

In addition, the single-stage formulation must capture the core elements of market dynamics, which inherently unfold through repeated interactions. These characteristics must be maintained in its dynamic extension, such as in the iterated Minimum Price Game (MPG) and the Minimum Price Markov Game (MPMG).

In this section, we provide a formulation for the MPMG, that we believe is a reasonable approximation of the typical real-world minimum price-ruled public procurement process. Our formulation is rooted in assumptions that capture the core of this process, creating a conducive environment to observe how artificial learners dynamically shape their strategies in such dilemma. The MPMG is founded on a static normal-form game, the MPG, which we will define first. 

\subsection{The Minimum Price Game: A Single-Stage Formulation}
\label{sec:Static Form} 

The MPG involves $n$ agents, represented by a set $\mathcal{N} = \{1,\dots,n\}$, competing to win a contract valued at $v$ in a first-price procurement auction. 

\begin{assumption}[Single Stage Auction]
The bidding process unfolds in a single round.
\label{assumption:Single Stage Auction}
\end{assumption}

\begin{assumption}[Simultaneous Play]
All agents submit their bids simultaneously.
\label{assumption:Simultaneity of Bids}
\end{assumption}

Although bids are, in practice, collected over a certain period, they remain sealed and confidential. This ensures that all bids are (virtually) submitted at the same time, thus justifying Assumption \ref{assumption:Simultaneity of Bids}. Assumptions \ref{assumption:Single Stage Auction} and \ref{assumption:Simultaneity of Bids} complement each other and allow for a normal-form game formulation.  

\begin{assumption}[Common Value]
The contract value is common for all agents.
\label{assumption:Common Value}
\end{assumption}

Assumption \ref{assumption:Common Value} simplifies the commonly used Bayesian game formulation of first-price auction to reflect the conditions of Bertrand competition \cite{Baye2008}, where agents compete on prices in a transparent market. This assumption is supported by the uniformity of contract requirements and regulatory costs across bidders, market transparency, and the homogeneity in the evaluation of contract values. These factors standardize the perceived value of the contract, thereby aligning all agents' bids based on a common economic assessment rather than individual speculative valuations.

\paragraph{Strategies.} All players are symmetric in their strategy set. Indeed, each agent $i$ chooses a strategy $s_i\in S$ where $S=\{\textit{FP}, \textit{CP}\}$. The strategy \textit{Fair Price} (\textit{FP}) is associated with bidding the fair price $b_i$, and \textit{Collusive Price} (\textit{CP}) with bidding a higher price, defined by $\alpha \cdot b_i$, where $\tau \geq \alpha > 1$, with the upper bound $\tau$ ensuring realistic bids. We denote the strategy profile for a given occurrence as the tuple $(s_1, s_2, \dots, s_n)$. 

\paragraph{Heterogeneity.} Each agent $i$ has a market power defined by the parameter $\beta_i \in (0,1)$, with $\sum_{i \in \mathcal{N}} \beta_i = 1$, indicating that the average market strength, $\mu(\beta)$, is always $\frac{1}{n}$. This market power might represent the size of the agent or its market share, influencing its ability to reduce costs and leverage economies of scale. All agents share the same cost function, which determines their bids
\begin{equation}
    b_i = (1-\beta_i)v \quad \forall i \in \mathcal{N},
\label{eq:Cost Function}
\end{equation}
The notion of strong and weak agents is relative to the distribution of $\beta$, but according to \eqref{eq:Cost Function}, the closer $\beta$ is to 1, the lower the agent can bid, thus the stronger it is. Market heterogeneity is quantified by $\sigma(\beta)$, the standard deviation of the distribution of power parameter values. In a perfectly homogeneous market, $\sigma(\beta) = 0$ and $\beta_i = \frac{1}{n}$ for all $i$.

\begin{remark}
As discussed earlier, we need to model heterogeneous agents within a game formulation that incorporates a social dilemma. Additionally, the payoff structure must reflect a reward function that enables adaptive agents to learn their strategies in repeated settings. Example \ref{ex:strict minimum price rule} illustrates why a strict application of the minimum price rule is unsuitable in this context. The core issue is that a strong player could dominate any auction regardless of its strategy, undermining the coexistence of heterogeneous agents.
\end{remark}

\begin{example}[Hard minimum price rule]\label{ex:strict minimum price rule}
    Let $b_{weak}$ and $b_{strong}$ be the bids in an auction within a heterogeneous duopoly. Naturally, and according to \eqref{eq:Cost Function}, we have $b_{strong} < b_{weak}$. Assuming the payoffs $u$ are determined by a minimum price rule, i.e., $u_{i} = b_{i}$ if $b_{i} = \min(b_{weak}, b_{strong})$ and $u_i = 0$ otherwise, we obtain the following bimatrix game:
    \begin{center}
        \begin{tikzpicture}[scale=2]
            % matrix
            \draw[thick] (0,0) rectangle (2,2);
            \draw[thick] (0,1) -- (2,1);
            \draw[thick] (1,0) -- (1,2);
            \draw (1,2) -- (2,1);
            \draw (1,0) -- (0,1);
            \draw (2,0) -- (1,1);
            \draw (1,1) -- (0,2);
            % nodes
            \node at (1.30,1.25) {$b_{strong}$};
            \node at (0.5,1.75) {$0$};
            \node at (1.35,0.25) {$\alpha b_{strong}$};
            \node at (1.75,0.75) {$0$};
            \node at (0.75,0.75) {$?$};
            \node at (1.75,1.75) {$0$};
            \node at (0.25,0.25) {$?$};
            \node at (0.30,1.25) {$b_{strong}$};
            % strategies
            \node at (-0.3, 1.5) {\textit{FP}};
            \node at (-0.3, 0.5) {\textit{CP}};
            \node at (0.5, 2.2) {\textit{FP}};
            \node at (1.5, 2.2) {\textit{CP}};
            % labels
            \node at (-1,1) {\textbf{strong agent}};
            \node at (1,2.5) {\textbf{weak agent}};
        \end{tikzpicture}
    \end{center}
    We observe that the outcomes of the asymmetric strategy profiles depend on the level of heterogeneity. The strong agent could potentially win regardless of the strategy profile, while the weak agent might struggle to learn a long-term strategy since it may only receive a payoff of $0$ as a reward signal in a repeated setting.
\end{example}

\paragraph{Payoffs.} Let \( u_i(\textit{FP}, k) \) represent the payoff of agent \( i \) when it bids the fair price and \( k \) opponents also bid the fair price, and let \( u_i(\textit{CP}, k) \) denote the payoff of agent \( i \) when it bids the collusive price while \( k \) opponents bid the fair price, where \( k \in [0, n-1] \). The key idea from the minimum price rule is that coordination must be unanimous to achieve collusive profits. Therefore, the individual payoff for playing \textit{CP} when all players cooperate surpasses the individual payoff from universal defection, i.e.,

\begin{equation}\label{eq:cooperation_yield}
u_i(\textit{CP}, k=0) > u_i(\textit{FP}, k=n-1).
\end{equation}

However, defection must always yield a higher payoff than cooperation when at least one opponent defects, leading to

\begin{equation}\label{eq:defection_yield}
u_i(\textit{FP}, k>0) > u_i(\textit{CP}, k>0).    
\end{equation}

In fact, because of the minimum price rule, we need to make sure that the cooperating agents receive payoffs of $0$ when some other players defect. Also, we set a sharing mechanism where defecting agents earn a share based on their market strengths relative to the defecting opponents. Therefore, let \( \Omega \subseteq \mathcal{N} \) denote the set of agents playing \textit{FP}, and let the total market power of the defecting agents be \( \beta_\Omega = \sum_{j \in \Omega} \beta_j \). This means that when agent $i$ plays \textit{FP}, its payoff, 

\[
\frac{\beta_i}{\beta_{\Omega}} \cdot b_i,
\]

depends on $\beta_{\Omega}$. Note that in this case, $\beta_{\Omega}=1$ when $k=n-1$, and $\beta_{\Omega}=\beta_i$ when $k=0$.  When agent $i$ plays \textit{CP}, then its payoff is either $0$ if $k>0$ or 

\[
\alpha \cdot \beta_i \cdot b_i,
\]

if $k=0$. Table \ref{tbl:payoffs} summarizes the individual payoffs for agent $i$ in all scenarios.

\begin{table}[ht!]
\renewcommand{\arraystretch}{1.5}
\caption[MPG payoffs]{MPG Payoffs}
\centering
    \begin{tabular}{c|c|c|c}
        \textbf{Strategy Profile} & \(\mathbf{k = 0}\) & \(\mathbf{n-1 > k > 0}\) & $\mathbf{k=n-1}$\\
        \hline
        \( u_i(\textit{FP}, k) \)          & \( b_i \)          & \( \frac{\beta_i}{\beta_{\Omega}} \cdot b_i \) & $\beta_i \cdot b_i$\\
        \hline
        \( u_i(\textit{CP}, k) \)          & \( \alpha \cdot \beta_i \cdot b_i \) & \( 0 \) & $0$ \\
    \end{tabular}
    \label{tbl:payoffs}
\end{table}

The payoff function presented in Table \ref{tbl:payoffs} respects \eqref{eq:cooperation_yield} and \eqref{eq:defection_yield} since, 

\[
b_i < \min_{j\neq i} \alpha \cdot b_j \quad \forall i,
\]

and 

\[
\frac{\beta_i}{\beta_{\Omega}} \cdot b_i > 0.
\] 

\paragraph{Analysis.}  
The parameter $\beta$ controls the incentive to defect. As we can derive from \eqref{eq:Cost Function} and Table \ref{tbl:payoffs}, and since $\alpha$ is the collusive factor common to all agents, the incentive factor for agent $i$ associated with defecting while every opponent colludes (\textit{FP}, $k=0$) is 
\begin{equation}\label{eq:incentive factor}
\gamma_i = \frac{1}{\beta_i}.
\end{equation}
The pairwise relation in terms of incentive factors between agent $i$ and an opponent $j$ is given by 
\begin{equation}\label{eq:parwise factor}
\gamma_i = \frac{\beta_j}{\beta_i}\gamma_j.
\end{equation}
The incentive to defect is proportional to the market power $\beta$, and therefore weak agents have more incentive to defect than strong ones, as expected in real-life scenarios \cite{Laffont1997a, Tirole2003}. Example \ref{ex:incentive factor} proposes a numerical depiction of the incentive to defect in the MPG. 

\begin{example}[Incentive factor]\label{ex:incentive factor} 
Consider three heterogeneous agents $\{1,2,3\}$, and their associated power parameters $[0.25, 0.25, 0.5]$. For simplicity, we set $v$ to $1$ in equation \eqref{eq:Cost Function}. In the collusive scenario, all agents choose to play \textit{CP}, meaning that $\beta_{\Omega} = 1$, and according to \eqref{eq:Cost Function} and Table \ref{tbl:payoffs}, we have $u_1 = u_2 = 0.19\alpha$ and $u_3 = 0.25\alpha$, with the strongest agent receiving the highest payoff. 

Now consider that agent $1$, a weak agent, chooses to defect instead. In this case, we have $\beta_{\Omega} = 0.25$, and thus, $u_1 = 0.75$ and $u_2 = u_3 = 0$. As predicted by \eqref{eq:incentive factor}, the incentive factor for the weak agent to defect is $4$.

If the strong agent decides to defect while others cooperate, then $\beta_{\Omega} = 0.5$, and thus, $u_1 = u_2 = 0$ and $u_3 = 0.5$. Here, the incentive factor for the strong agent to defect is $2$. As predicted by \eqref{eq:parwise factor}, the incentive factor for the weak agent to defect is twice as high because its market power ($\beta_1 = \beta_2 = 0.25$) is half that of its stronger opponent. 
\end{example}

Our sharing mechanism, using $\beta_{\omega}$, allows for heterogeneous agents to coexist while preserving the dynamic stemming from the minimum price. In fact, the power parameter $\beta$ not only represents the ability to cut costs but also the work capacity. As strong agents are more rewarded than weak ones in symmetrical plays (see Table \ref{tbl:payoffs}), the payoff structure of the MPG emulates the natural turnover in auction participation of heterogeneous agents observed in markets, making the coexistence between weak and strong agents possible. Indeed, a higher payoff mathematically represents the higher frequency of wins by stronger agents in a strict minimum price-ruled process. When the setting simplifies to a homogeneous game, the sharing mechanism in symmetrical plays smooths out a random turnover between agents of equal power, as shown in Example \ref{ex:homogeneous case}.

\begin{example}[Homogeneous special case]\label{ex:homogeneous case}
    Consider a homogeneous duopoly ($\sigma(\beta) = 0$ and $n = 2$). According to \eqref{eq:Cost Function} and Table \ref{tbl:payoffs}, we get the following bimatrix game:
    \begin{center}
        \begin{tikzpicture}[scale=2]
            % matrix
            \draw[thick] (0,0) rectangle (2,2);
            \draw[thick] (0,1) -- (2,1);
            \draw[thick] (1,0) -- (1,2);
            \draw (1,2) -- (2,1);
            \draw (1,0) -- (0,1);
            % nodes
            \node at (0.5,1.5) {$b/2$};
            \node at (1.5,0.5) {$\alpha b/2$};
            \node at (0.25,0.25) {$0$};
            \node at (0.75,0.75) {$b$};
            \node at (1.25,1.25) {$b$};
            \node at (1.75,1.75) {$0$};
            % strategies
            \node at (-0.3, 1.5) {\textit{FP}};
            \node at (-0.3, 0.5) {\textit{CP}};
            \node at (0.5, 2.2) {\textit{FP}};
            \node at (1.5, 2.2) {\textit{CP}};
            % labels
            \node at (-1,1) {\textbf{agent 1}};
            \node at (1,2.5) {\textbf{agent 2}};
        \end{tikzpicture}
    \end{center}
    In symmetrical plays, the payoffs are the same for both agents. This is mathematically equivalent to a random selection mechanism for breaking ties between the lowest bidders under the hard minimum price rule. Since homogeneous players produce the same bid, $b$, a selection mechanism would involve a random process to break ties in symmetrical play. This could be done by either using the uniform distribution, $i^* \sim \mathcal{U}(\mathcal{N})$, or by adding a stochastic perturbation to the bids, $b_i + \epsilon_i$.
\end{example}

\begin{proposition}
The $n$-player homogeneous MPG is a Prisoner's Dilemma. 
\label{proposition:Homogeneous MPG is a PD}
\end{proposition}

\begin{proof}
    Consider a reference player $i\in\mathcal{N}$. Also, consider $\alpha \leq 2$, a realistic upper bound for the collusive factor. The following payoff structure unfolds by definition:

    \begin{itemize}
        \item $T$: Temptation payoff (player $i$ bids the fair price while all opponents bid their collusive price),
        \item $R$: Reward payoff (everybody cooperates),
        \item $P$: Punishment payoff (everybody defects),
        \item $S$: Sucker’s payoff (player $i$ bids the collusive price while at least one opponent defects).
    \end{itemize}

    \begin{enumerate}[i.]
        \item \textit{$P$ is the dominant strategy profile.}
        
            In the homogeneous MPG, $\beta_i = \frac{1}{n}$ for all $i\in\mathcal{N}$, meaning that $b_i = b$ for all $i\in\mathcal{N}$ according to \eqref{eq:Cost Function}. Hence,
            \[
            P = u_i(\textit{FP}, k=n-1) = \frac{b}{n} \quad \text{and} \quad S = u_i(\textit{CP}, k>0) = 0, 
            \]
            and $P > S$. Since 
            \[
            u_i(\textit{FP}, k=0) = b > \frac{\alpha b}{n} = u_i(\textit{CP}, k=0) \quad \forall i\in\mathcal{N}
            \] 
            for any $\alpha < 2$ (since $n \geq 2$), $T > S$, meaning that \textit{FP} (defection) is the best response for every player. The strategy profile $P$ is therefore the unique Nash equilibrium.  
   
        \item \textit{The strategy profile $R$ is Pareto Optimal.}
    
            A strategy profile is Pareto Optimal if no player can improve their outcome without making at least one other player worse off. Since 
            \[
            u_i(\textit{CP}, k=0) = \frac{\alpha b}{n} >  \frac{b}{n} = u_j(\textit{FP}, k=n-1) \quad \forall (i,j) \in \mathcal{N}, 
            \] 
            because $\alpha > 1$, $R > P$, and $(\textit{CP}, k=0)$ is therefore Pareto Optimal. 
    
        \item \textit{The payoff structure satisfies the inequality $T > R > P > S$.}
        
            Since 
            \[
            u_i(\textit{FP}, k=0) = b >  \frac{\alpha b}{n} = u_i(\textit{CP}, k=0) \quad \forall i\in\mathcal{N},
            \]
            because $\frac{\alpha}{n} < 1$, $T > R$. By combining the results of i. and ii., we conclude $T > R > P > S$.
    \end{enumerate}
\end{proof}

\begin{corollary}
The $n$-player heterogeneous MPG is not a Prisoner's Dilemma.
\label{corollary:generalization}
\end{corollary}

\begin{proof}
    In the heterogeneous MPG, the market powers differ among agents, leading to variations in the payoffs. If the level of heterogeneity is high enough, the strongest agent may receive a higher payoff with $R$ (collusion) than with $T$ (being the only defector). In this case,
    \[ 
    u_i(\textit{CP}, k=0) = \alpha \beta_i b_i \quad \forall i\in\mathcal{N}, 
    \]
    if $\alpha \beta_i > 1$, then $u_i(\textit{CP}, k=0) = R > T = u_i(\textit{FP}, k=0)$, and condition iii. in the proof of Proposition \ref{proposition:Homogeneous MPG is a PD} no longer holds.
\end{proof}

\subsection{Markov Game Formulation}
\label{sec:Markov Game Formulation}

In contrast to iterated games, Markov games evolve over multiple stages, with states dependent on previous actions. This dynamic nature better captures the complexity of real-world markets where strategies evolve over time, especially in the context of behavioral emergence, which may heavily rely on self-strengthening trends. The MPMG extends the framework of the MPG by providing the advantage of the Markov property, that relies on a state space that encodes information about opponents' past plays and characteristics. 

\begin{definition}{\textit{Markov game.}} A Markov game, also known as a stochastic game, is a generalization of a MDP to multiple interacting agents contained in the set \( N \). Each player \( i \) has an associated action space \( \mathcal{A}_i \). The game is defined over a common state space \( \mathcal{S} \). The dynamics of the game are captured by a transition function \( T: \left(\bigtimes_{i \in N} \mathcal{A}_i\right) \times \mathcal{S} \rightarrow [0,1] \), which determines the probability of moving from one state to another based on the actions of all players. Each player \( i \) has a reward function \( R_i: \left(\bigtimes_{i \in N} \mathcal{A}_i\right) \times \mathcal{S} \rightarrow \mathbb{R} \), defining the payoff \( r_i \) for player \( i \) based on the current state and actions taken by all players. The objective for each player \( i \) is to maximize the expected sum of discounted rewards, expressed as \( E\left[\sum_{t=0}^\infty \gamma^t r_{i,t}\right] \), where \( \gamma \) is the discount factor that values present rewards over future rewards. A key feature of an MDP, and therefore of a Markov game, is the Markov property, which asserts that future states depend only on the current state and the actions taken by the agents, not on the sequence of events that preceded it, ensuring that past states are irrelevant once the current state is known.
\end{definition}

\paragraph{Action space.} The action space $\mathcal{A} = \{0, 1\}$ is binary, reflecting the two strategies available to players, with bidding the fair price being represented by $0$ and bidding the collusive price with $1$. This representation simplifies the computation of action and joint action frequencies. Let $a_i^t \in \mathcal{A}$ be the action taken by agent $i$ at time $t$, and $f^t \in \mathcal{A}^n$ be the joint action $(a_1, \dots, a_n)$ in auction $t$. The action and joint action frequencies up to iteration $t$ are then respectively given by
\[
\bar{a}_i^t = \frac{1}{t} \sum_{m=1}^{t} a_i^m,
\]
and
\[
\bar{f}^t = \frac{1}{t} \sum_{m=1}^{t} f^m.
\]

\begin{assumption}[Information availability]
All players have access to information about the game status and basic features of their opponents.
\label{assumption:Availability of information}
\end{assumption}

Assumption \ref{assumption:Availability of information} is supported by the availability of auction data, which results from transparency laws. This assumption is crucial for constructing the state space, as it provides the signals necessary for players to make informed decisions during implementation.

\paragraph{State space.} The choice of state space in a Markov game is often at the discretion of the implementer, depending on the specifics of the problem and the goals of the implementation. However, for our purposes, we define the MPMG with core variables that capture essential aspects of the game. Specifically, these sets include all players' information on action frequencies $\mathbf{\bar{a}}^t = \{ \bar{a}_i^t \}_{i \in \mathcal{N}}$, joint action frequencies $\mathbf{\bar{f}}^t = \{ \bar{f}_i^t \}_{i \in \mathcal{N}}$, average rewards $\mathbf{\bar{r}}^t = \{ \bar{r}_i^t \}_{i \in \mathcal{N}}$, and power parameters $\mathbf{\beta} = \{ \beta_i \}_{i \in \mathcal{N}}$. Therefore, the state space at any time $t$ is $\mathcal{S}_t = \{ \mathbf{\bar{a}}^t, \mathbf{\bar{f}}^t, \mathbf{\bar{r}}^t, \mathbf{\beta} \}$, and is of size $2n + |\mathcal{A}|^n$. This minimal configuration for the state space aligns with Assumption \ref{assumption:Availability of information}. Furthermore, the power parameters also allow for the possibility of modeling a dynamic market in which agents could lose or gain market power.

\paragraph{Reward function.} The payoff structure defined for the MPG leads to a reward function in the MPMG, emphasizing the inherent connection between the Markov game and its RL implementation. The reward function, together with the transition function, which defines how the game state evolves in response to actions, fundamentally shapes the dynamics and outcomes of the game. In fact, in the classical RL setting, the normalized reward function for agent $i$ at each iteration $t$ is expressed as
\[
G_i^t = (1-\gamma)\sum_{m=0}^{t-1} \gamma^{m} r_i^{t-m},
\]
where $\gamma$ is the discount factor and $r_i^{t-m}$ is the immediate reward received by agent $i$ at time $t-m$. The MPMG relies on the tension between Nash and Pareto dominant strategies inherent in the classical Prisoner's Dilemma. However, as with the state space, the reward function can be adjusted for specific implementation purposes (e.g., reward shaping) as long as the fundamental dynamics are preserved. Following Assumption \ref{assumption:Common Value}, we can further assume that the contract value $v$ remains equivalent as the iterations unfold, normalizing it to $1$. Therefore, \eqref{eq:Cost Function} becomes
\[
b_i = 1 - \beta_i.
\]
We do not make this assumption formal, as this normalization does not affect the generality of the model but only simplifies the calculations. The immediate rewards for agent $i$ at any time $t$ are given in Table \ref{tbl:payoffs}.

\begin{assumption}[Single Public Contractor]
All public contractors share the same goal.
\label{assumption:Single Public Contractor}
\end{assumption}

In real-world scenarios, the set of public contractors in a given procurement market usually includes many entities. Despite this plurality, it is safe to assume that all public contractors work towards the same overarching goal. Assumption \ref{assumption:Single Public Contractor} allows us to reduce the set of public contractors to a singleton.

\paragraph{Dynamic.} The MPMG unfolds as an episodic game of arbitrary length (number of episodes), with each episode consisting of a single-step auction (Assumptions \ref{assumption:Single Stage Auction} and \ref{assumption:Simultaneity of Bids}). Agents compete against each other in an environment representing the same entity in each episode (Assumption \ref{assumption:Single Public Contractor}). Due to the Markov property, memory is encapsulated in the observation space $\mathcal{S}$ at any episode $t$. Agents use this observation space to learn their strategies, effectively playing a Markov game where the reward function, $R(a_1, \dots, a_n)$, is deterministic and based on the joint action. The transition probability function is given by $P(s'|s, a)$, where $s'$ is the new state and $s \in \mathcal{S}$. When the state space is not used for decision-making, the agents are effectively playing an iterated MPG.

\section{Computer Experiments}\label{sec:ComputerExperiments}
Cooperation in MARL settings is well studied. However, this algorithmic cooperation mostly relies on explicit designs—such as shared joint action spaces \cite{Zhou2023}, algorithmic communication \cite{Sukhbaatar2016}, reputation-based models \cite{Zhang2021}—which are not suitable for our context, as these approaches emulate scenarios where agents deliberately bias their actions towards cooperation.

In this study, we evaluate the inherent properties of the MPMG, requiring unbiased players. If collusion occurs, it would arise solely from the game’s intrinsic structure. Therefore, we employed basic yet established state-of-the-art MARL algorithms to define artificial agents that are symmetric in their decision-making processes, do not communicate, and do not hold beliefs about their opponents. In other words, we aim to create isolated and equally rational agents.

We conducted a total of twenty-four experiments, combining six agent types with four MPMG configurations (2-player homogeneous, 5-player homogeneous, 2-player heterogeneous, and 5-player heterogeneous MPMG). As discussed in Section \ref{sec:The Minimum Price Markov Game}, \(0 \leq \sigma(\beta) \leq 1\) represents the standard deviation of the power parameters (\(\beta\)) distribution, with \(\sigma(\beta)=0\) representing the homogeneous case. For experiments in the heterogeneous setting, we set \(\sigma(\beta) = 0.5\). Additionally, we set \(\alpha = 1.3\), meaning that the collusive bids are 30\% higher than the fair bids. Each agent was trained on 100 auctions over 100 replications for statistical robustness.

\subsection{Agent Representation} \label{sec:AgentRepresentation}
We based our agents on three foundational approaches of MARL, that is, Multi-Armed Bandits (MAB), Deep Q-learning, and Actor-Critic Policy Gradient. Implementation details are available in Appendix \ref{app:ImplementationDetails}. 

\begin{definition}{\textit{Multiagent Reinforcement Learning.}} MARL extends single-agent RL techniques to settings with multiple agents learning concurrently in a shared environment. In RL, an agent interacts with its environment to maximize cumulative rewards, where key components include agent, environment, state, action, reward, policy, and value function. The agent, acting as the decision-maker, responds to the environment, which consists of external factors and situations represented as states. Actions are possible decisions the agent can take, evaluated by rewards as feedback on their success. The policy is the agent’s strategy for selecting actions based on the current state, while the value function estimates future rewards to determine the desirability of states or actions. In MARL, multiple agents simultaneously learn and update their policies, causing the environment to appear non-stationary from any single agent's perspective. The dynamics lead to more complex interactions compared to single-agent RL, requiring robust techniques to address the challenges of Markov games and multiagent systems in general.
\end{definition}

\paragraph{The MPG as a bandit problem.} MAB algorithms \cite{Lattimore2020} represent a simplified form of RL or MARL, focusing on the exploration-exploitation trade-off without considering future states (\textit{stateless}), that is, without a model of the environment. MAB methods solve the bandit problem, where the decision-maker must sequentially choose among different actions with unknown reward distributions in order to maximize the total reward over time. Each action represents a distinct stochastic process with its own, a priori unknown, probability distribution of rewards. The goal is to minimize \textit{regret}, which is the difference between the total reward the decision-maker could have earned by always choosing the best action and the actual reward obtained by balancing exploration and exploitation. MAB agents actually play an iterated MPG, as they do not utilize the observation space. We deployed $\epsilon$-greedy agents \cite{Vermorel2005} that balance exploitation and exploration by selecting a random action with a probability $\epsilon$ decreasing over time, Bayesian Thompson Sampling (TS) agents \cite{Gupta2011} that maintain a probability distribution for the expected reward of each action, and Upper Confidence Bound (UCB) agents \cite{Auer2002} that select the action with the highest upper confidence bound of the estimated return.

\paragraph{Deep Q-learning.} Methods based on Q-learning \cite{Sutton2018} maintain state-action pair values, named Q-values, based on the state-action pair expected reward. Decision making is then conducted following the highest Q-values associated with the state-action pair in question. In our experiments, we used the Double Dueling Deep Q-network (D3QN) algorithm \cite{Hasselt2016}. In the D3QN, a multi-layer perceptron is used to approximate the Q-values, and the Double Q-learning approach is employed to reduce overestimation bias by maintaining two sets of Q-value estimates. The Dueling architecture, on the other hand, decomposes the Q-value into two separate streams representing the state-value function and the advantage function, which helps the model learn better representations of the value of being in a given state, independent of the specific actions.

\paragraph{Actor-critic policy gradient.} Unlike Q-learning, policy gradient methods focus on optimizing a parameterized policy $\pi_{\theta}(a|s)$ directly by gradient ascent on expected rewards \cite{Sutton2018}. In the actor-critic framework, two function approximators are employed. The \textit{actor}, which is responsible for selecting actions based on the current policy, and the \textit{critic}, which evaluates the action-value function to provide feedback on the quality of the chosen actions. The actor updates the policy in the direction suggested by the critic, while the critic learns to estimate the expected return, enabling a stable and efficient training process by reducing the variance of the policy gradient estimates. Here, we used Multi-Agent Proximal Policy Optimization (MAPPO) agents \cite{Schulman2017} that extend the Proximal Policy Optimization (PPO), a popular policy gradient method, to multi-agent settings. The core idea of PPO is to strike a balance between policy improvement and preventing drastic updates to the policy, which could harm performance. To achieve this, PPO uses a modified objective function that constrains how much the new policy is allowed to diverge from the old policy during an update.

\subsection{Results} \label{sec:Results}

Before presenting the results of the experiments, it is essential to outline several important considerations for evaluating artificial agents in the multi-agent context. Indeed, the interpretation of the results depends on both individual and collective convergence. By individual convergence, we mean the convergence of the algorithms themselves. For D3QN and MAPPO, this is assessed by examining training loss function values over time. In the case of MAPPO, both the actor and critic networks must converge. In our experiment, they all converged, as can be seen in Figure \ref{fig:loss_functions}. Additionally, D3QN and MAPPO offer internal metrics indicating their policies. For MAPPO, we can directly observe the evolution of the policies regarding actions \textit{FP} and \textit{CP}, while D3QN uses Q-values for each action. The evolution of MAPPO policy values and D3QN Q-values during training is shown in Figure \ref{fig:policies}. In the case of the stateless bandit algorithms ($\epsilon$-greedy, TS, and UCB), cumulative regrets are generally used to establish learning in the single-agent context, though they are less meaningful in Markov games. Nevertheless, cumulative regrets are presented in Figure \ref{fig:cumulative_regrets}. Figures \ref{fig:loss_functions}, \ref{fig:policies}, and \ref{fig:cumulative_regrets} are provided in Appendix \ref{app:AdditionalFigures}.

However, the joint action frequencies, or the frequencies at which each strategy profile is played, are metrics common to all experiments, and they provide a clear picture of collective convergence. Therefore, they offer a reliable way to evaluate and compare our MARL agents in terms of behavioral trends. In the single-stage MPMG (MPG) or in similar static games, three distinguishable collective outcomes emerge: agents play the Nash profile, the Pareto profile, or they diverge in their choices, meaning that only a subgroup of players share a nonzero return associated with the Nash play. In dynamic settings, outcomes are based on convergence toward a particular outcome, with agents represented by stochastic processes typically converging toward mixed strategies. The challenge, then, is fitting a discrete interpretation (do they cooperate or not?) onto a continuous spectrum of behaviors.

The evolution of the average joint action frequencies over training episodes is displayed in Figure \ref{fig:joint_action_frequencies} in Appendix \ref{app:AdditionalFigures}. The average joint action frequencies from the last training episode for all experiments and outcomes are displayed in Table \ref{tbl:joint_action_frequencies}, showcasing the impact of atomicity and heterogeneity on the agents' ability to coordinate, as well as the different behavioral trends emerging from each configuration of agents and market parameters.

\begin{table}[h!]
\footnotesize
\centering
\renewcommand{\arraystretch}{1.3}
\caption{\footnotesize Average joint action frequencies across repeats at the last training episode for each outcome in each experiment. Bold values indicate the highest average joint action frequency among the three outcomes.}
\begin{tabular}{l|ccc|ccc|ccc|ccc}
    \toprule
    \multicolumn{1}{c|}{} & \multicolumn{3}{c}{$n=2$, $\sigma(\beta)=0.0$} & \multicolumn{3}{c}{$n=2$, $\sigma(\beta)=0.5$} & \multicolumn{3}{c}{$n=5$, $\sigma(\beta)=0.0$} & \multicolumn{3}{c}{$n=5$, $\sigma(\beta)=0.5$} \\
    \cmidrule(lr){2-4} \cmidrule(lr){5-7} \cmidrule(lr){8-10} \cmidrule(lr){11-13}
    & \textit{FP} & \textit{CP} & Other & \textit{FP} & \textit{CP} & Other & \textit{FP} & \textit{CP} & Other & \textit{FP} & \textit{CP} & Other \\
    \midrule
    D3QN & 0.37 & 0.16 & \textbf{0.47} & 0.37 & 0.17 & \textbf{0.47} & 0.07 & 0.02 & \textbf{0.91} & 0.07 & 0.02 & \textbf{0.91} \\
    TS & \textbf{0.72} & 0.04 & 0.24 & \textbf{0.69} & 0.04 & 0.27 & \textbf{0.41} & 0.01 & 0.59 & 0.32 & 0.01 & \textbf{0.67} \\
    $\epsilon$-greedy & \textbf{0.70} & 0.03 & 0.27 & \textbf{0.70} & 0.03 & 0.27 & 0.35 & 0.01 & \textbf{0.64} & 0.35 & 0.01 & \textbf{0.64} \\
    MAPPO & \textbf{0.52} & 0.10 & 0.38 & \textbf{0.51} & 0.11 & 0.39 & 0.24 & 0.01 & \textbf{0.75} & 0.21 & 0.01 & \textbf{0.78} \\
    UCB & 0.40 & \textbf{0.58} & 0.02 & 0.43 & \textbf{0.55} & 0.02 & 0.34 & \textbf{0.43} & 0.23 & 0.22 & 0.08 & \textbf{0.70} \\
    \bottomrule
\end{tabular}
\label{tbl:joint_action_frequencies}
\end{table}

In addition, Figure \ref{fig:heatmap} ranks each agent-MPMG configuration pair in terms of collusive potential by using scores based on the data in Table \ref{tbl:joint_action_frequencies}. More specifically, each value in Figure \ref{fig:heatmap} represents the difference between the average joint action frequency of the last training episode of the Nash and Pareto outcomes. The resulting value is then scaled from $0$ (most Nash) to $1$ (most Pareto) using the \textit{minmax} formula.

\begin{figure}[!ht]
    \centering
    \caption{\small One-dimensional heat-map highlighting the relative difference among agents and MPMG configurations in terms of their collusive potential. From the most Nash (left) to the most Pareto (right) configuration.}
    \includegraphics[width=\textwidth]{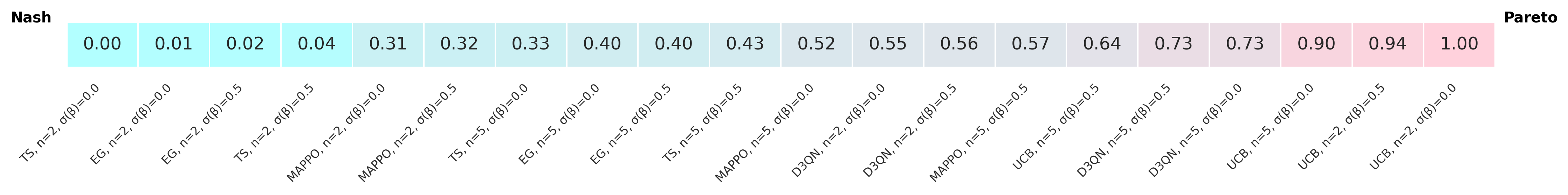}
    \label{fig:heatmap}
\end{figure}

\subsection{Analysis} \label{sec:conclusion}
Our results, summarized in Table \ref{tbl:joint_action_frequencies} and Figure \ref{fig:heatmap}, highlight the varying effectiveness of different MARL approaches in the MPMG. We can see that, in the $5$-player heterogeneous MPMG, and in contrast to the results of the $2$-player homogeneous MPMG, agents always converge towards a suboptimal outcome far from the Pareto profile, but closer to the Nash Equilibrium. This shows the impact of atomicity and heterogeneity on the capacity of agents to coordinate themselves, let alone cooperate in collusion.

Also, UCB agents seem to favor the Pareto profile, consistently across different configurations. In fact, UCB is the only algorithm that not only fosters higher rates of convergence towards the Pareto profile but also shows robustness against atomicity and heterogeneity, at least when they are increased individually. This is clear in Figure \ref{fig:heatmap}, where we can see that the three highest scores belong to UCB agents. However, we can see in Table \ref{tbl:joint_action_frequencies} that in the $5$-player heterogeneous MPMG, UCB agents behave similarly to their counterparts.

Moreover, it seems there is no correlation between algorithmic sophistication and collusive potential. Stateless bandits can either converge towards an approximate Nash Equilibrium (TS, $\epsilon$-greedy) or favor the Pareto outcome (UCB), while methods using more information either favor the Nash Equilibrium (MAPPO) or other suboptimal outcomes (D3QN).

\section{Discussion}\label{sec:Discussion}

In this section, we discuss, interpret, and put our results into perspective.

\subsection{Interpretation}
Three key insights emerge from the computer experiments.

First, the MPMG is a reliable approximation of market dynamics, as it exhibits properties consistent with Jean Tirole's seminal work on market collusion potential \cite{Tirole2003}. Specifically, as shown in Section \ref{sec:Static Form}, in the MPMG, the incentive to defect while opponents cooperate is asymmetrical with respect to the power parameter $\beta$, thereby undermining the intrinsic value of the Pareto strategy. This effect strengthens as the number of agents increases, further reducing the likelihood of successful collusion.

Second, tacit coordination relies heavily on self-reinforcing trends. Indeed, UCB displays counter-intuitive behaviors for a no-regret algorithm. In such settings, the regret of not defecting is higher than the regret associated with not cooperating (see Table \ref{tbl:payoffs}). Additionally, other no-regret algorithms, TS and $\epsilon$-greedy, favor the Nash Equilibrium. The particularity of UCB, at least in the original version we implemented \cite{Auer2002}, is that it is a deterministic method because of the way exploration is handled. Therefore, when fully symmetric UCB agents populate the game, coordination towards the Pareto outcome builds upon self-reinforcing beliefs as every agent acts the same. As we can see from Figure \ref{fig:ucb}, the confidence intervals for the average (across repeats) joint action frequencies associated with defection and cooperation of the UCB agents are null. In contrast, Figure \ref{fig:mappo} shows non-null confidence intervals for MAPPO agents, and Figure \ref{fig:joint_action_frequencies} in Appendix \ref{app:AdditionalFigures} shows that only UCB-related graphs can have null confidence intervals. Furthermore, it has been shown that UCB agents are the RL agents most prone to cooperate in the Iterated Prisoner's Dilemma, with cooperation rates similar to ours \cite{Lin2022}.

\begin{figure}[!ht]
    \centering
    \caption{\small Average joint action frequencies of UCB (a) and MAPPO (b) agents in the $2$-player homogeneous MPMG.}
    \scalebox{0.7}{ % Scale factor; adjust as needed
    \begin{minipage}{\textwidth}  % Use minipage to scale everything together
        \centering
        % First subfigure
        \begin{subfigure}[b]{0.45\textwidth}  % Adjust width as needed
            \centering
            \includegraphics[width=\textwidth]{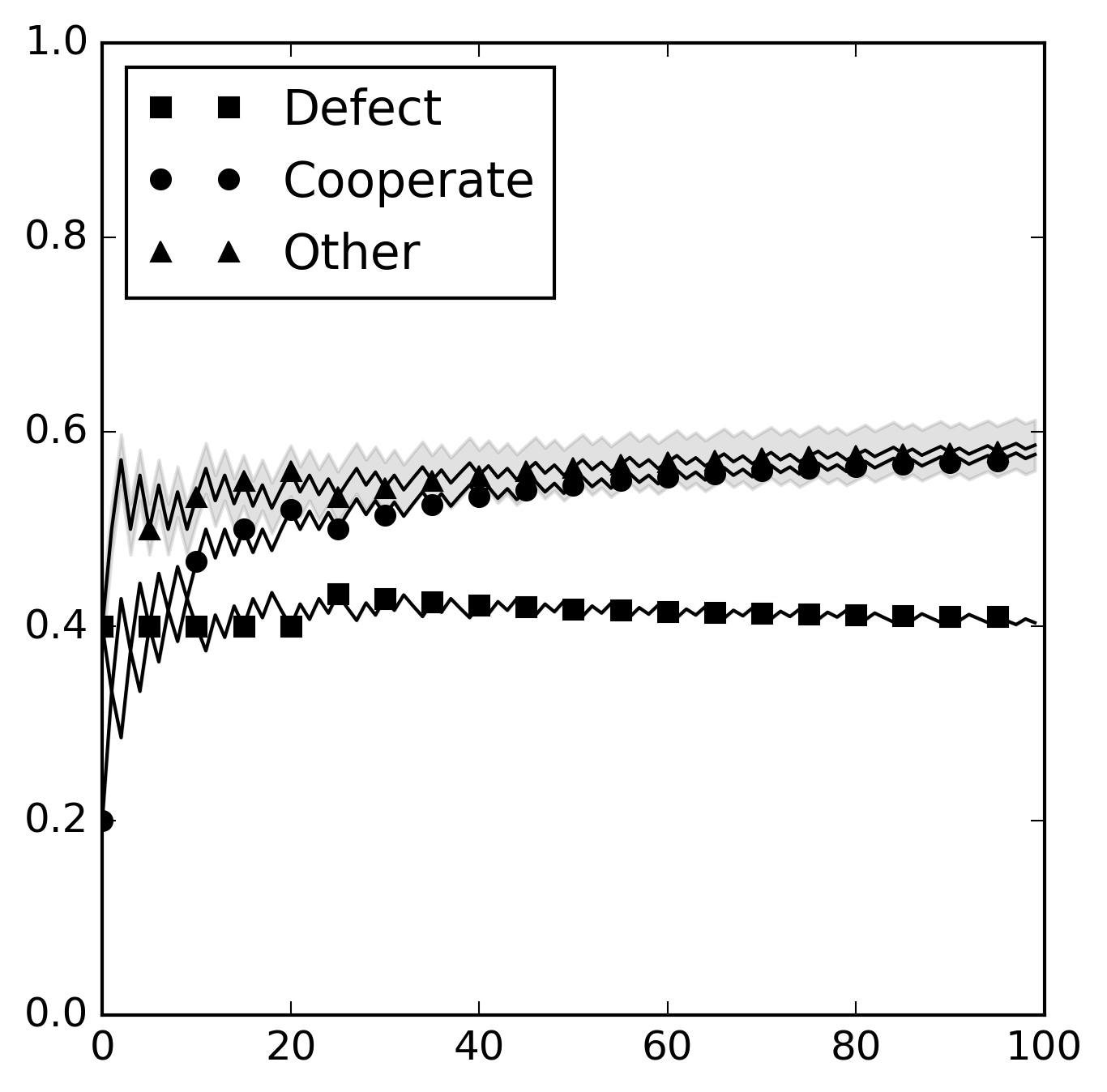}
            \caption{\small UCB}
            \label{fig:ucb}
        \end{subfigure}
        \hfill
        % Second subfigure
        \begin{subfigure}[b]{0.45\textwidth}  % Adjust width as needed
            \centering
            \includegraphics[width=\textwidth]{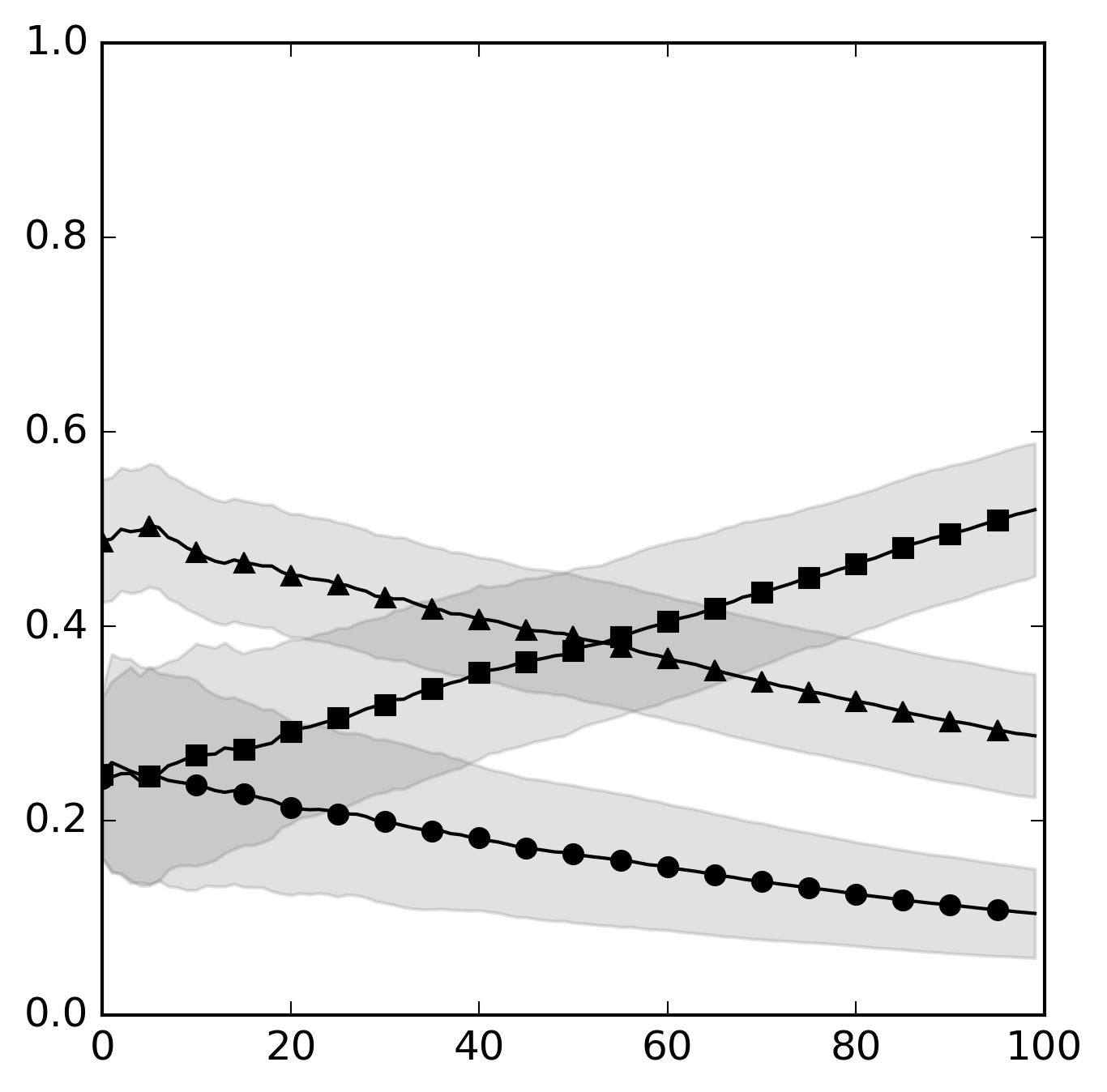}
            \caption{\small MAPPO}
            \label{fig:mappo}
        \end{subfigure}    
        
        \label{fig:ucb_mappo_ci}
    \end{minipage}
    }
\end{figure}

Finally, more information does not lead to more cooperation. There appears to be no positive correlation between the sophistication of our AI agents and the emergence of tacit coordination, a conclusion we share with other authors \cite{Lin2022, Vassiliades2011, Marty2024}. In fact, in the Prisoner's Dilemma, no amount of information can a priori defuse the tension created by the existence of the Nash Equilibrium, as long as agents do not engage in self-reinforcing trust over cooperation.

We can therefore conclude that, provided the MPMG reasonably approximates minimum price-ruled auctions, and although collusion can occur in the most naive settings, the minimum price rule is resilient to non-engineered tacit coordination among rational actors, as is any rule preserving the structure of a social dilemma.

\subsection{Practical Considerations}
This study focuses on evaluating the behavior of artificial agents without explicitly quantifying their collusive potential in terms of collusive gains. Although the Pareto outcome provides higher returns compared to the Nash outcome, agents cannot strictly maximize their returns, as this would involve behavior divergent from typical learning processes. Moreover, even during evaluation—when agents no longer explore and strictly follow their policies—they often converge to mixed strategies rather than a pure Nash or Pareto strategy. Additionally, unlike bandits, D3QN and MAPPO map a policy to state observation, meaning that even when they learn the pure Pareto profile, their collusive gains are highly context-dependent. As a result, their returns in evaluation remain stochastic, and assessing them in terms of collusive gains may be both challenging and irrelevant.

It is also important to note that, in the context of the MPMG, other (non-Nash) suboptimal outcomes can be viewed as rotation schemes, that is, turnovers in the winning set of players. This kind of strategy emerges from mixed behavior and can lead to greater returns than a consistent Nash play.

\subsection{Algorithmic Pricing and Legal Perspective}
Our results align with specialized studies that reveal the challenge of achieving algorithmic cooperation, whether tacit or explicit \cite{Miklos-Thal2019}. However, they do not contradict the legal literature, where scholars often assert that algorithmic collusion is relatively straightforward to achieve \cite{Ezrachi2016a, Ezrachi2016b, Li2023, Green2015}. This study offers a framework and benchmark for tacit collusion and helps explain the contradictory reports in the literature by shedding light on the opaque nature of algorithmic tacit coordination.

Nonetheless, it is essential to add nuance for a comprehensive understanding of this complex issue. While legal concerns stem from tangible evidence found in digital open markets and auctions (see Section \ref{sec:Introduction}), algorithmic implicit coordination has been shown to be impeded in the specific context of unbiased social dilemmas, as discussed in Section \ref{sec:Related Work}. Additionally, the potential for algorithmic and tacit collusion in transparent markets (including public markets) will always exist. Complex cooperative behaviors can be engineered in virtually any large-scale agent system \cite{Gupta2017}, even if such implementations remain challenging. For instance, cooperation in competitive settings can be elusive even when explicit communication is permitted \cite{Harrington2016, Fonseca2012}. What is certain, at least from our perspective, is that further research is required to establish the precise theoretical and algorithmic conditions, as well as the practical considerations, regarding algorithmic collusion in public

\subsection{Future Work}
\paragraph{Extending the MPMG.} Future research should aim to extend the MPMG in both quantitative and qualitative dimensions. Quantitatively, this could involve creating environments that more closely mimic actual simulations, using approaches like empirical games \cite{Greenwald2024}. Qualitatively, it could involve relaxing certain assumptions, such as the common value assumption (Assumption \ref{assumption:Common Value}), particularly in highly heterogeneous market scenarios. Ultimately, developing a realistic, full-scale simulation of such a system would also necessitate a continuous formulation for the action space and a modular environment allowing for the implementation of various supplier selection models, thus generalizing the MPMG to a general auction Markov game.

Additionally, the public contractor could be modeled as an active participant in the environment, represented by an adaptive learning agent. This scenario has been explored by \cite{Brero2022}, where reinforcement learning was used in the context of Stackelberg games as an algorithmic defense against algorithmic pricing for e-commerce platforms. Practically, these perspectives could shape the future of more robust E-procurement designs, and theoretically, they offer a computer-based experimental framework that complements laboratory-based behavioral studies.

\paragraph{AI rationality and cyber cartels.} In an AI-governed world, the concept of rationality takes on new dimensions, particularly in economic interactions and market behavior. The experiments in this study operate under a set of restrictions to maintain an environment unbiased towards cooperation. However, as mentioned in Section \ref{sec:ComputerExperiments}, market players could align on algorithmic pricing using centralized learning and execution (i.e., a shared joint action space). This raises the question: could tacit collusion, emerging from a deliberately biased algorithmic structure, be easily detectable?

The potential existence of cyber cartels introduces new research opportunities into how agents might evade detection and punishment by law enforcers through planned algorithmic collusion. Such a consortium could gain advantages in both efficiency and stealth. Explicit coding for collusion would typically require real-life coordination among market players, potentially leaving traces and being susceptible to algorithmic auditing by legal authorities. However, implicit algorithmic coordination can either be programmed or accidental, creating a grey zone that complicates regulatory efforts. It is also worth noting that both digitized and non-digitized markets can be affected, as the only requirement for algorithmic pricing is the availability of data.

\clearpage

\addcontentsline{toc}{section}{References}
\bibliographystyle{utphys}
\bibliography{bib}

\clearpage

\appendix
\section{Additional Figures}\label{app:AdditionalFigures}

\begin{figure}[!ht]
    \caption{\small Average (across repeats) joint action frequencies over training episodes. Rows (up to down): UCB, $\epsilon$-greedy, TS, D3QN, MAPPO. Columns (left to right): $(n=2, \sigma(\beta)=0.0)$, $(n=2, \sigma(\beta)=0.5)$, $(n=5, \sigma(\beta)=0.0)$, $(n=5, \sigma(\beta)=0.5)$.}
    \centering
    \begin{minipage}[b]{0.23\textwidth}
    \centering
        \includegraphics[width=\textwidth]{png_files/joint_action_frequencies/joint_action_frequencies_ucb_2_0.0_plot.png}
    \end{minipage}
    \begin{minipage}[b]{0.23\textwidth}
        \centering
        \includegraphics[width=\textwidth]{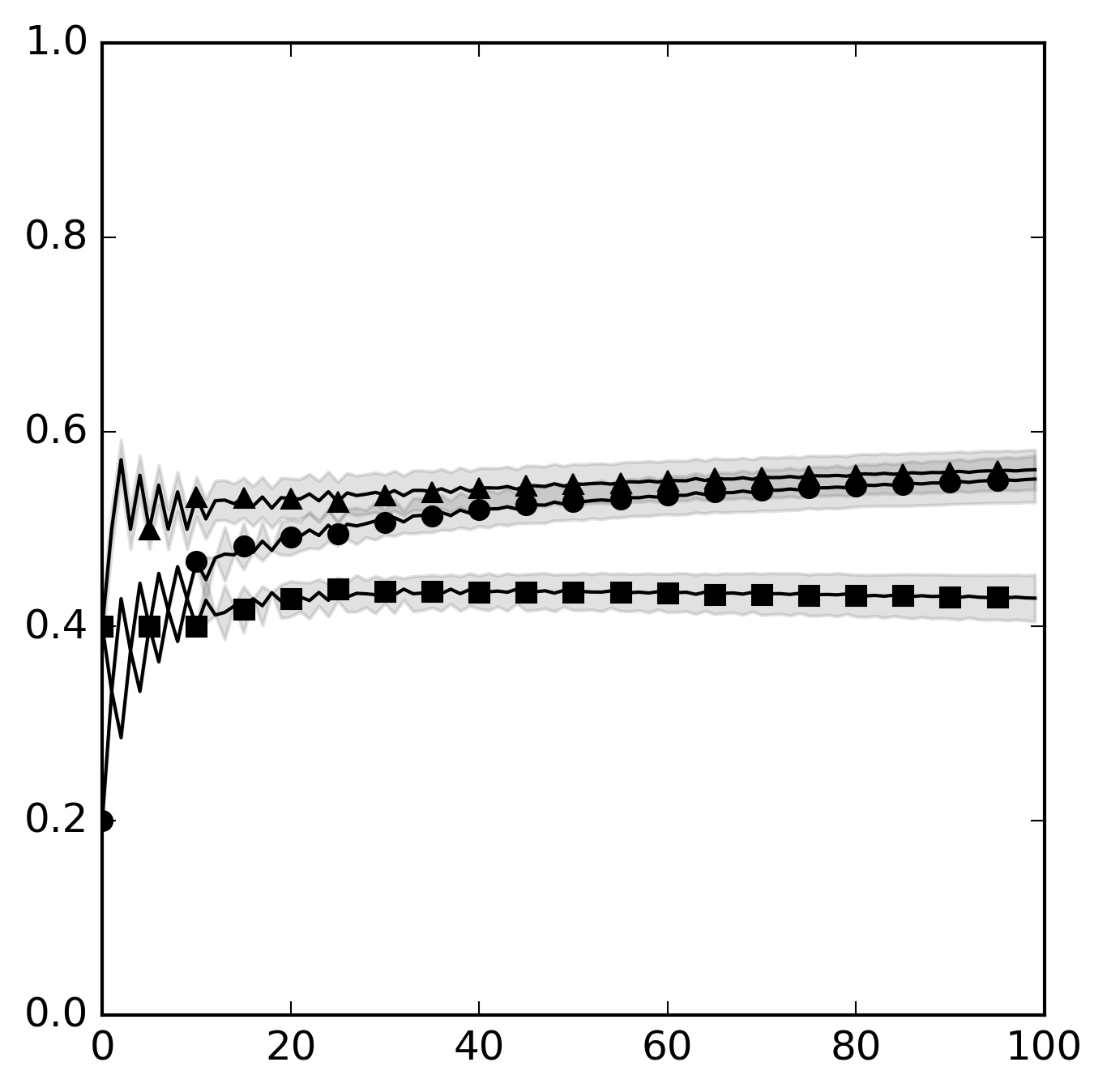}
    \end{minipage}
    \begin{minipage}[b]{0.23\textwidth}
        \centering
        \includegraphics[width=\textwidth]{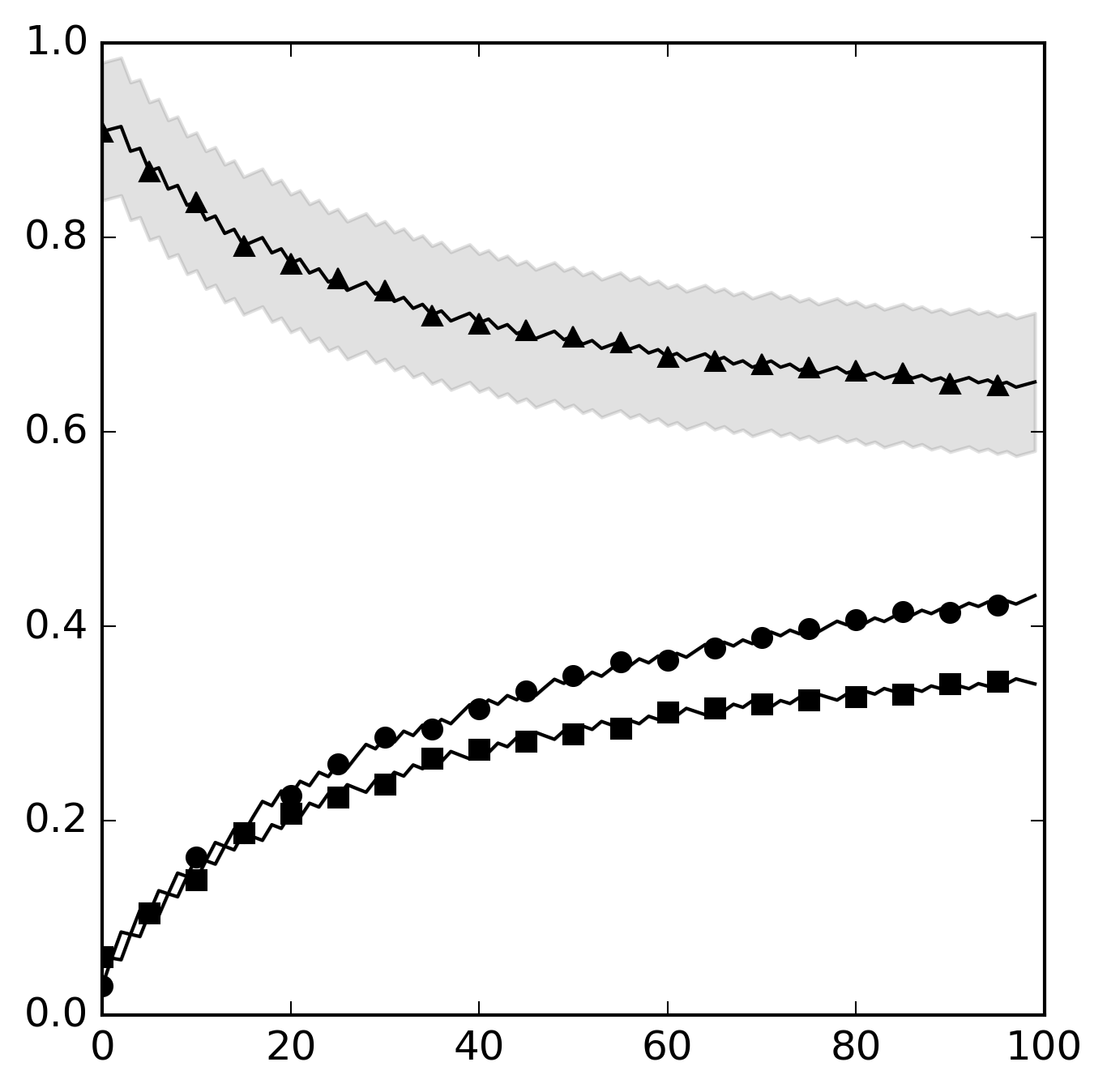}
    \end{minipage}
    \begin{minipage}[b]{0.23\textwidth}
        \centering
        \includegraphics[width=\textwidth]{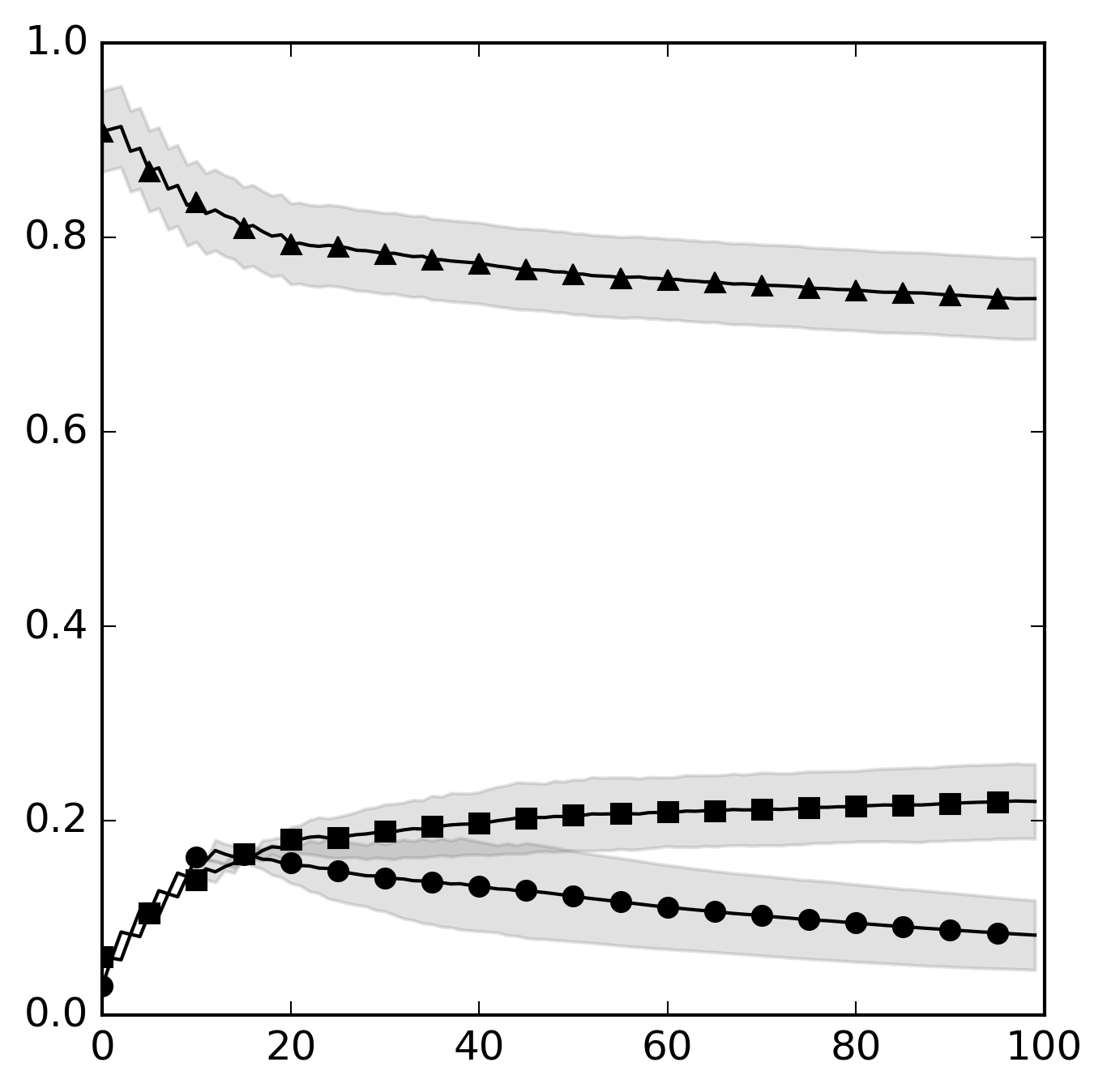}
    \end{minipage}
    \vfill
        \begin{minipage}[b]{0.23\textwidth}
    \centering
        \includegraphics[width=\textwidth]{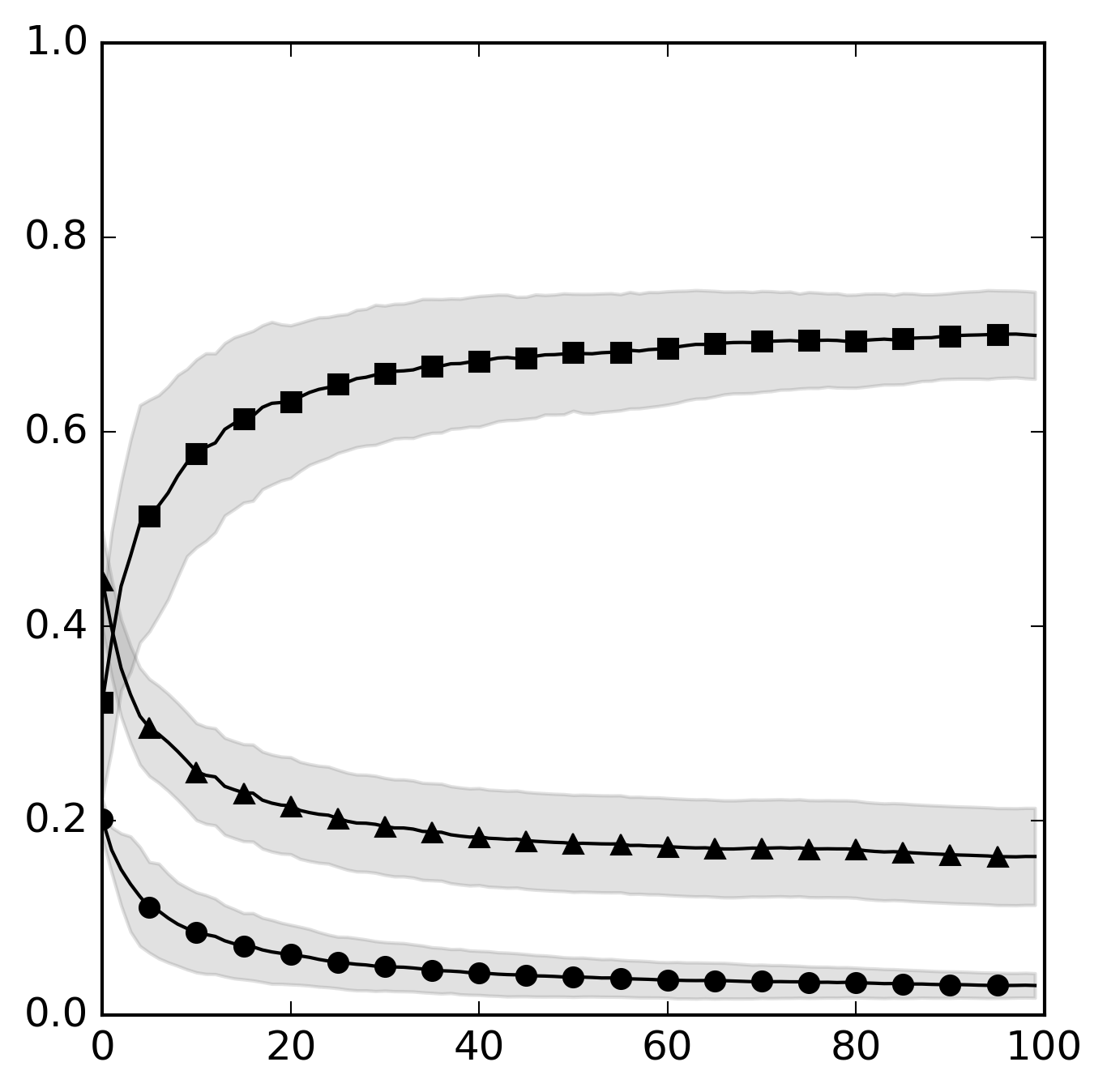}
    \end{minipage}
    \begin{minipage}[b]{0.23\textwidth}
        \centering
        \includegraphics[width=\textwidth]{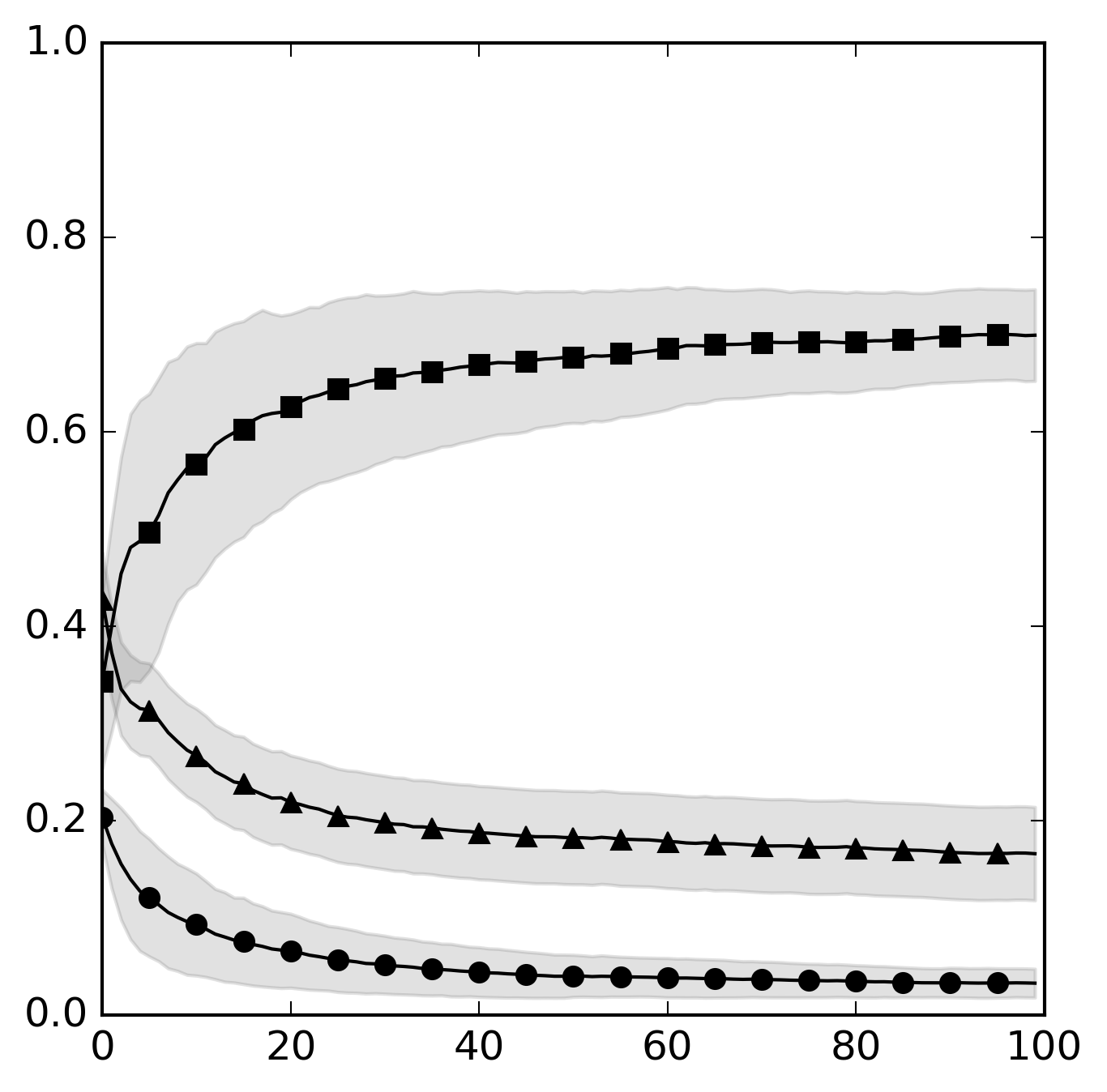}
    \end{minipage}
    \begin{minipage}[b]{0.23\textwidth}
        \centering
        \includegraphics[width=\textwidth]{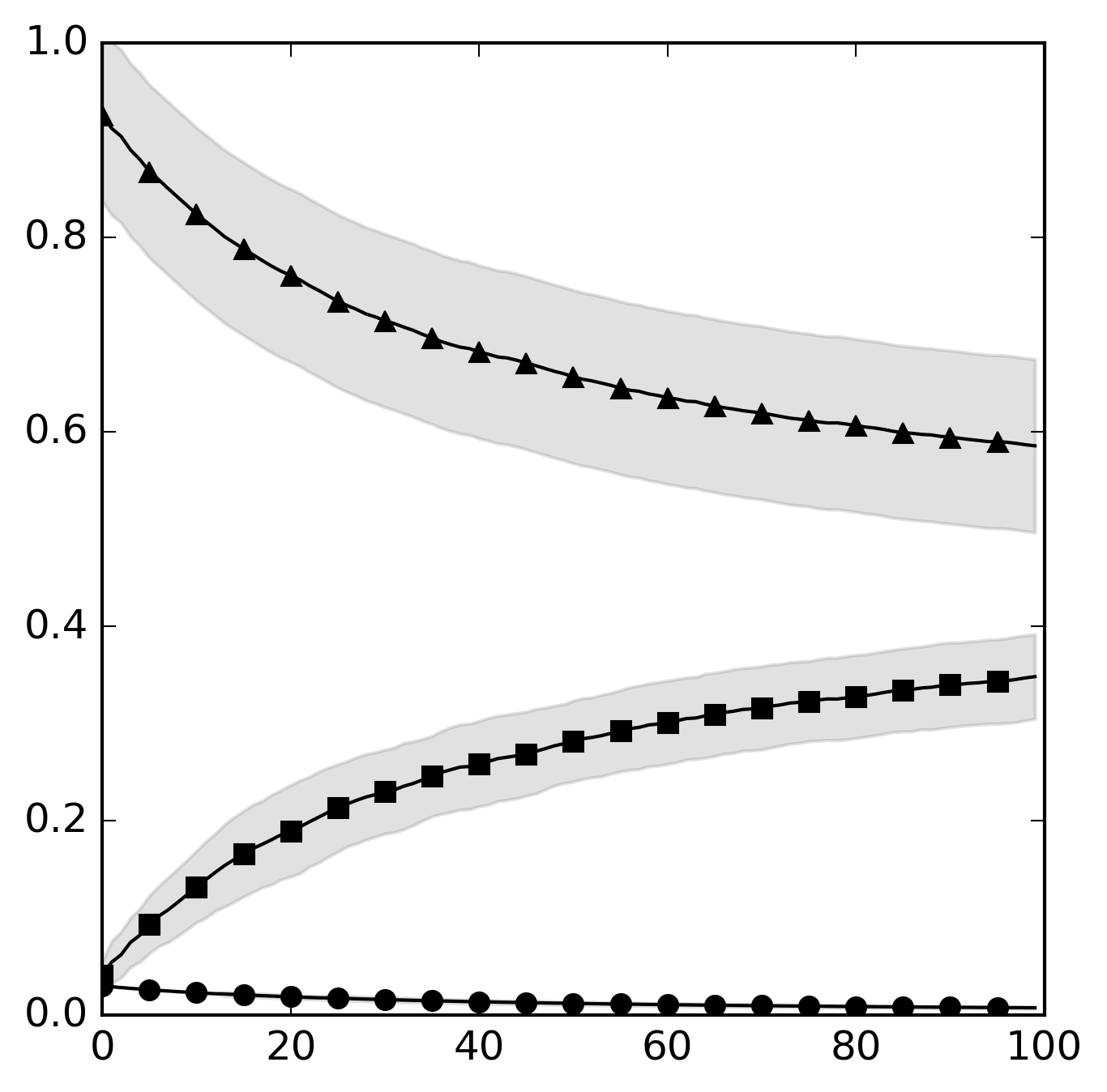}
    \end{minipage}
    \begin{minipage}[b]{0.23\textwidth}
        \centering
        \includegraphics[width=\textwidth]{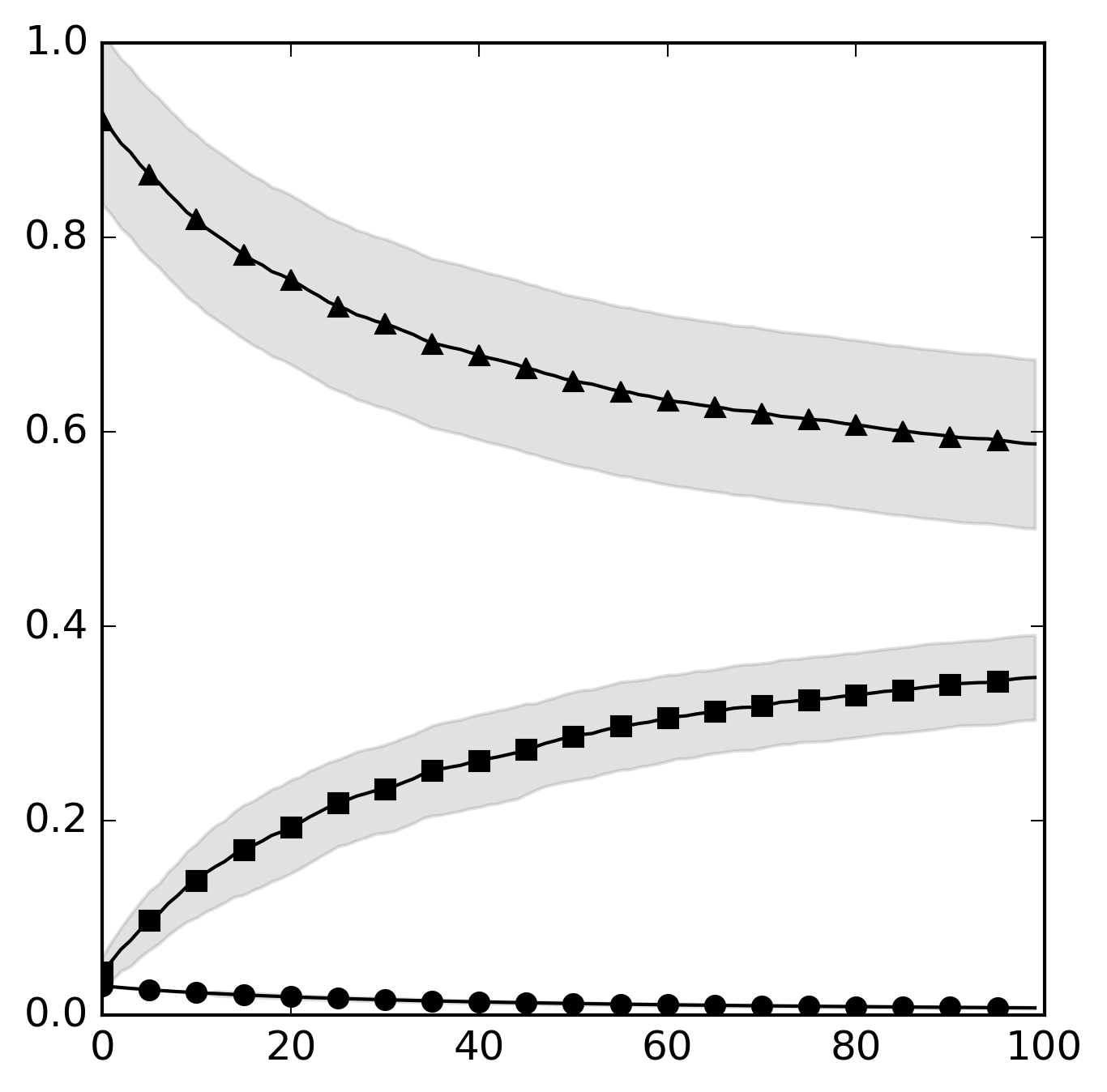}
    \end{minipage}
    \vfill
        \begin{minipage}[b]{0.23\textwidth}
    \centering
        \includegraphics[width=\textwidth]{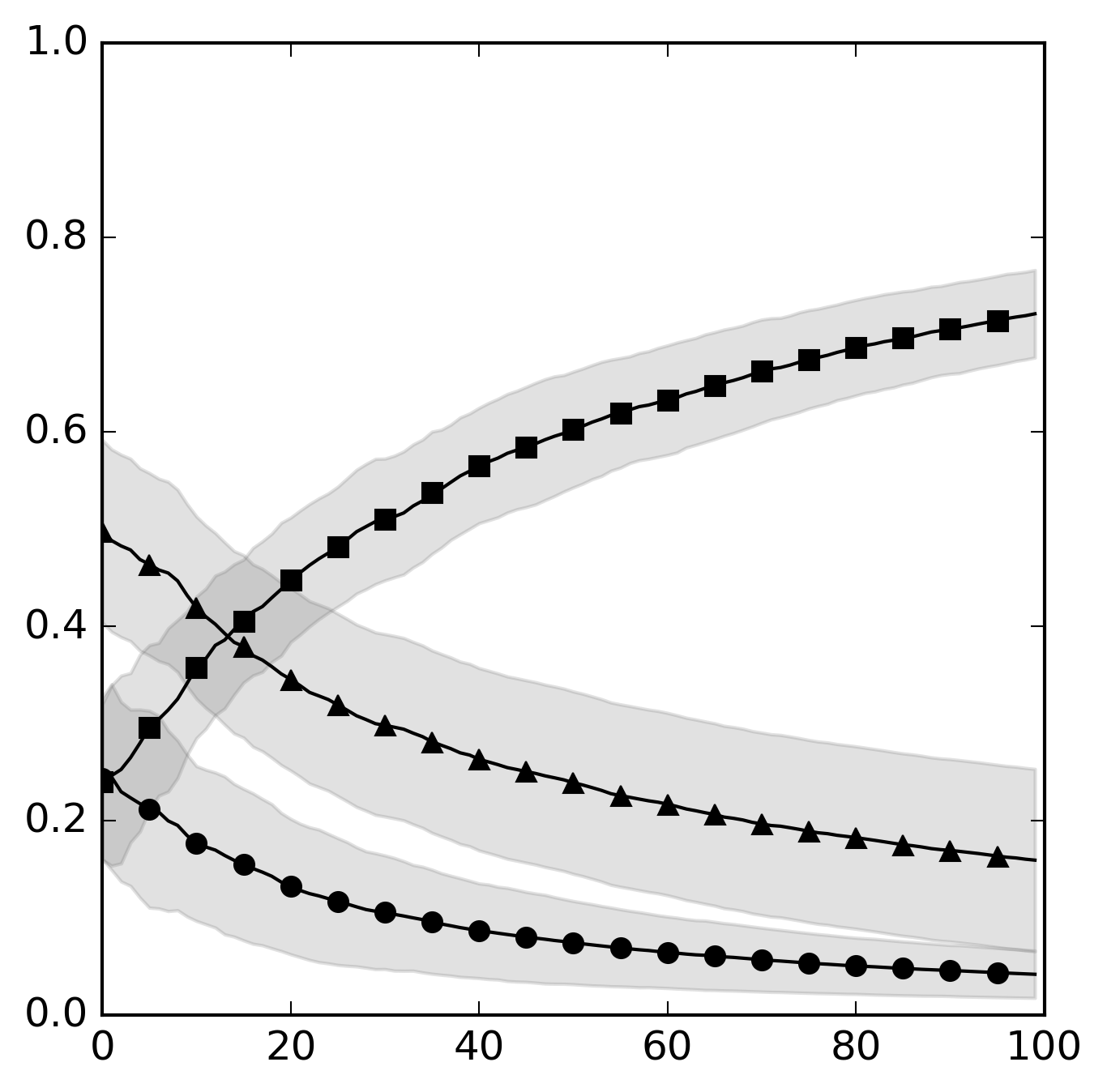}
    \end{minipage}
    \begin{minipage}[b]{0.23\textwidth}
        \centering
        \includegraphics[width=\textwidth]{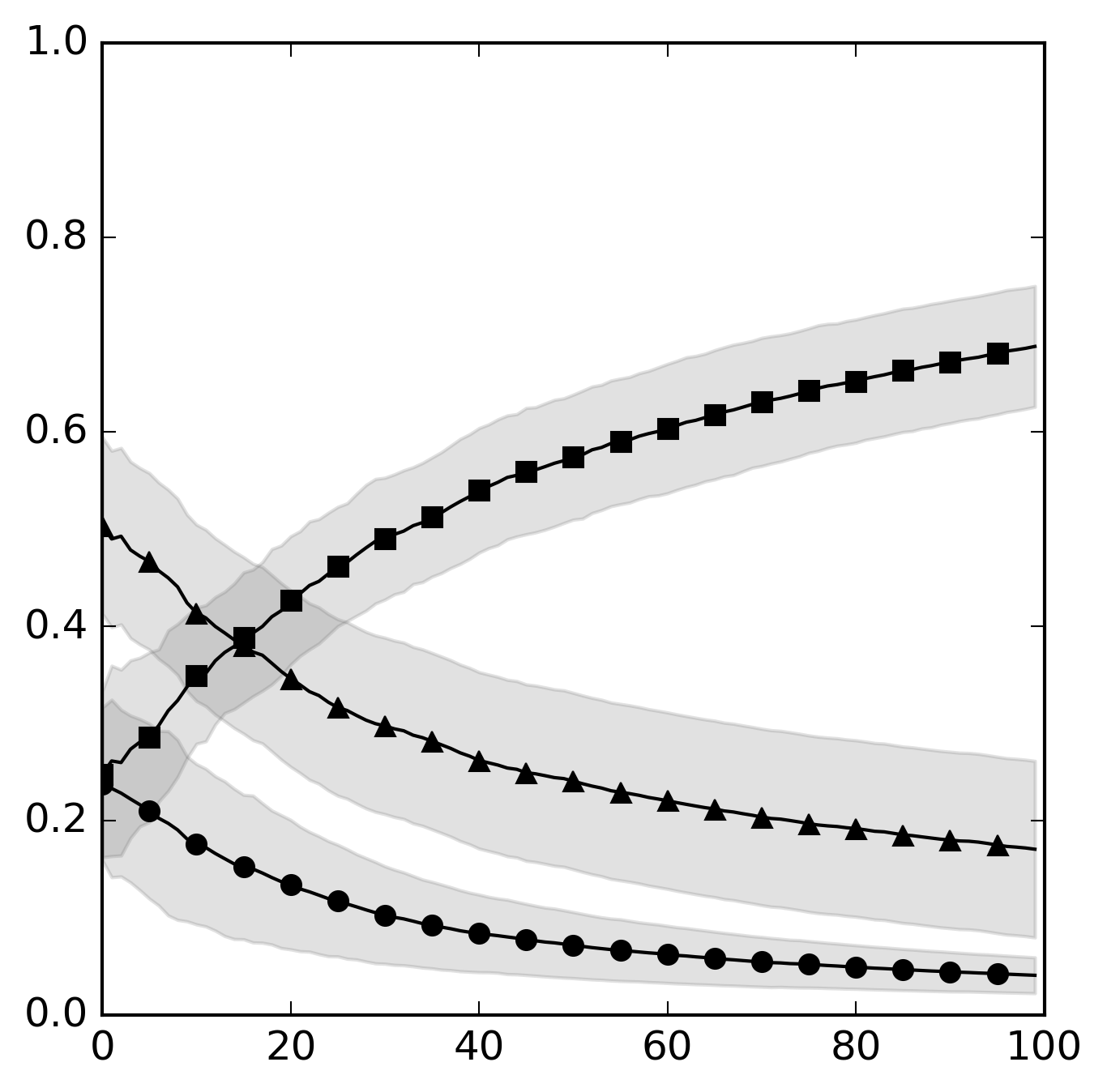}
    \end{minipage}
    \begin{minipage}[b]{0.23\textwidth}
        \centering
        \includegraphics[width=\textwidth]{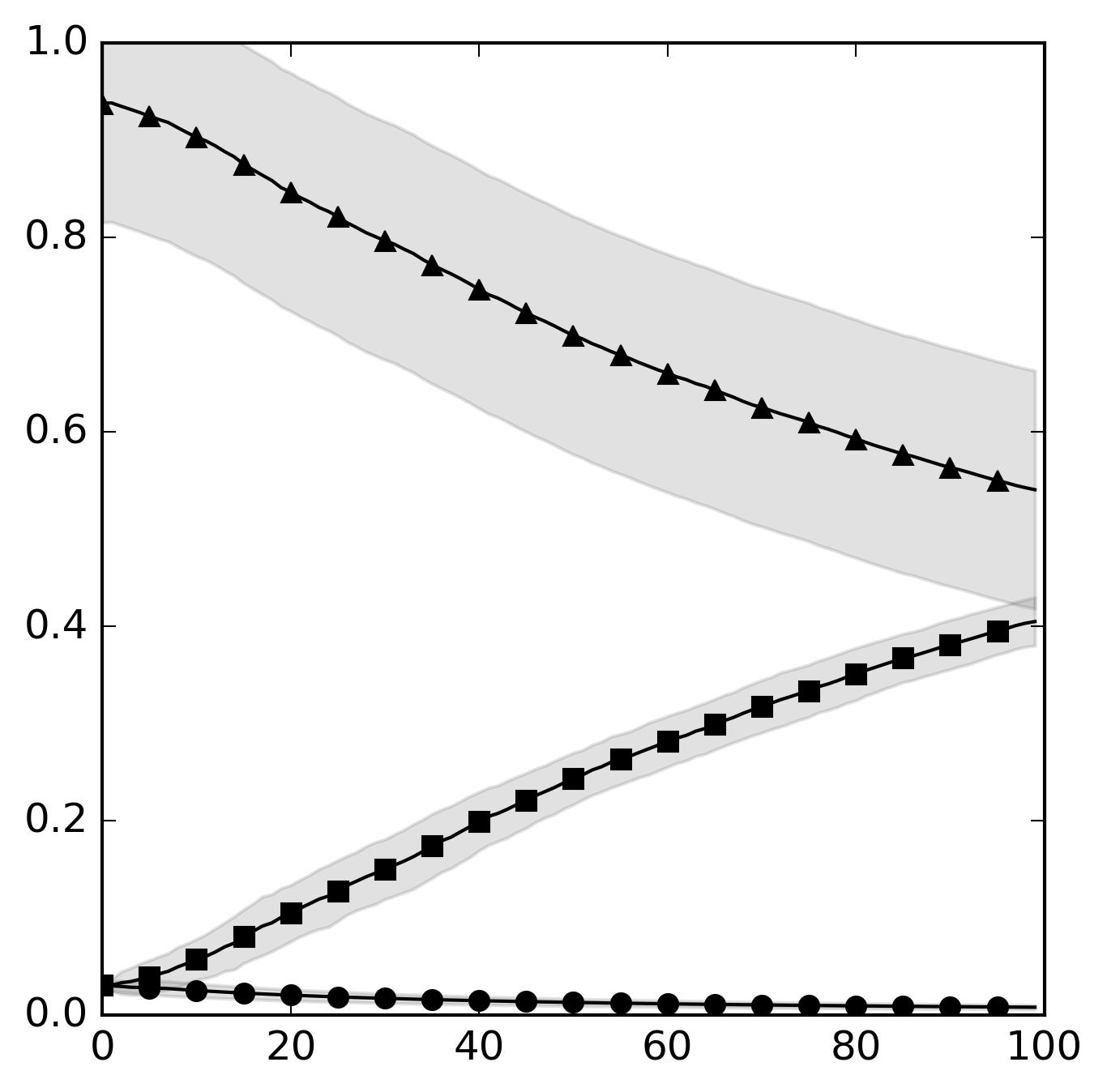}
    \end{minipage}
    \begin{minipage}[b]{0.23\textwidth}
        \centering
        \includegraphics[width=\textwidth]{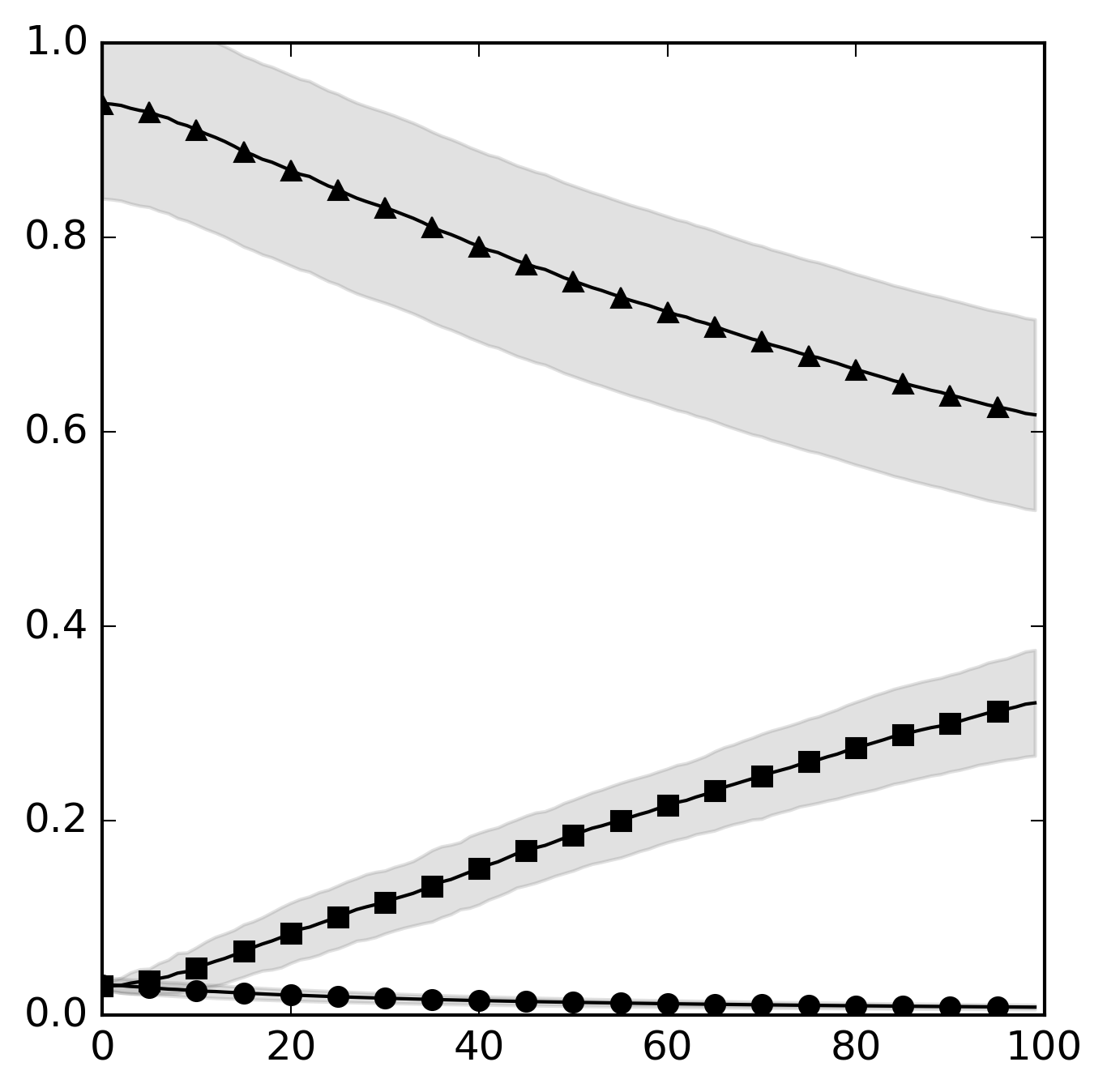}
    \end{minipage}
    \vfill
        \begin{minipage}[b]{0.23\textwidth}
    \centering
        \includegraphics[width=\textwidth]{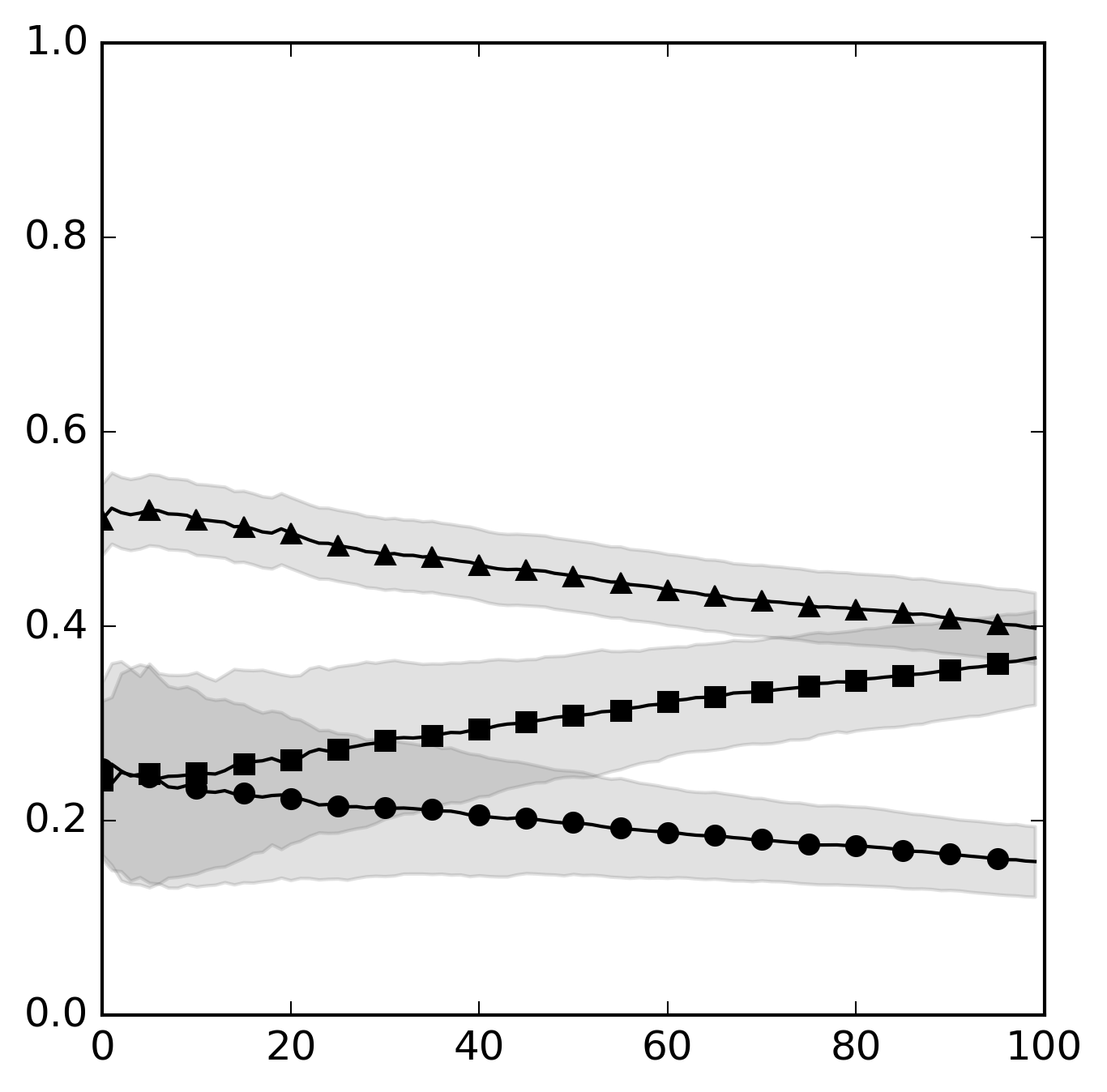}
    \end{minipage}
    \begin{minipage}[b]{0.23\textwidth}
        \centering
        \includegraphics[width=\textwidth]{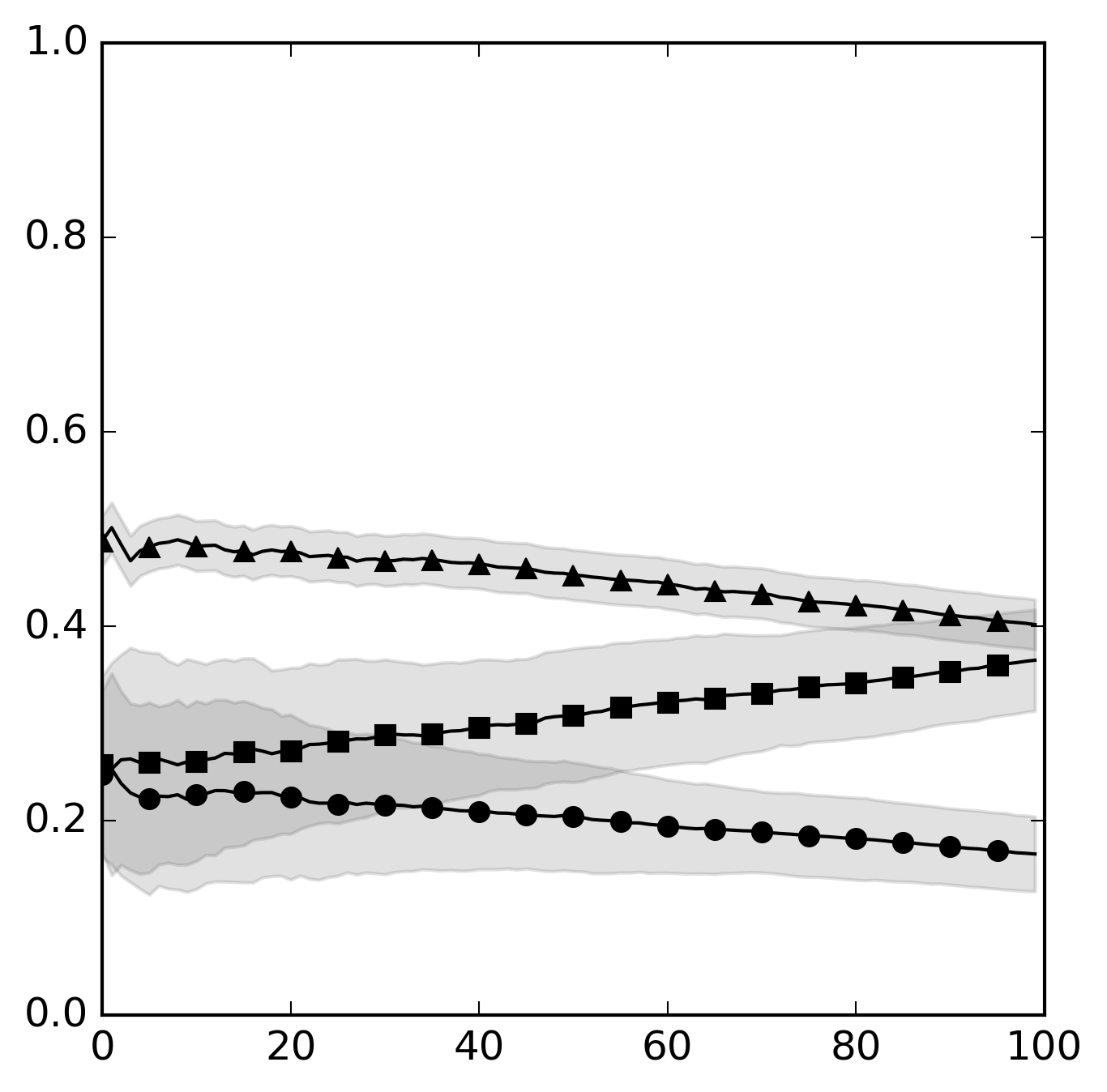}
    \end{minipage}
    \begin{minipage}[b]{0.23\textwidth}
        \centering
        \includegraphics[width=\textwidth]{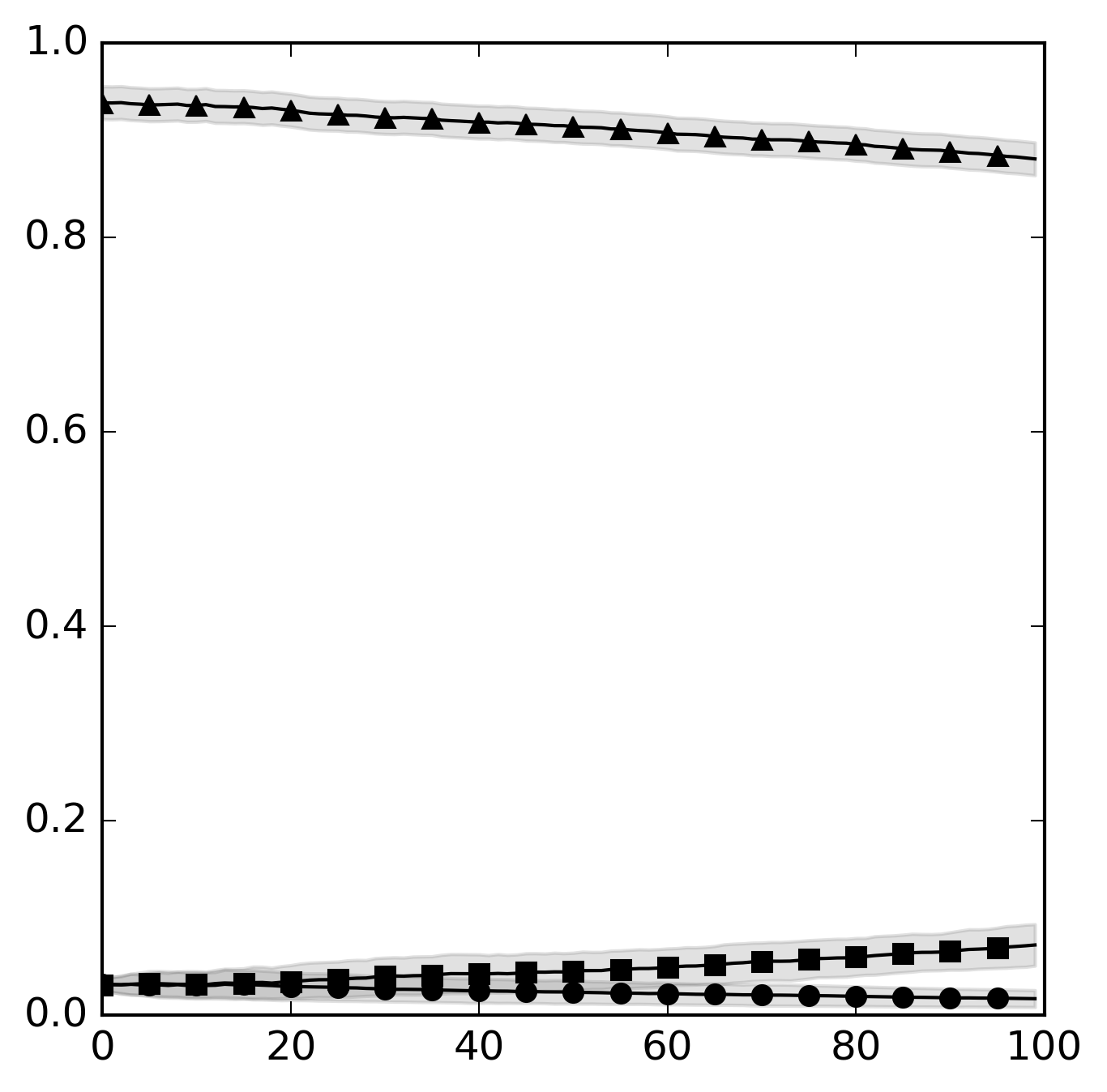}
    \end{minipage}
    \begin{minipage}[b]{0.23\textwidth}
        \centering
        \includegraphics[width=\textwidth]{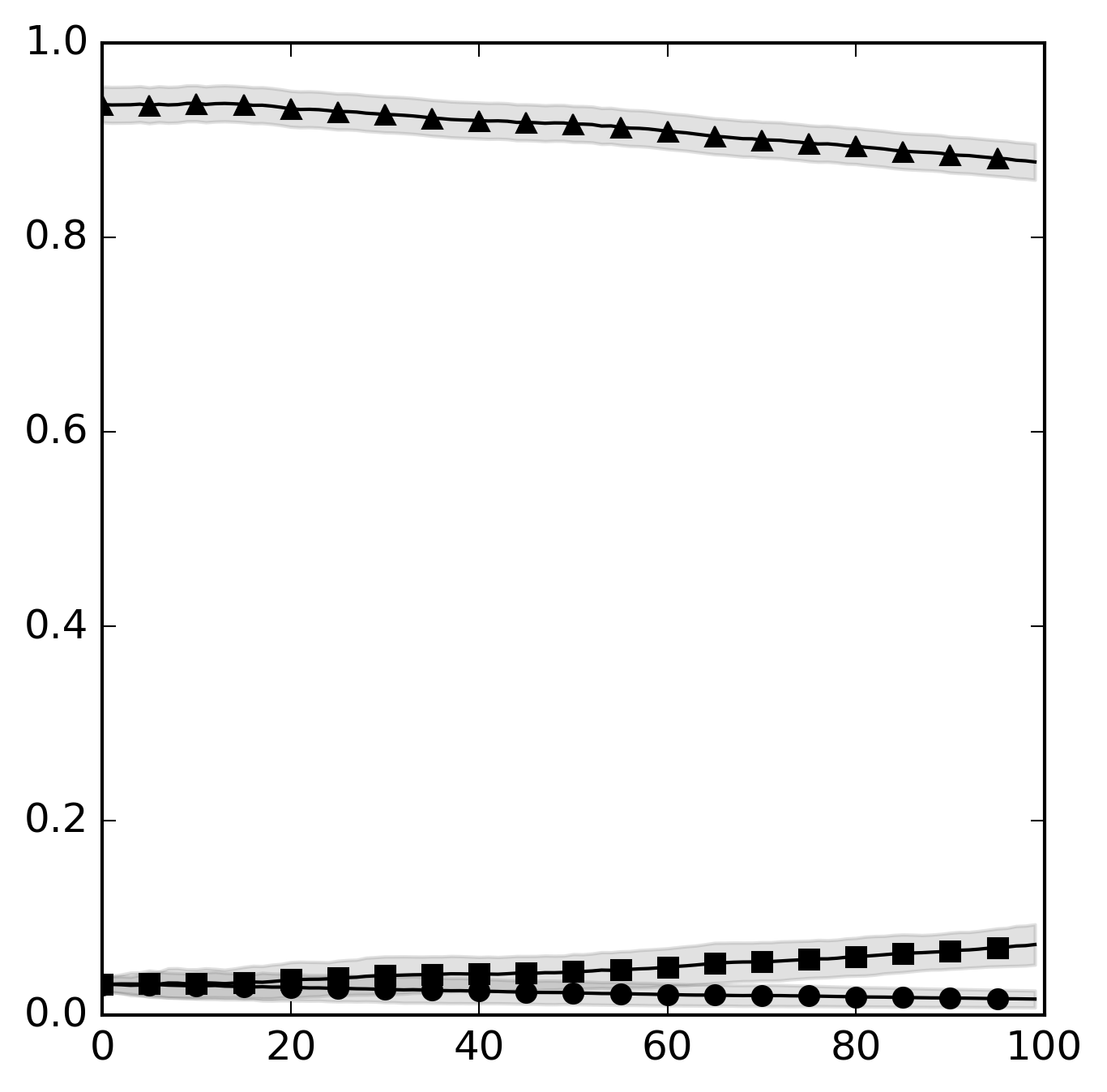}
    \end{minipage}
    \vfill
        \begin{minipage}[b]{0.23\textwidth}
    \centering
        \includegraphics[width=\textwidth]{png_files/joint_action_frequencies/joint_action_frequencies_mappo_2_0.0_plot.png}
    \end{minipage}
    \begin{minipage}[b]{0.23\textwidth}
        \centering
        \includegraphics[width=\textwidth]{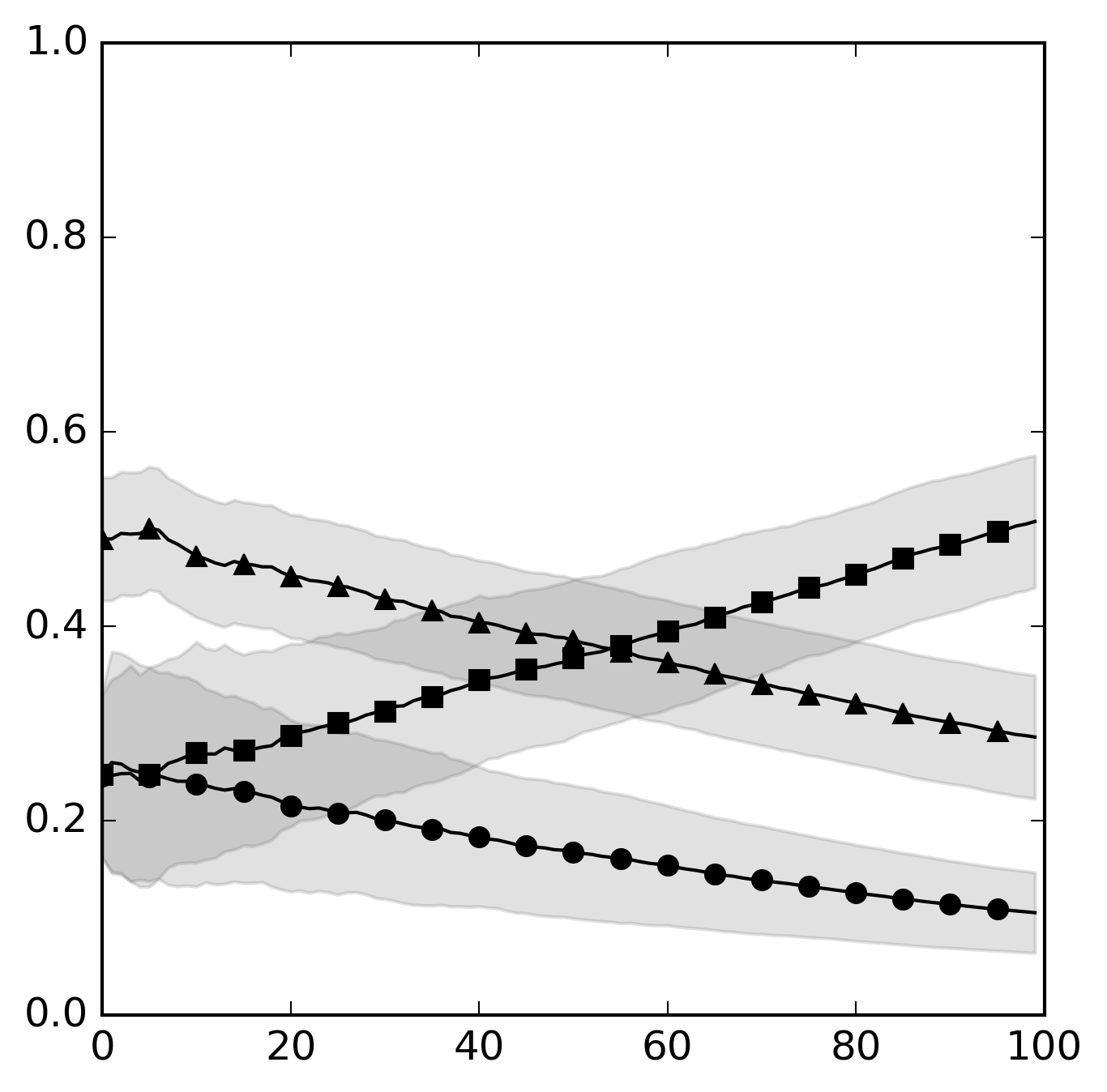}
    \end{minipage}
    \begin{minipage}[b]{0.23\textwidth}
        \centering
        \includegraphics[width=\textwidth]{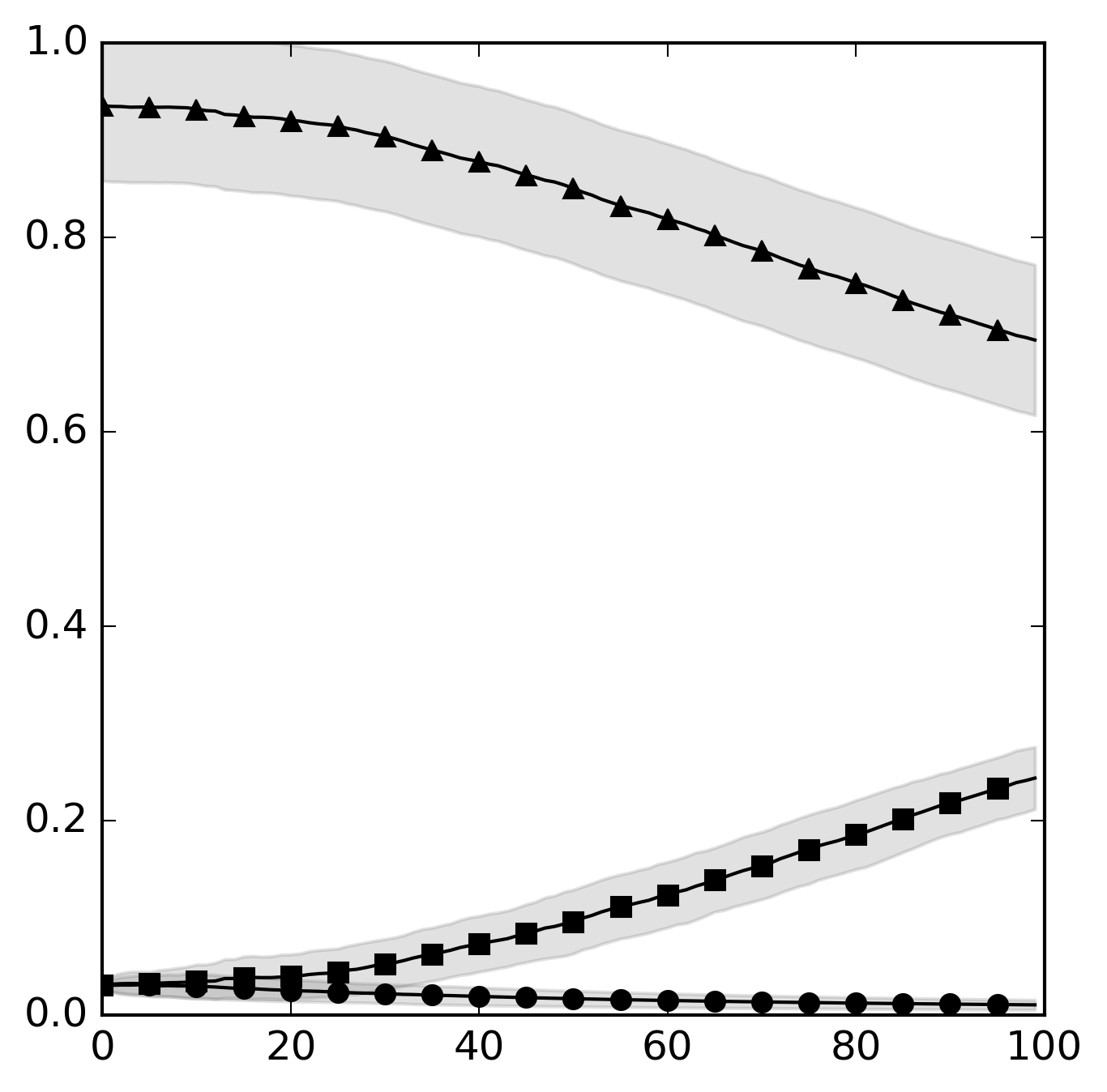}
    \end{minipage}
    \begin{minipage}[b]{0.23\textwidth}
        \centering
        \includegraphics[width=\textwidth]{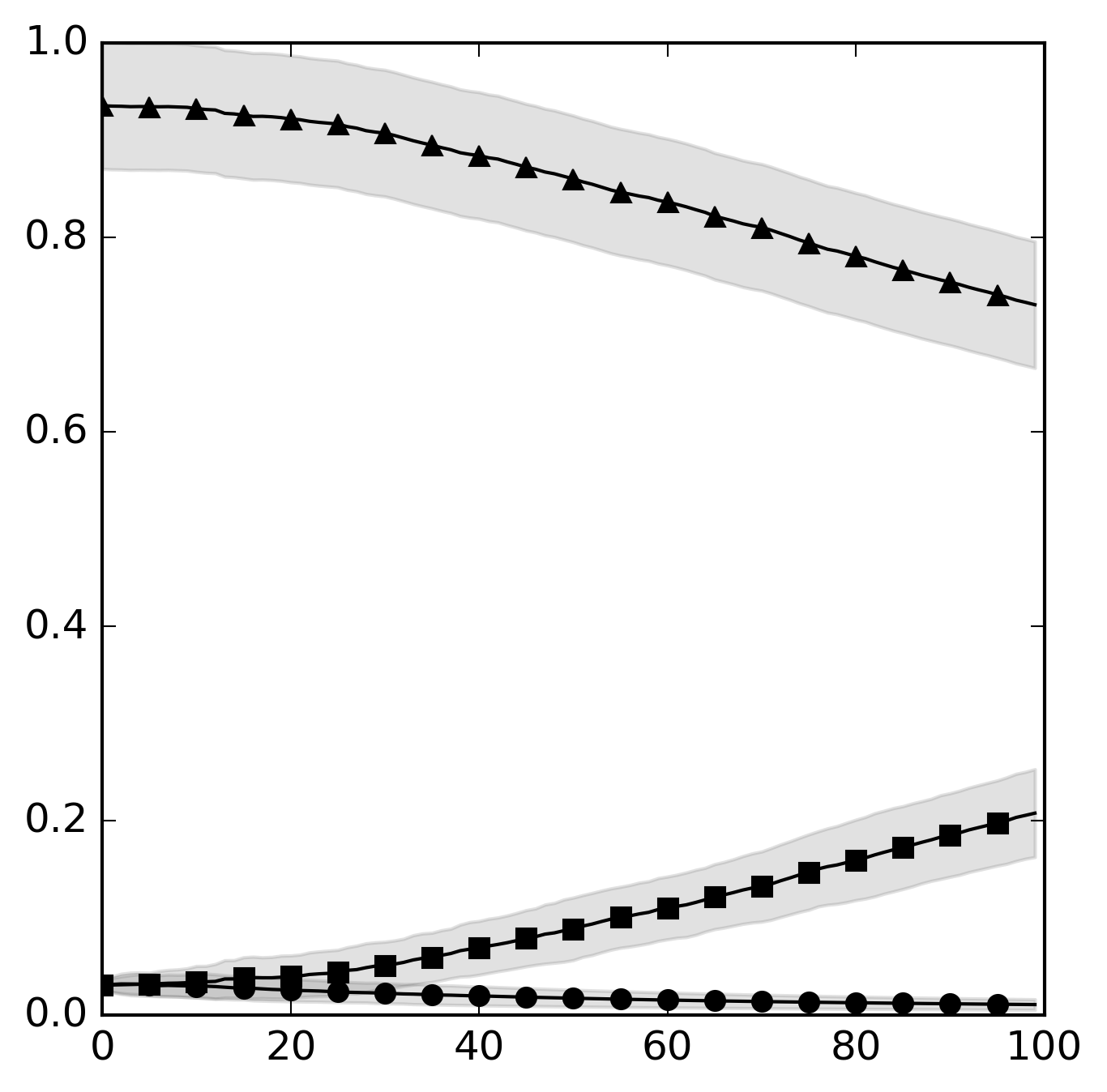}
    \end{minipage}
\label{fig:joint_action_frequencies}
\end{figure}

\begin{figure}[!ht]
    \caption{\small Average (across repeats) training losses. Rows (up to down): D3QN, MAPPO Actor, MAPPO Critic. Columns (left to right): $(n=2, \sigma(\beta)=0.0)$, $(n=2, \sigma(\beta)=0.5)$, $(n=5, \sigma(\beta)=0.0)$, $(n=5, \sigma(\beta)=0.5)$.}
    \centering
    \begin{minipage}[b]{0.23\textwidth}
    \centering
        \includegraphics[width=\textwidth]{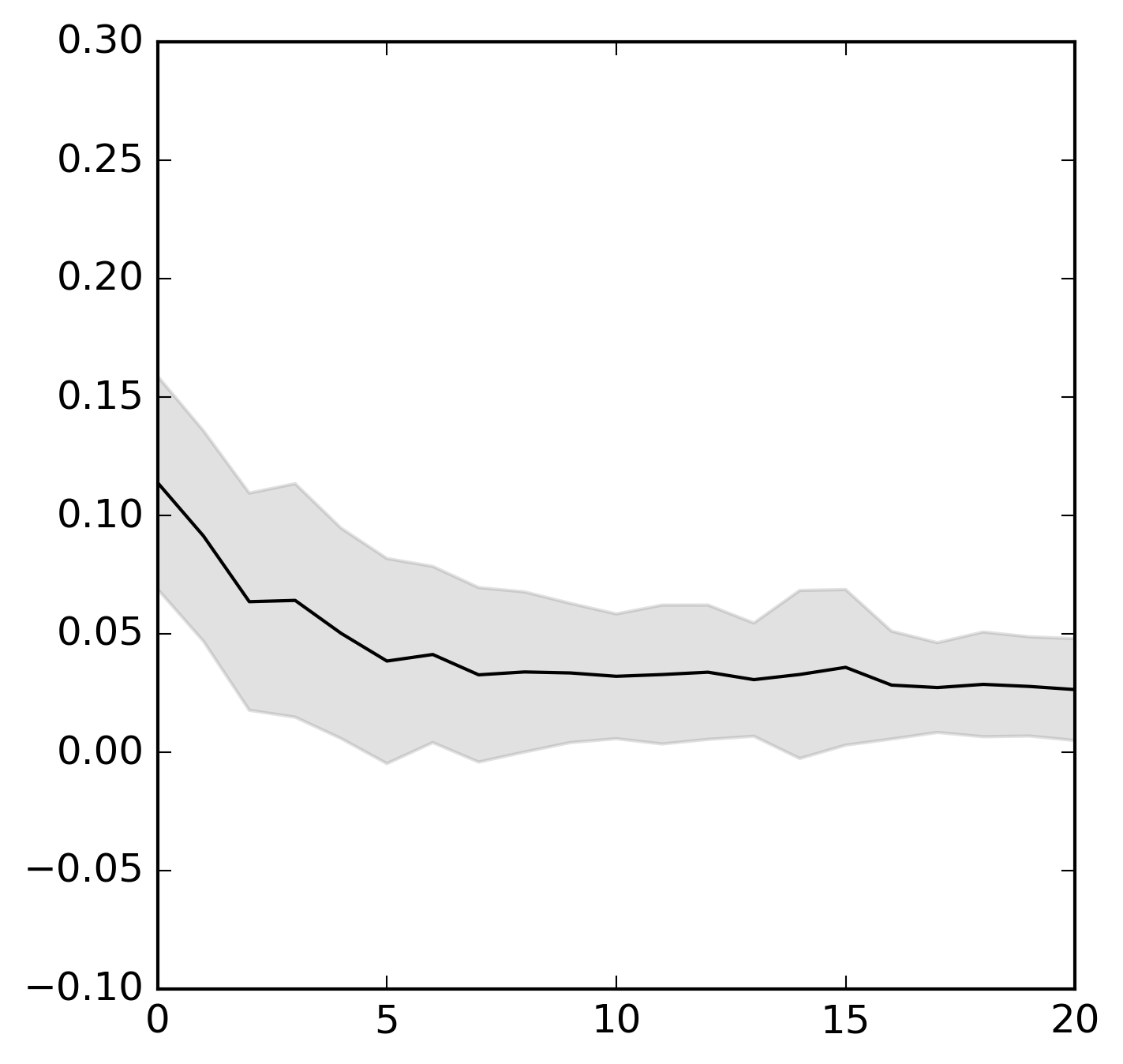}
    \end{minipage}
    \begin{minipage}[b]{0.23\textwidth}
        \centering
        \includegraphics[width=\textwidth]{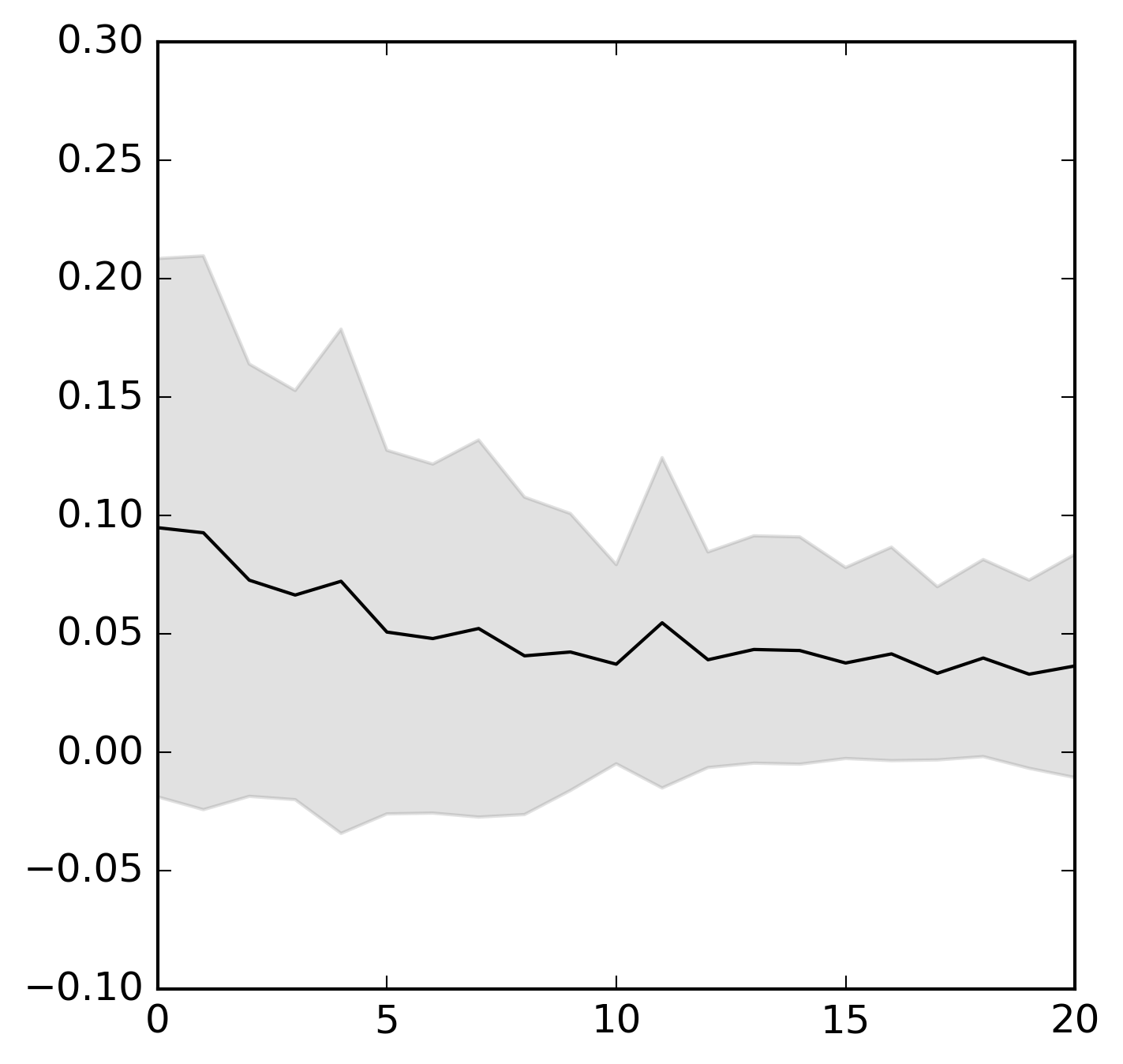}
    \end{minipage}
    \begin{minipage}[b]{0.23\textwidth}
        \centering
        \includegraphics[width=\textwidth]{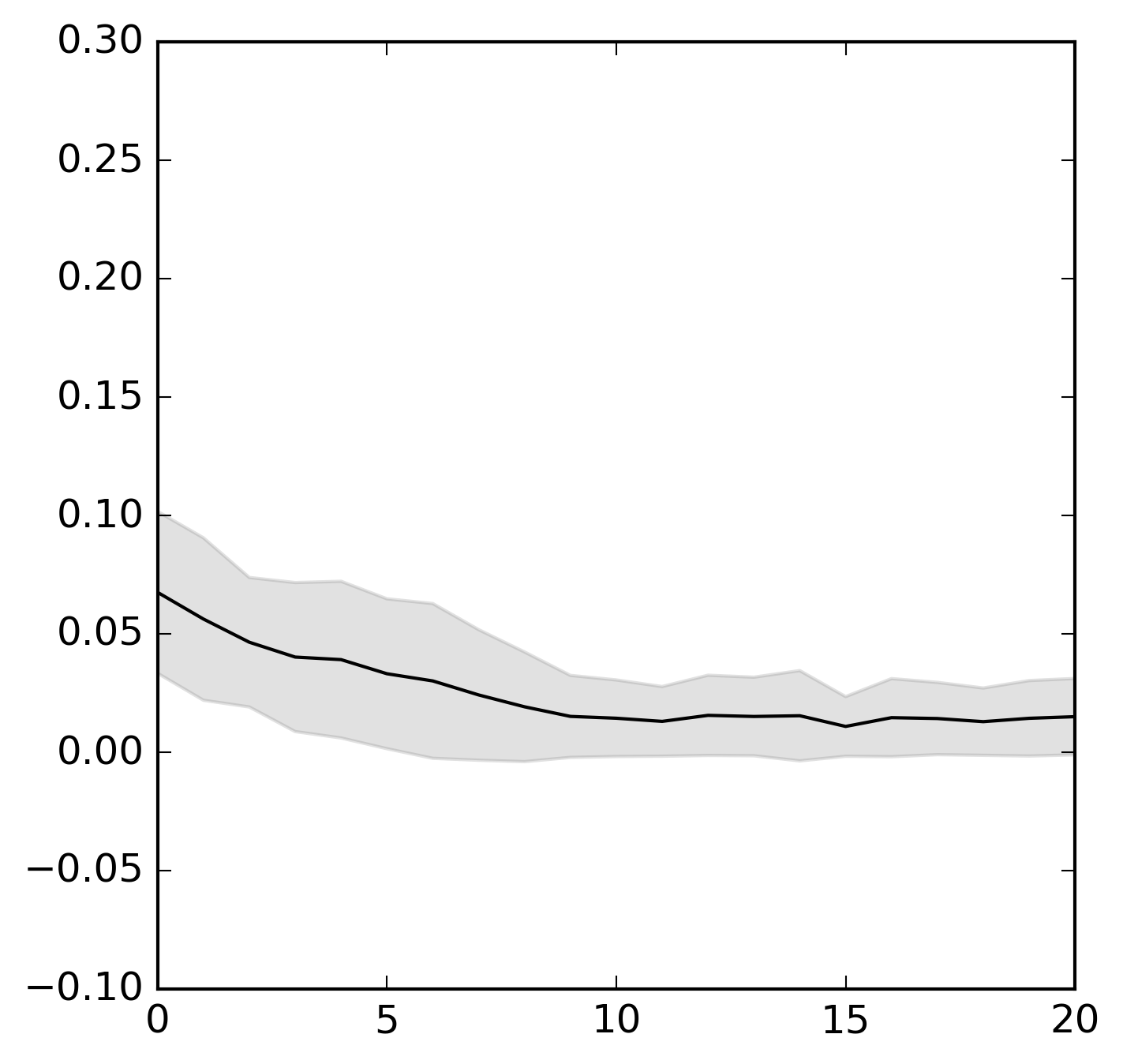}
    \end{minipage}
    \begin{minipage}[b]{0.23\textwidth}
        \centering
        \includegraphics[width=\textwidth]{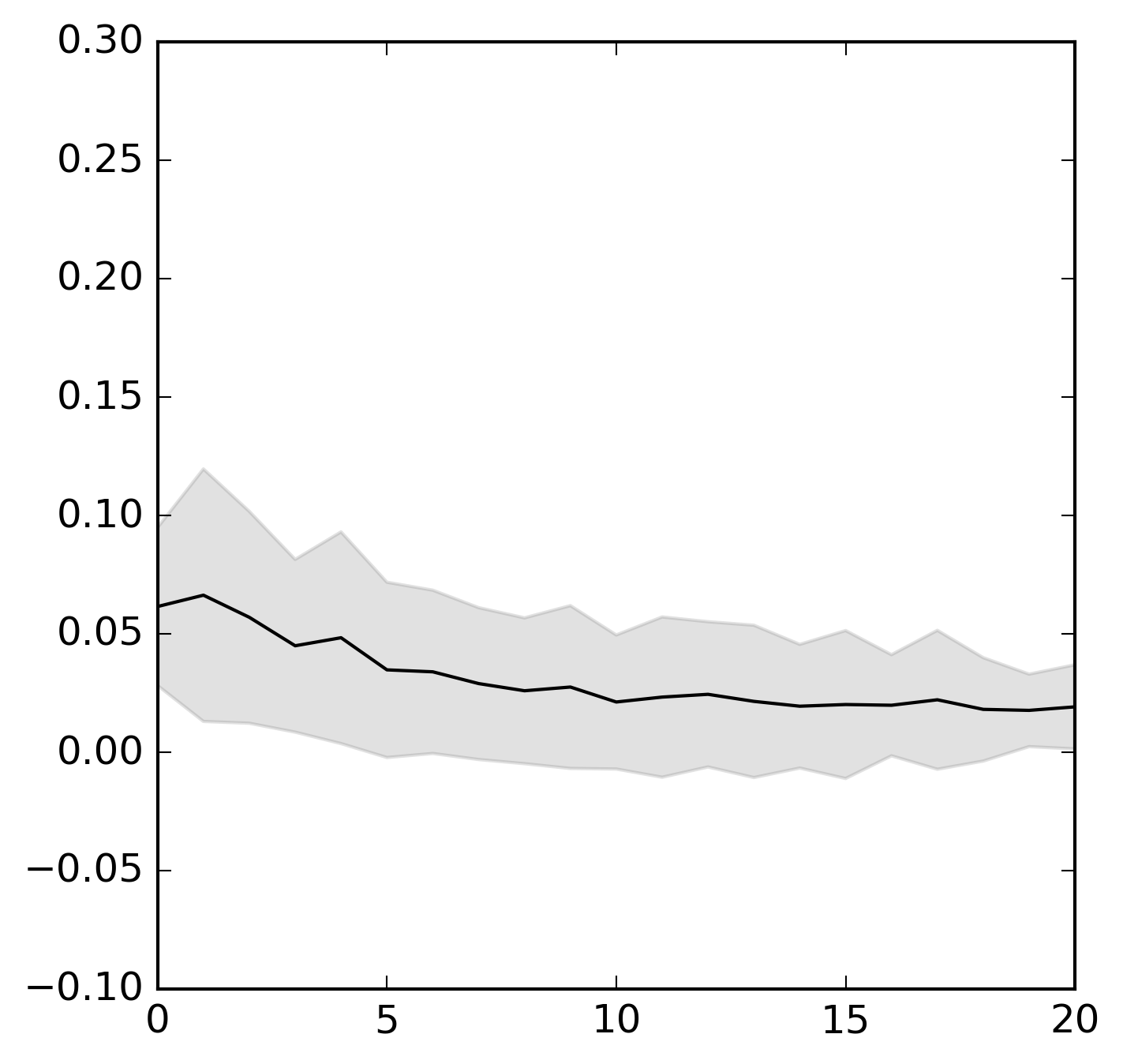}
    \end{minipage}
    \vfill
        \begin{minipage}[b]{0.23\textwidth}
    \centering
        \includegraphics[width=\textwidth]{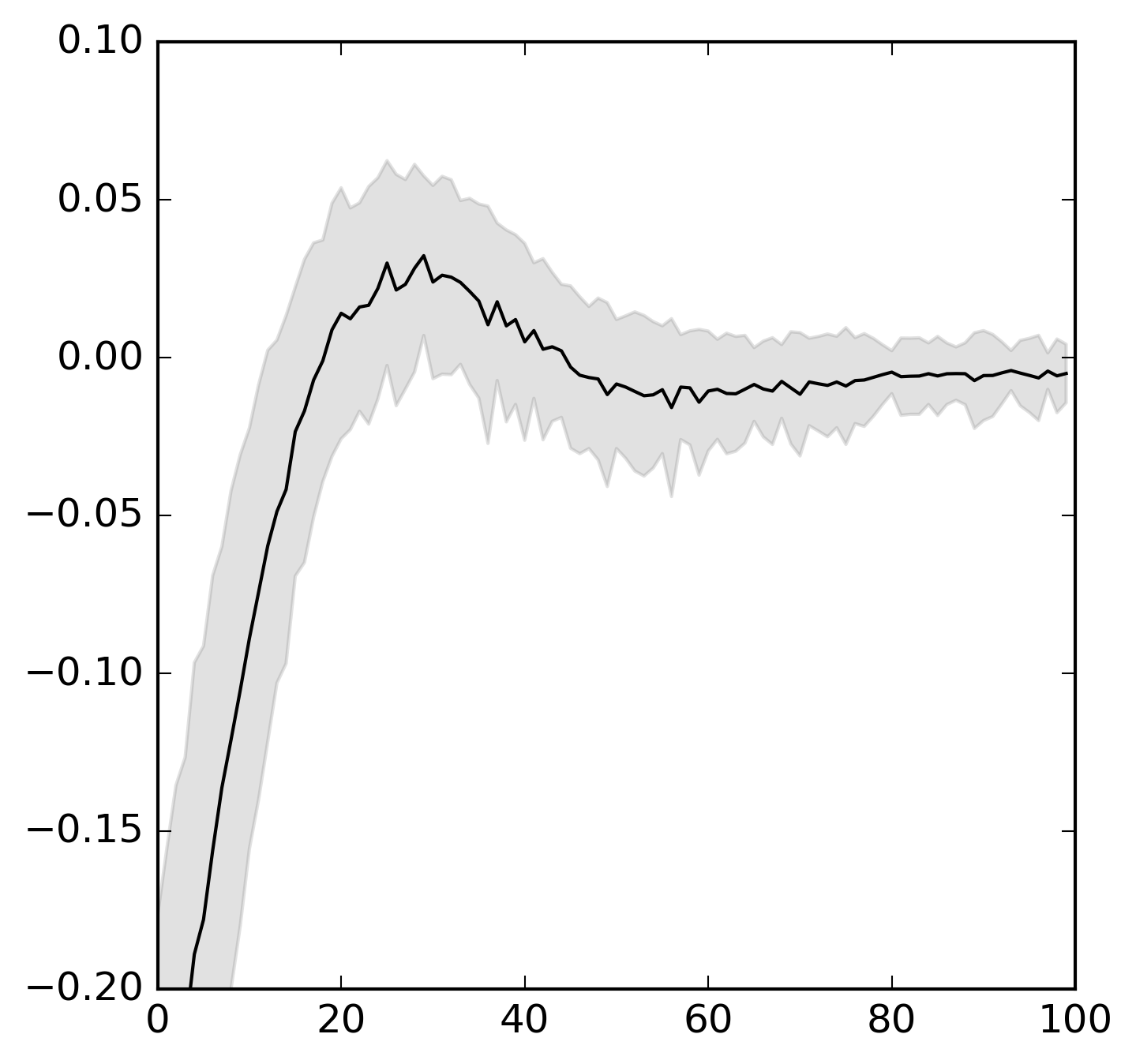}
    \end{minipage}
    \begin{minipage}[b]{0.23\textwidth}
        \centering
        \includegraphics[width=\textwidth]{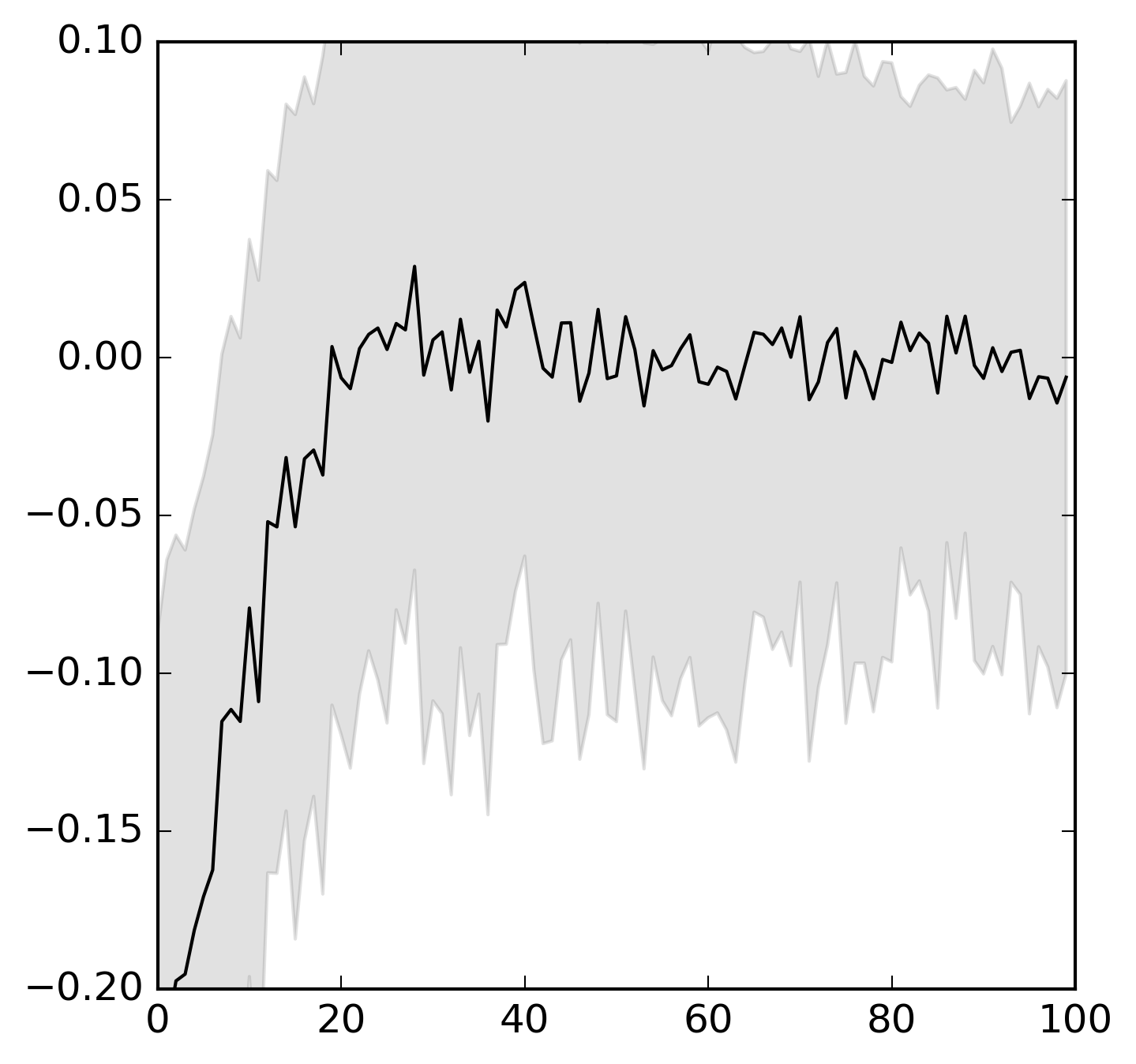}
    \end{minipage}
    \begin{minipage}[b]{0.23\textwidth}
        \centering
        \includegraphics[width=\textwidth]{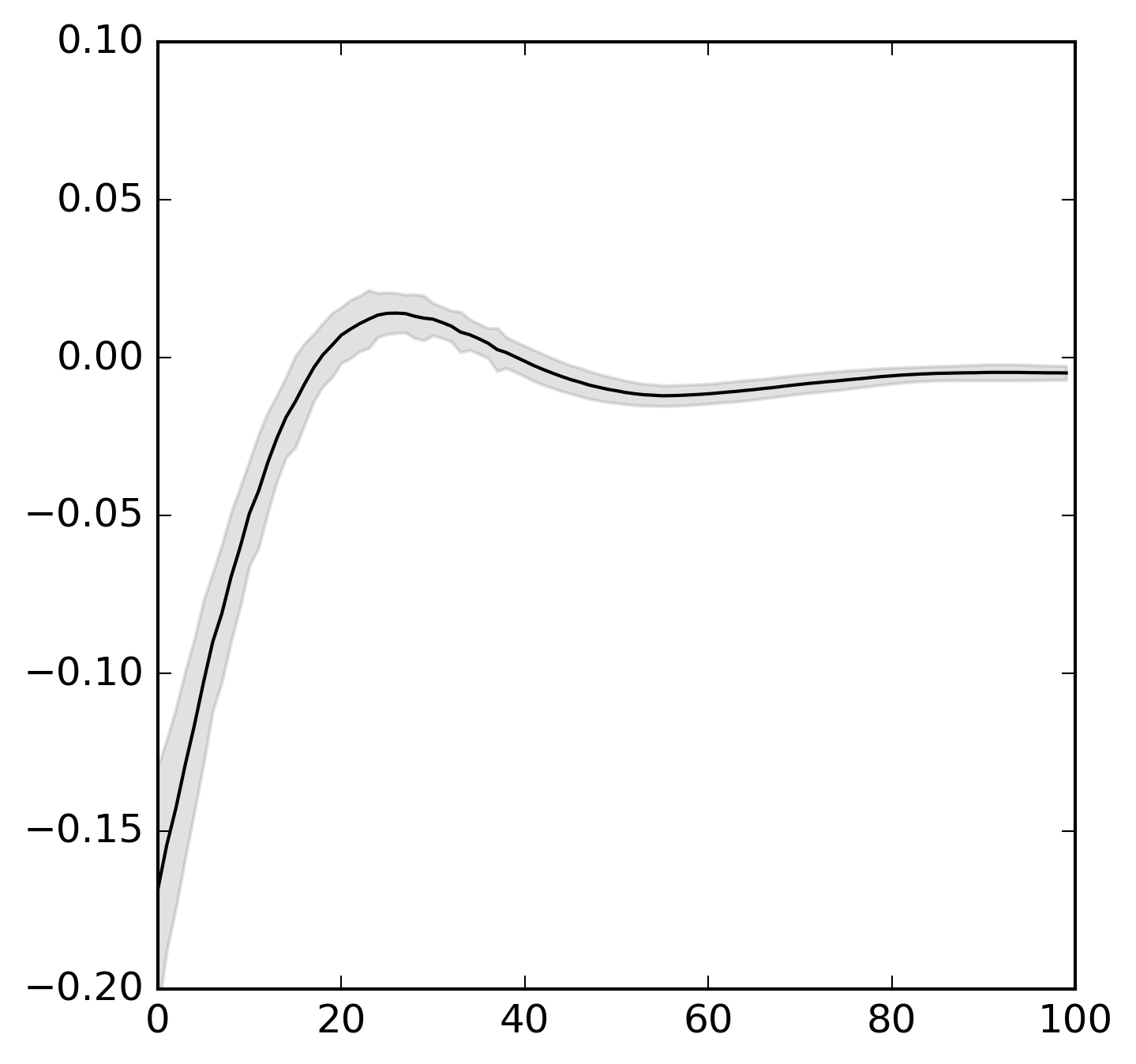}
    \end{minipage}
    \begin{minipage}[b]{0.23\textwidth}
        \centering
        \includegraphics[width=\textwidth]{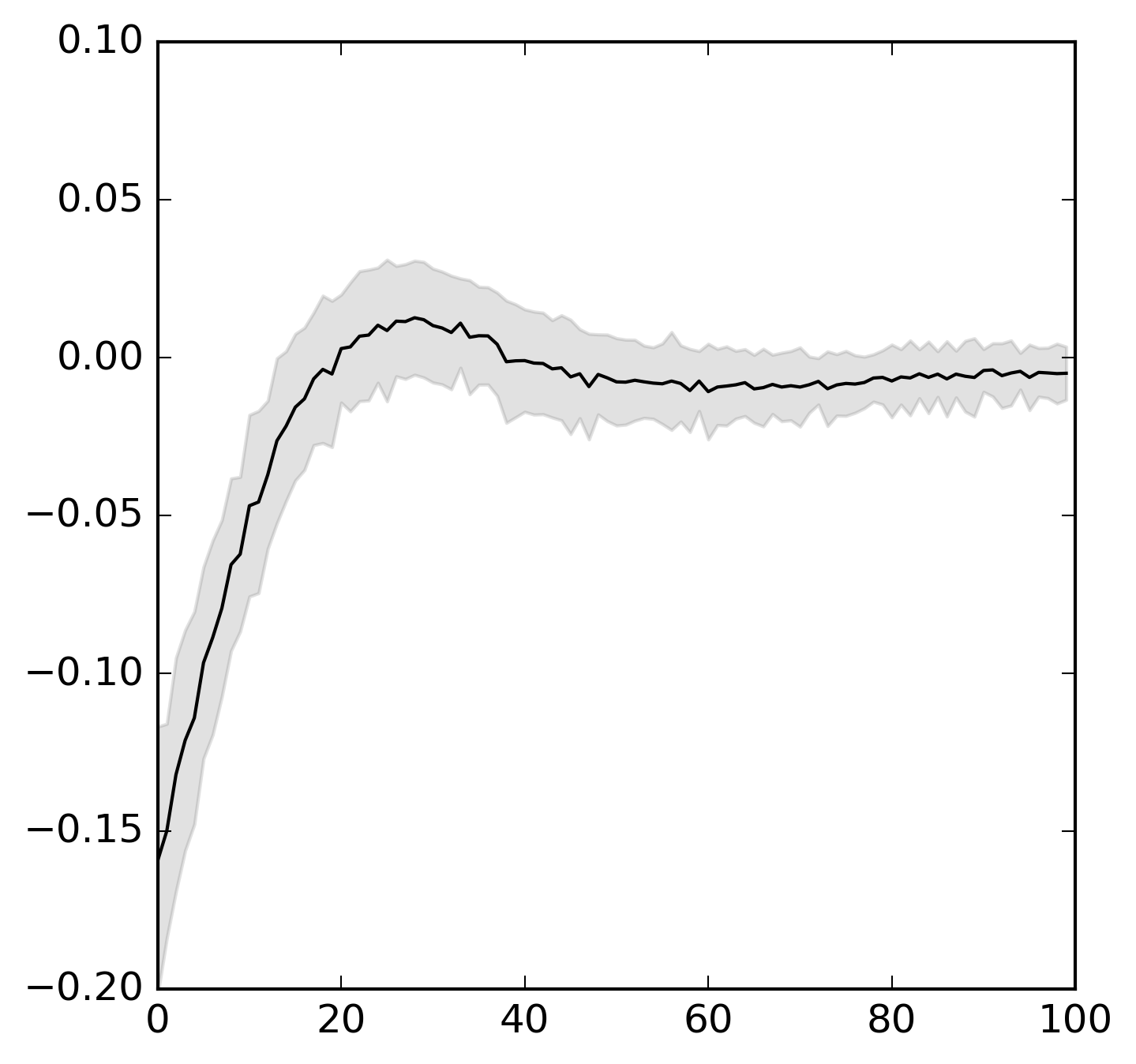}
    \end{minipage}
    \vfill
        \begin{minipage}[b]{0.23\textwidth}
    \centering
        \includegraphics[width=\textwidth]{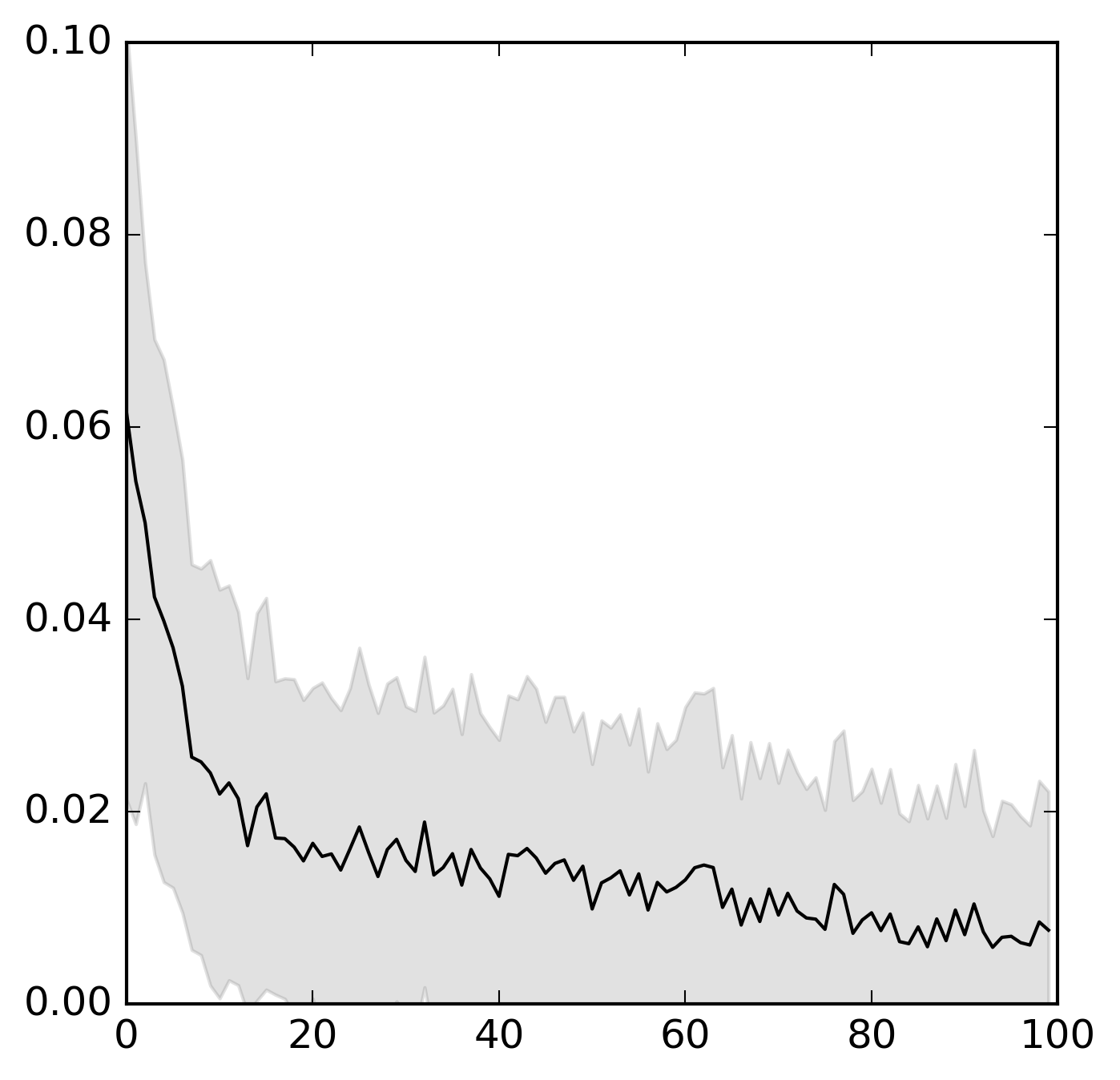}
    \end{minipage}
    \begin{minipage}[b]{0.23\textwidth}
        \centering
        \includegraphics[width=\textwidth]{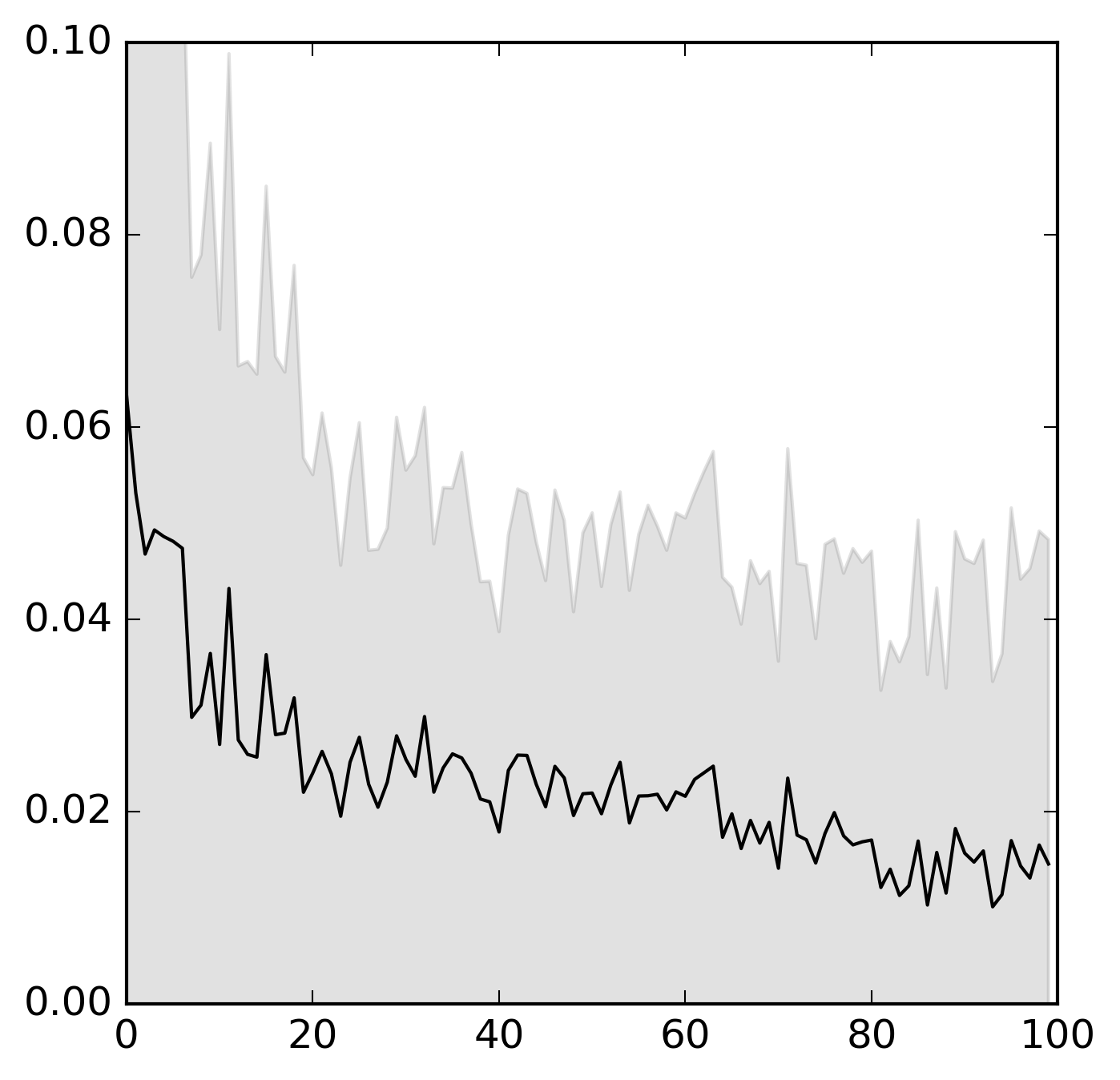}
    \end{minipage}
    \begin{minipage}[b]{0.23\textwidth}
        \centering
        \includegraphics[width=\textwidth]{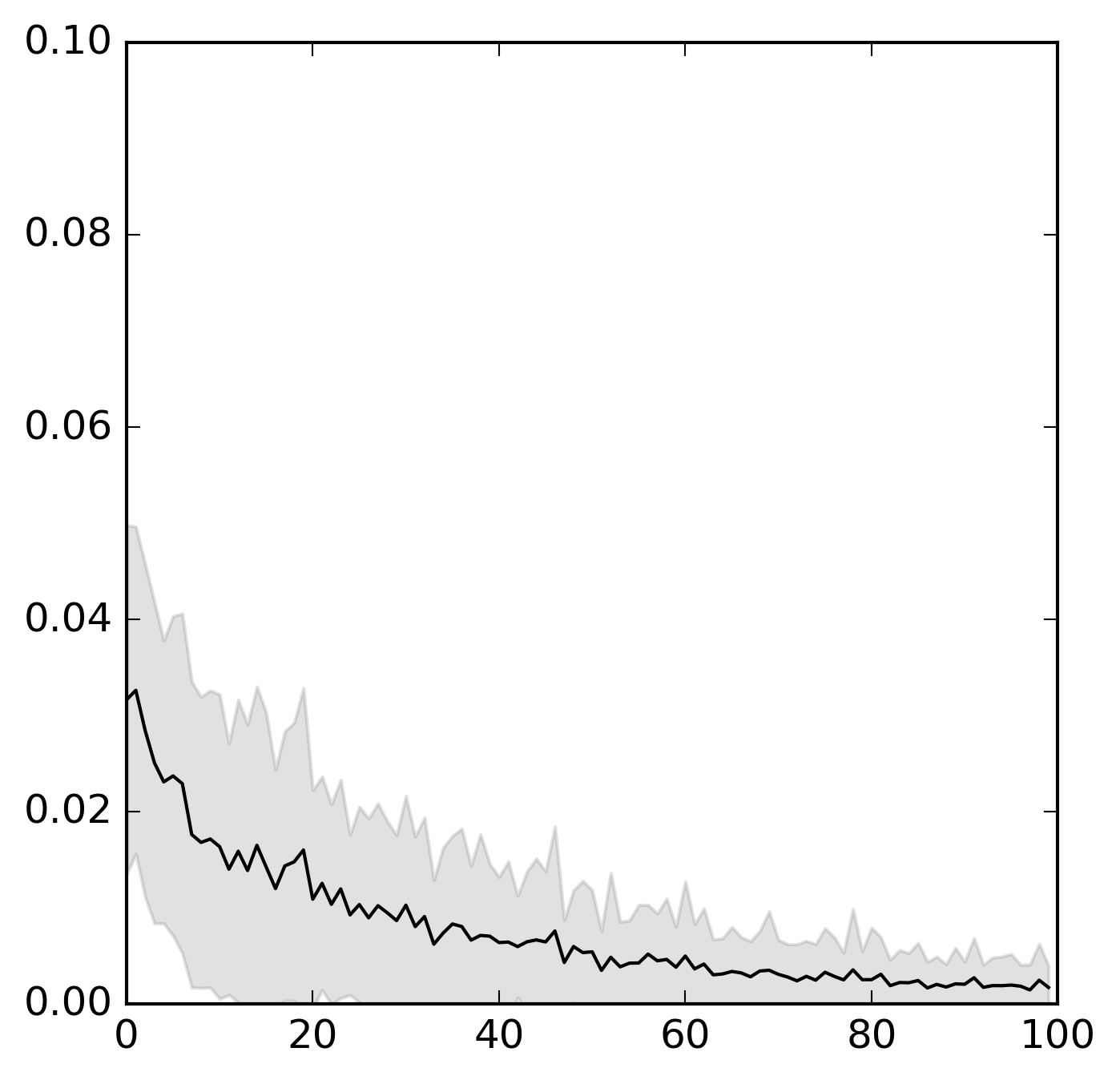}
    \end{minipage}
    \begin{minipage}[b]{0.23\textwidth}
        \centering
        \includegraphics[width=\textwidth]{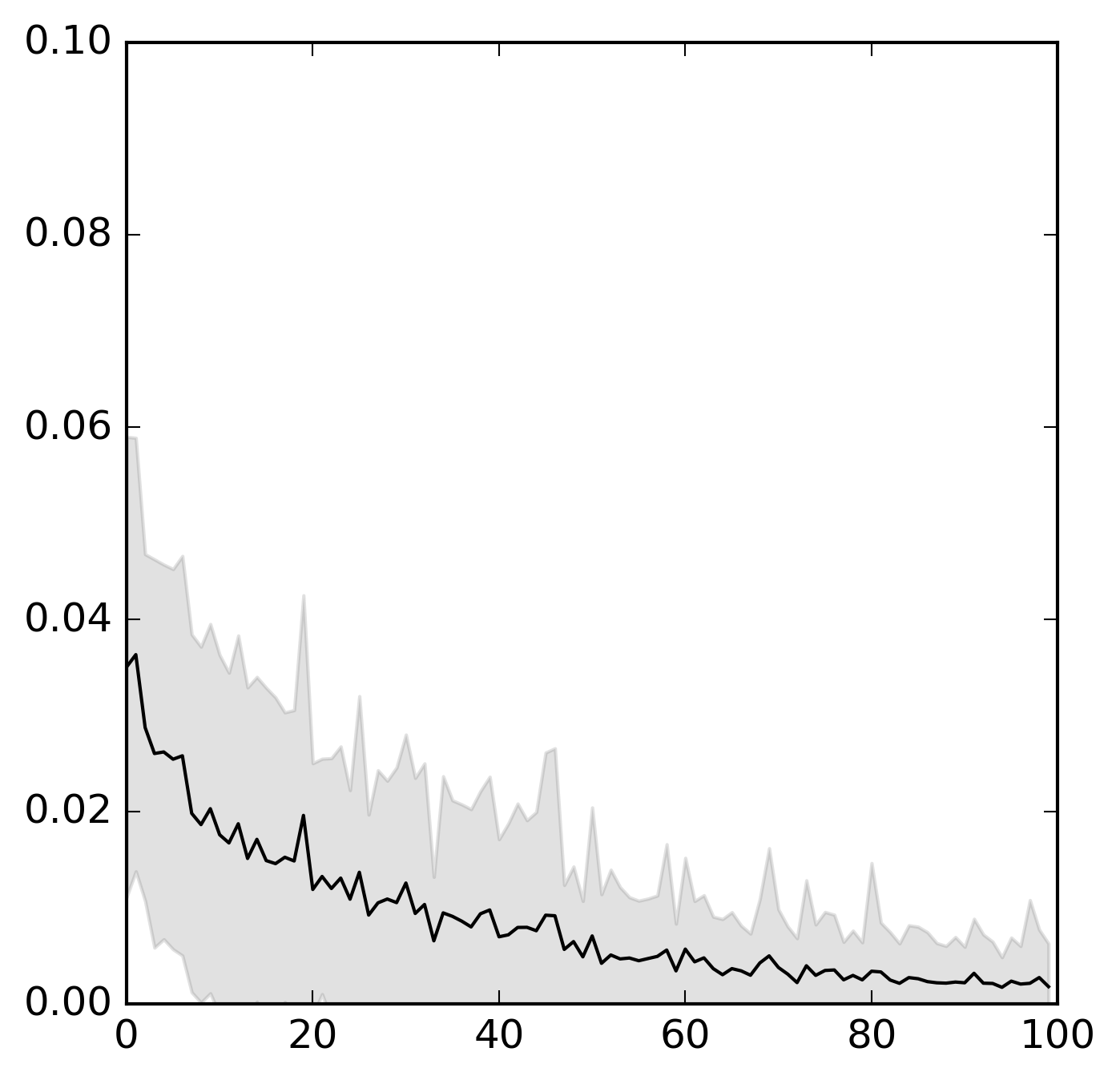}
    \end{minipage}
\label{fig:loss_functions}
\end{figure}

\begin{figure}[!ht]
    \caption{\small Average (across repeats) training policy values. Rows (up to down): D3QN Q-values, MAPPO policy values. Columns (left to right): $(n=2, \sigma(\beta)=0.0)$, $(n=2, \sigma(\beta)=0.5)$, $(n=5, \sigma(\beta)=0.0)$, $(n=5, \sigma(\beta)=0.5)$.}
    \centering
    \begin{minipage}[b]{0.23\textwidth}
    \centering
        \includegraphics[width=\textwidth]{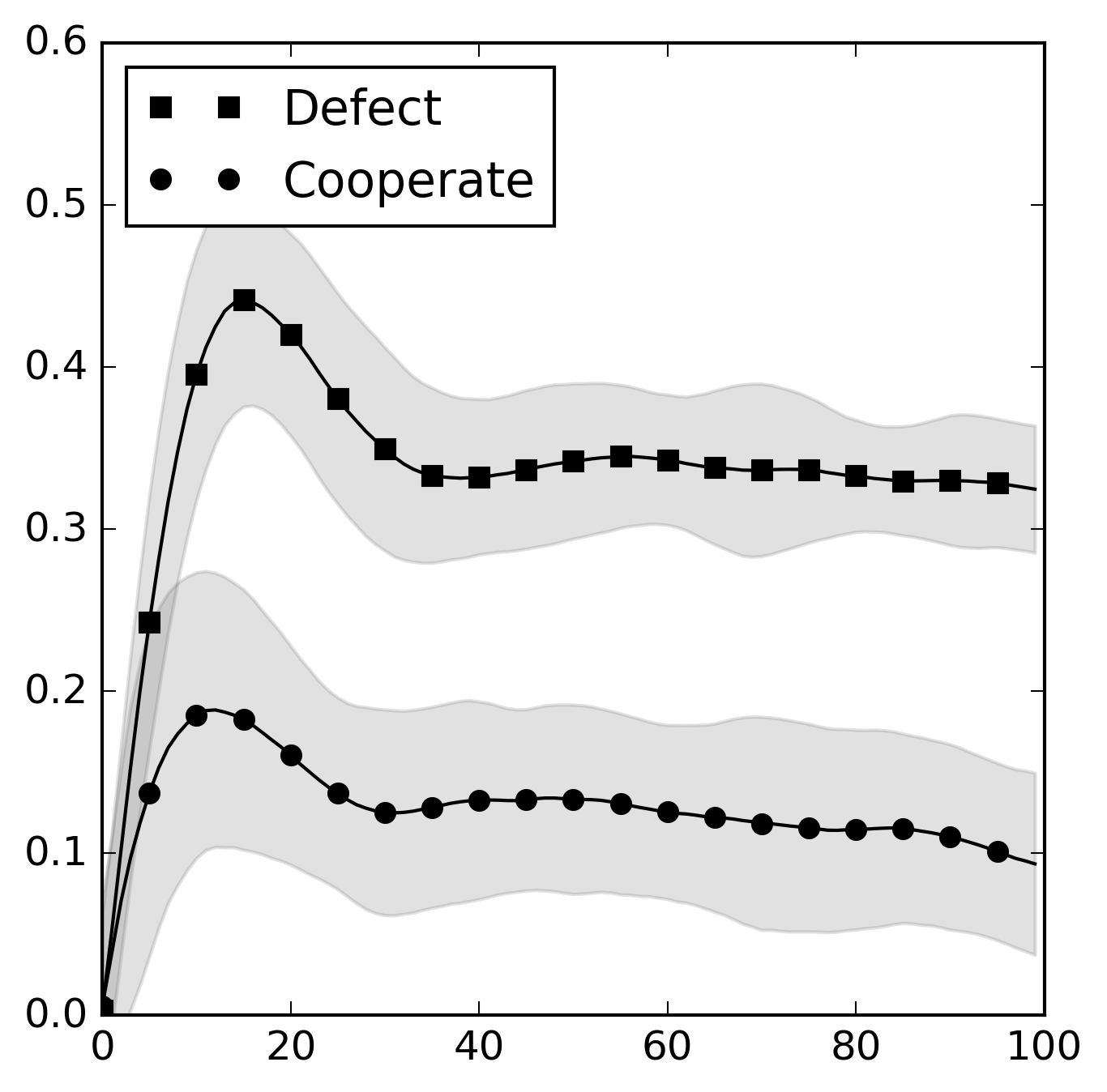}
    \end{minipage}
    \begin{minipage}[b]{0.23\textwidth}
        \centering
        \includegraphics[width=\textwidth]{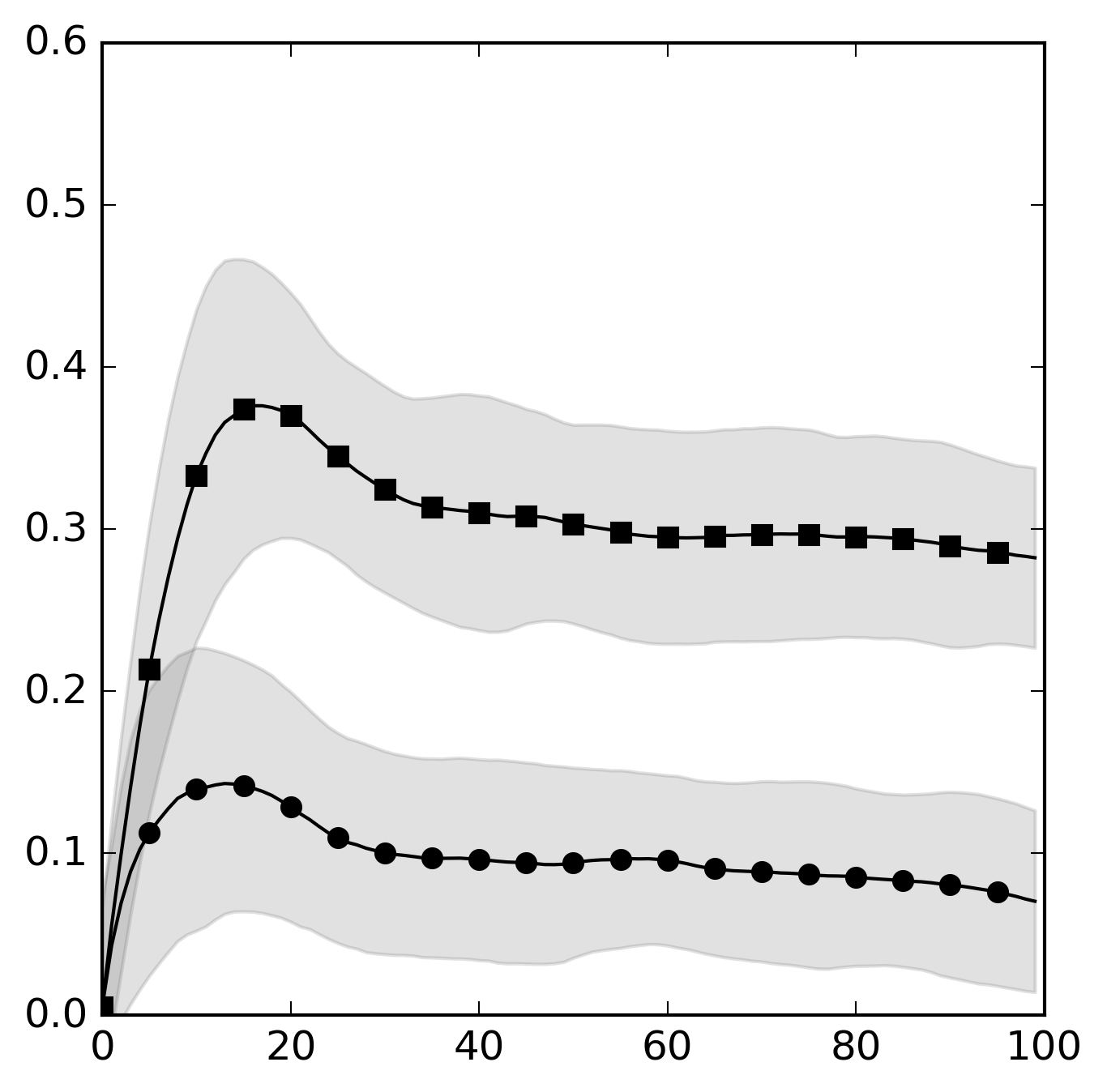}
    \end{minipage}
    \begin{minipage}[b]{0.23\textwidth}
        \centering
        \includegraphics[width=\textwidth]{png_files/loss/loss_d3qn_5_0.0_plot.png}
    \end{minipage}
    \begin{minipage}[b]{0.23\textwidth}
        \centering
        \includegraphics[width=\textwidth]{png_files/loss/loss_d3qn_5_0.5_plot.png}
    \end{minipage}
    \vfill
        \begin{minipage}[b]{0.23\textwidth}
    \centering
        \includegraphics[width=\textwidth]{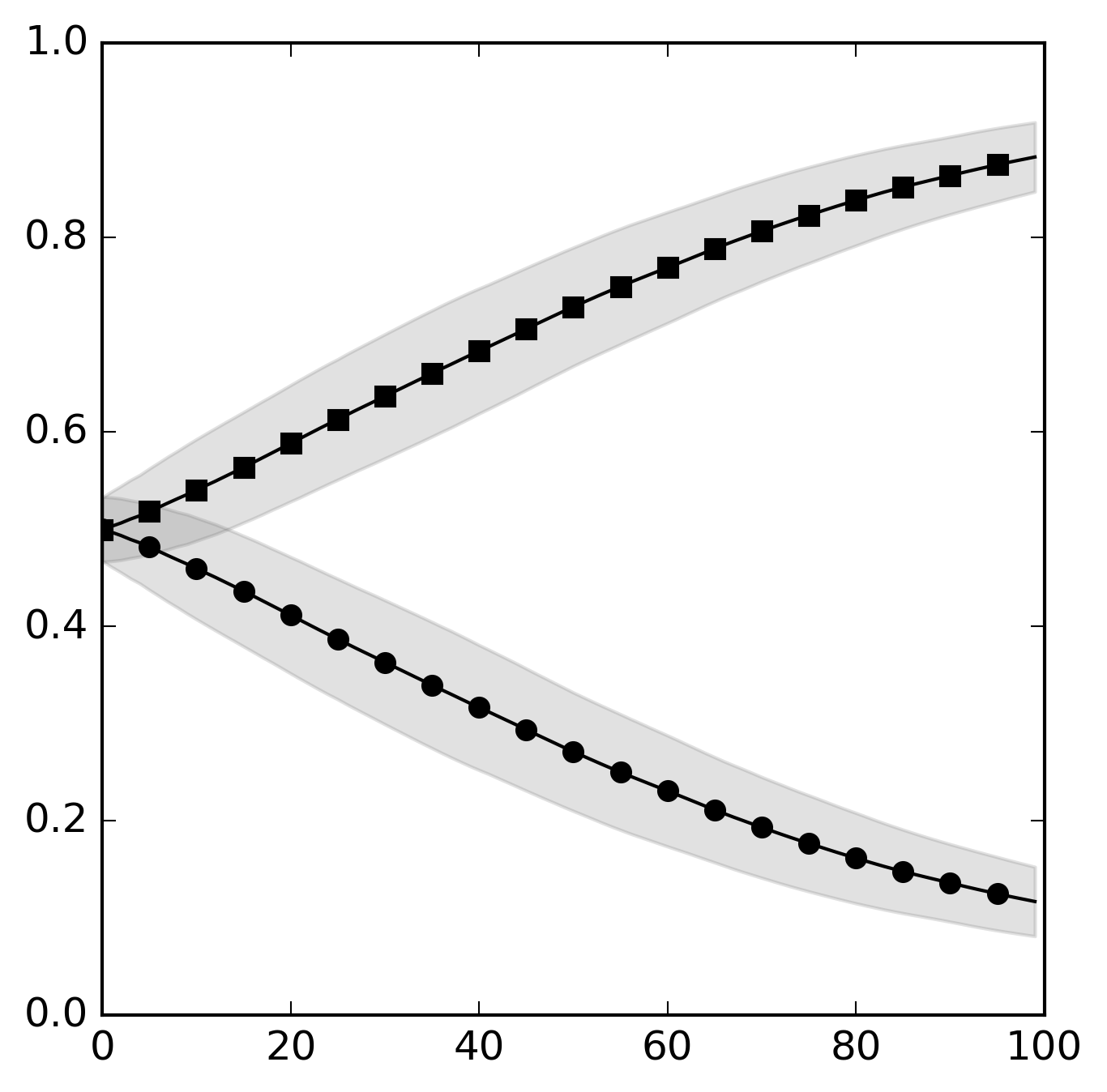}
    \end{minipage}
    \begin{minipage}[b]{0.23\textwidth}
        \centering
        \includegraphics[width=\textwidth]{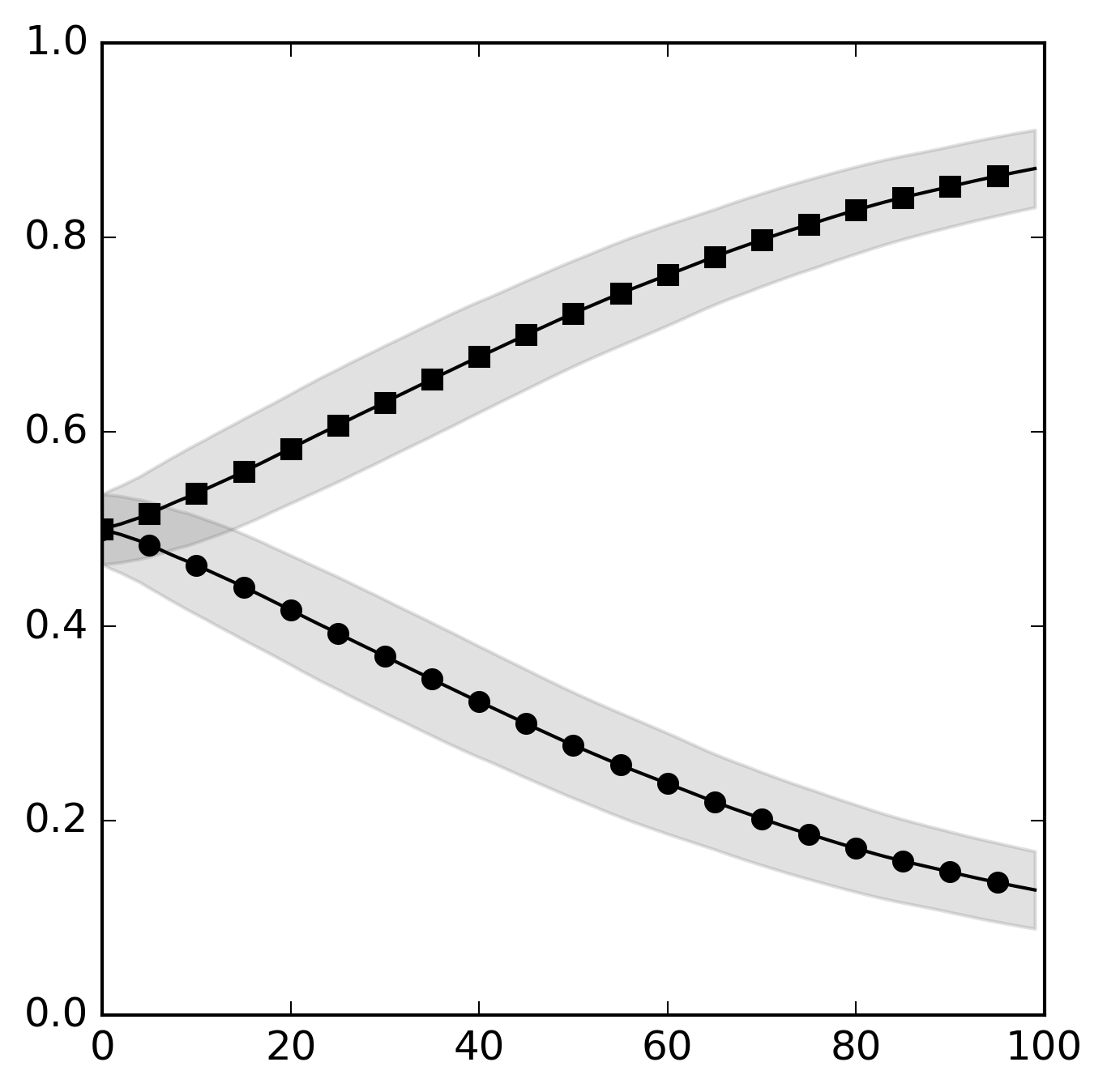}
    \end{minipage}
    \begin{minipage}[b]{0.23\textwidth}
        \centering
        \includegraphics[width=\textwidth]{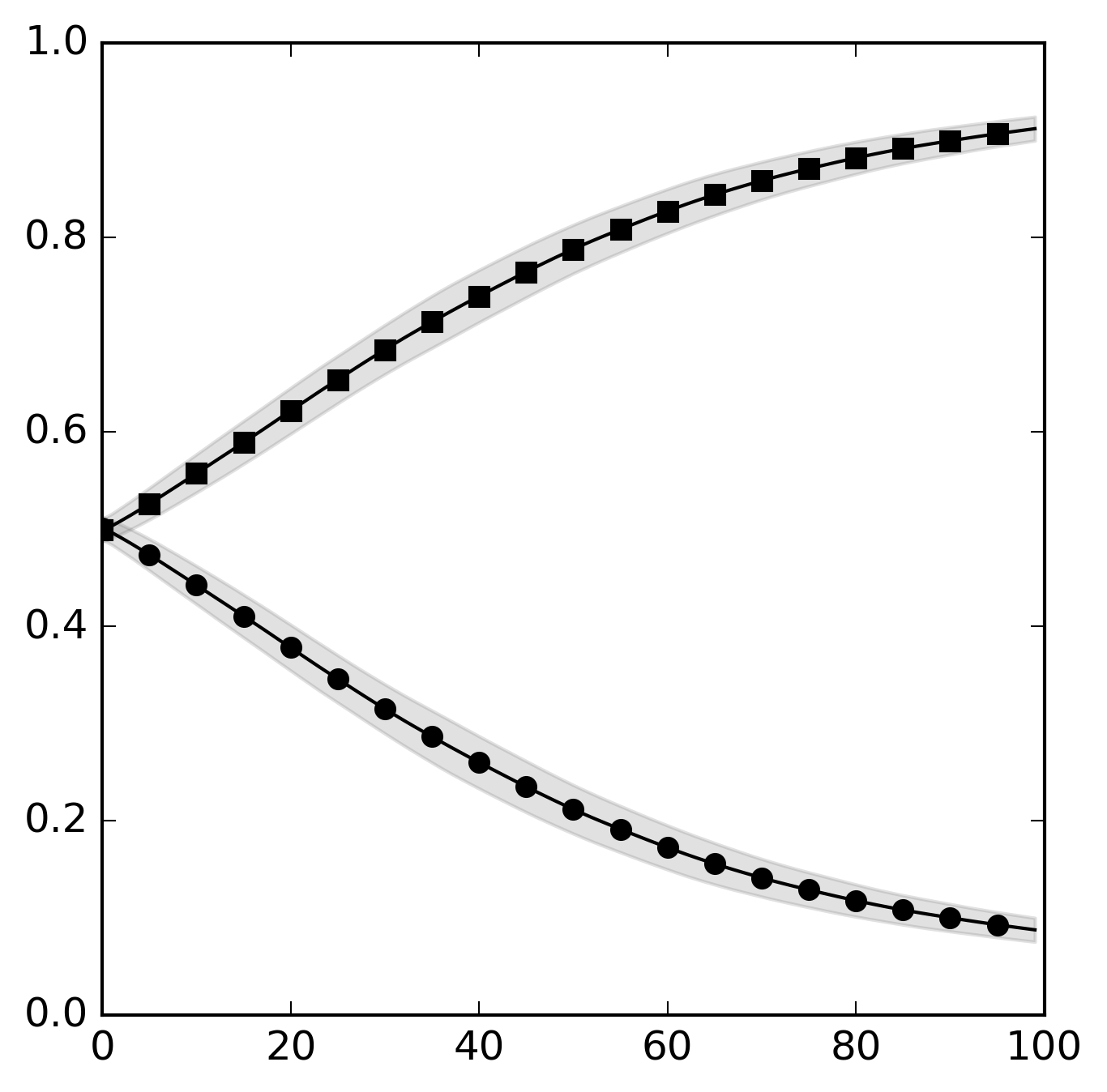}
    \end{minipage}
    \begin{minipage}[b]{0.23\textwidth}
        \centering
        \includegraphics[width=\textwidth]{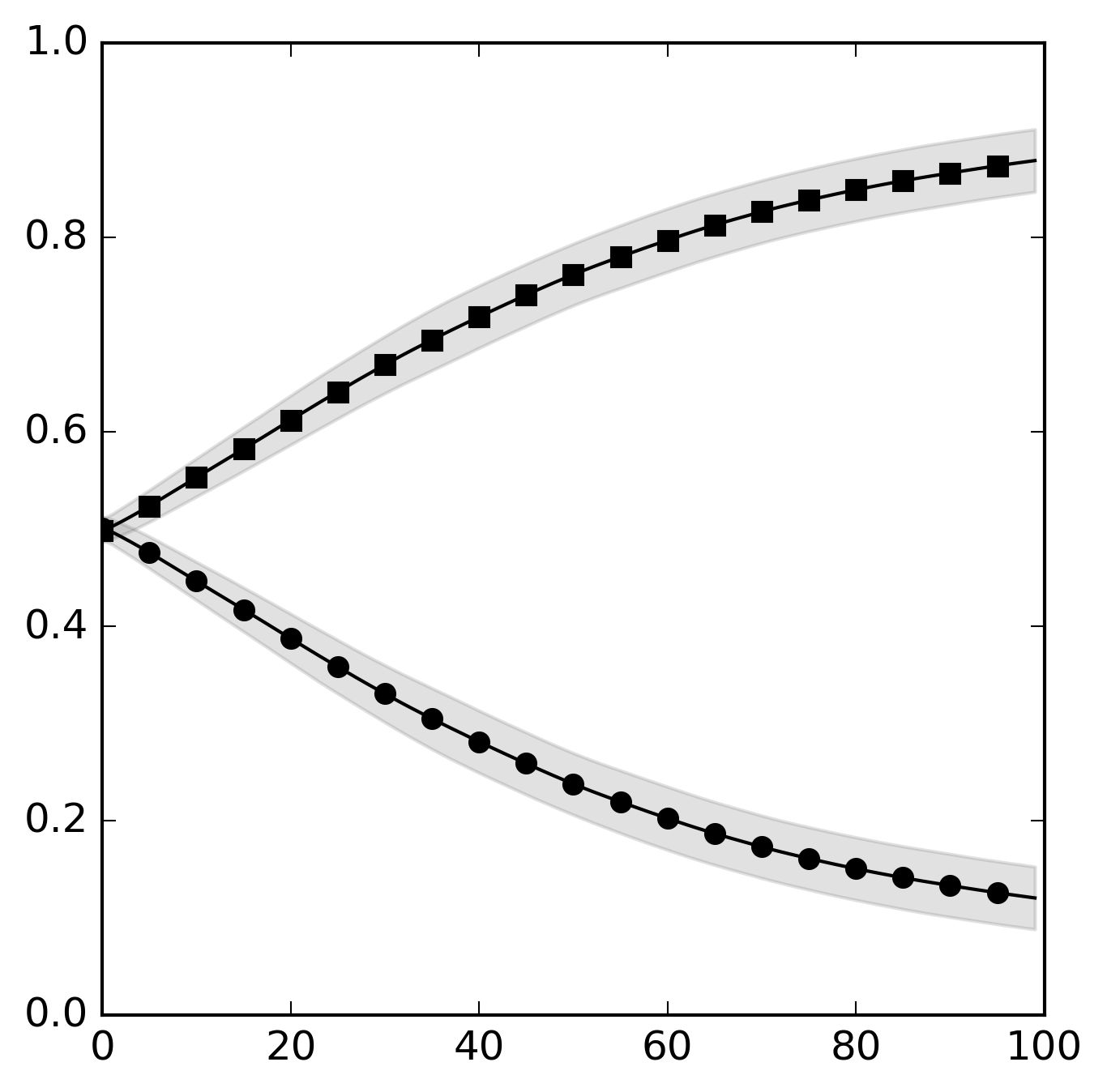}
    \end{minipage}
\label{fig:policies}
\end{figure}

\begin{figure}[!ht]
    \caption{\small Average (across repeats) training losses. Rows (up to down): D3QN, MAPPO Actor, MAPPO Critic. Columns (left to right): $(n=2, \sigma(\beta)=0.0)$, $(n=2, \sigma(\beta)=0.5)$, $(n=5, \sigma(\beta)=0.0)$, $(n=5, \sigma(\beta)=0.5)$.}
    \centering
    \begin{minipage}[b]{0.23\textwidth}
    \centering
        \includegraphics[width=\textwidth]{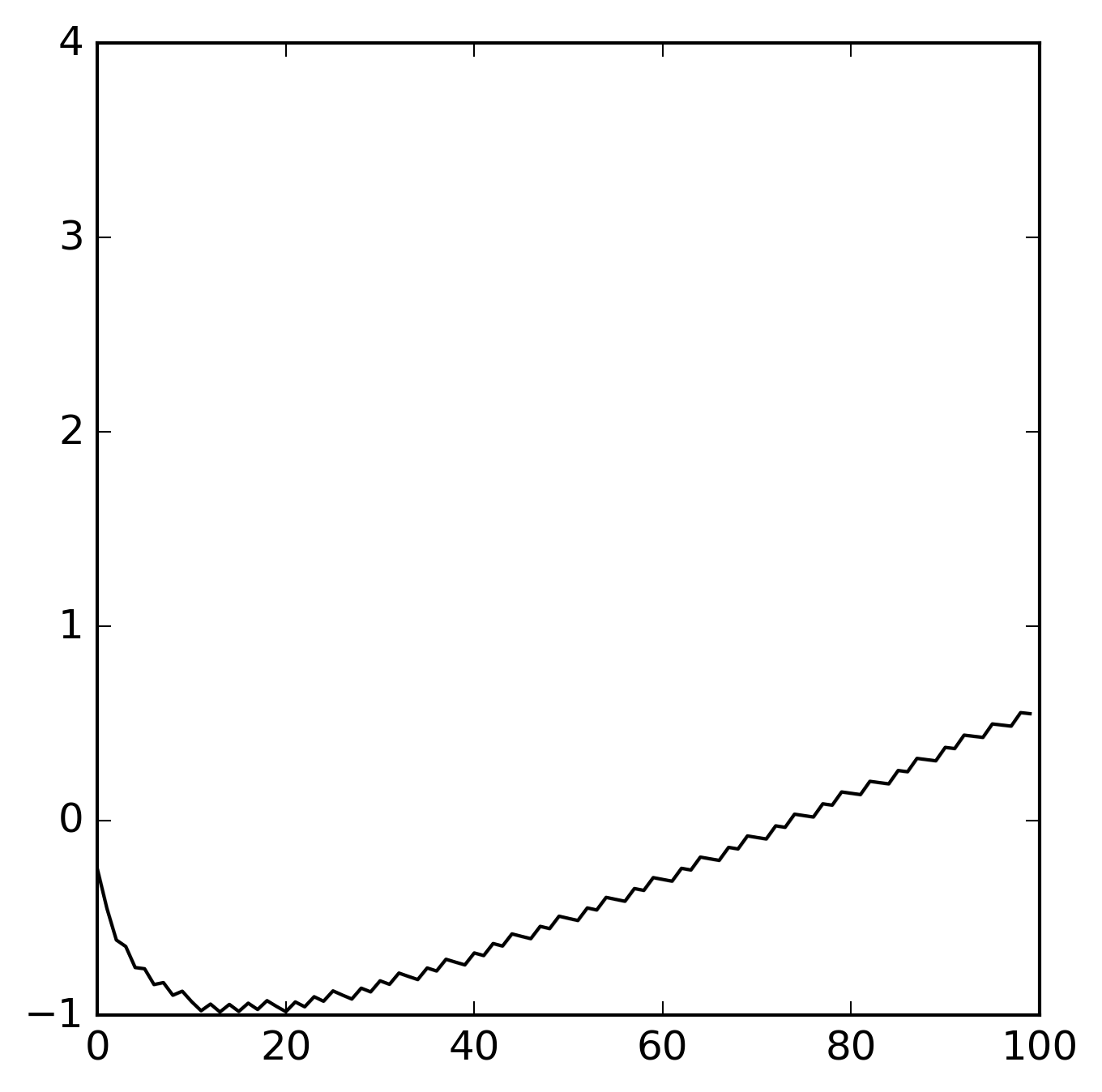}
    \end{minipage}
    \begin{minipage}[b]{0.23\textwidth}
        \centering
        \includegraphics[width=\textwidth]{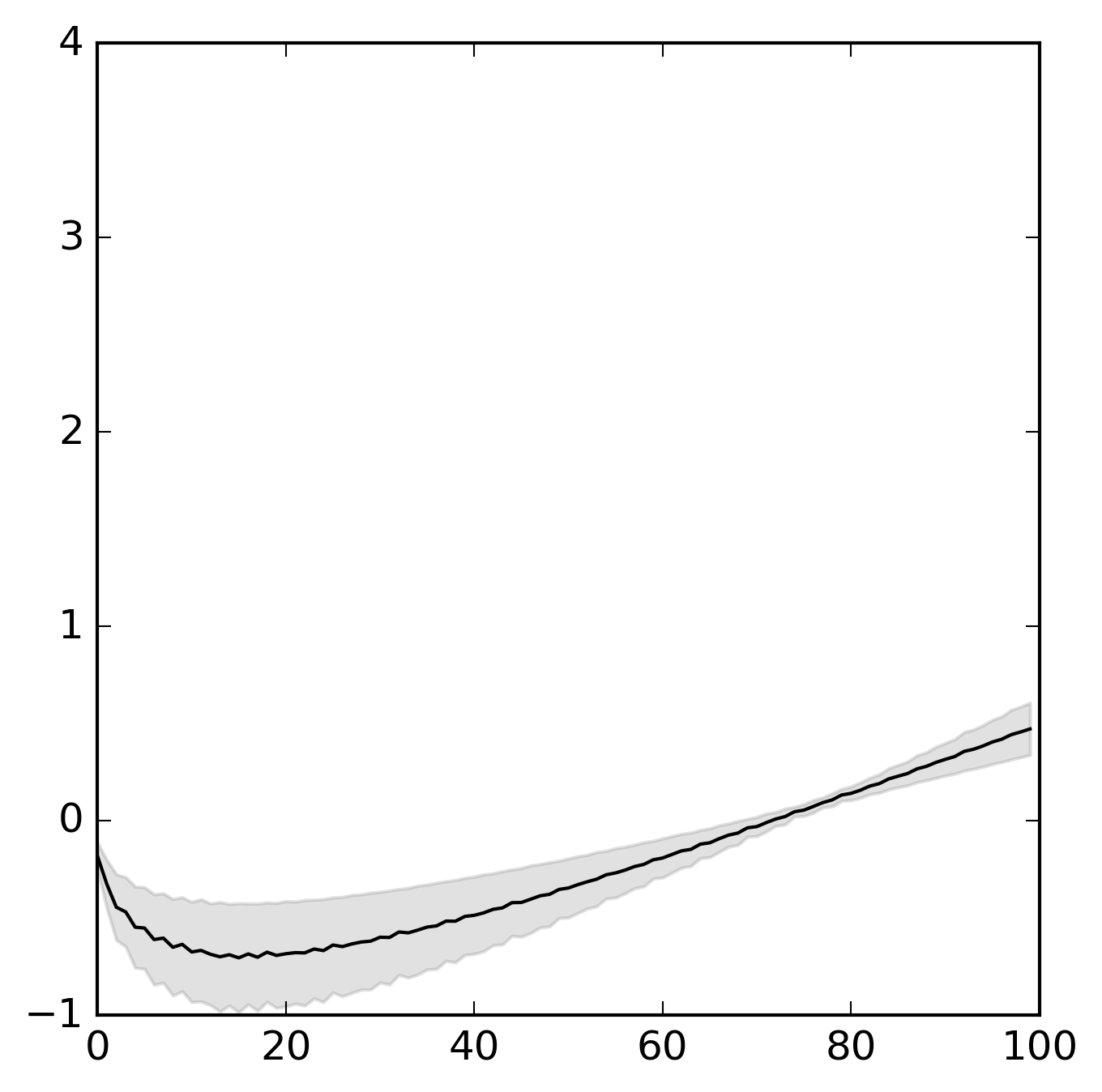}
    \end{minipage}
    \begin{minipage}[b]{0.23\textwidth}
        \centering
        \includegraphics[width=\textwidth]{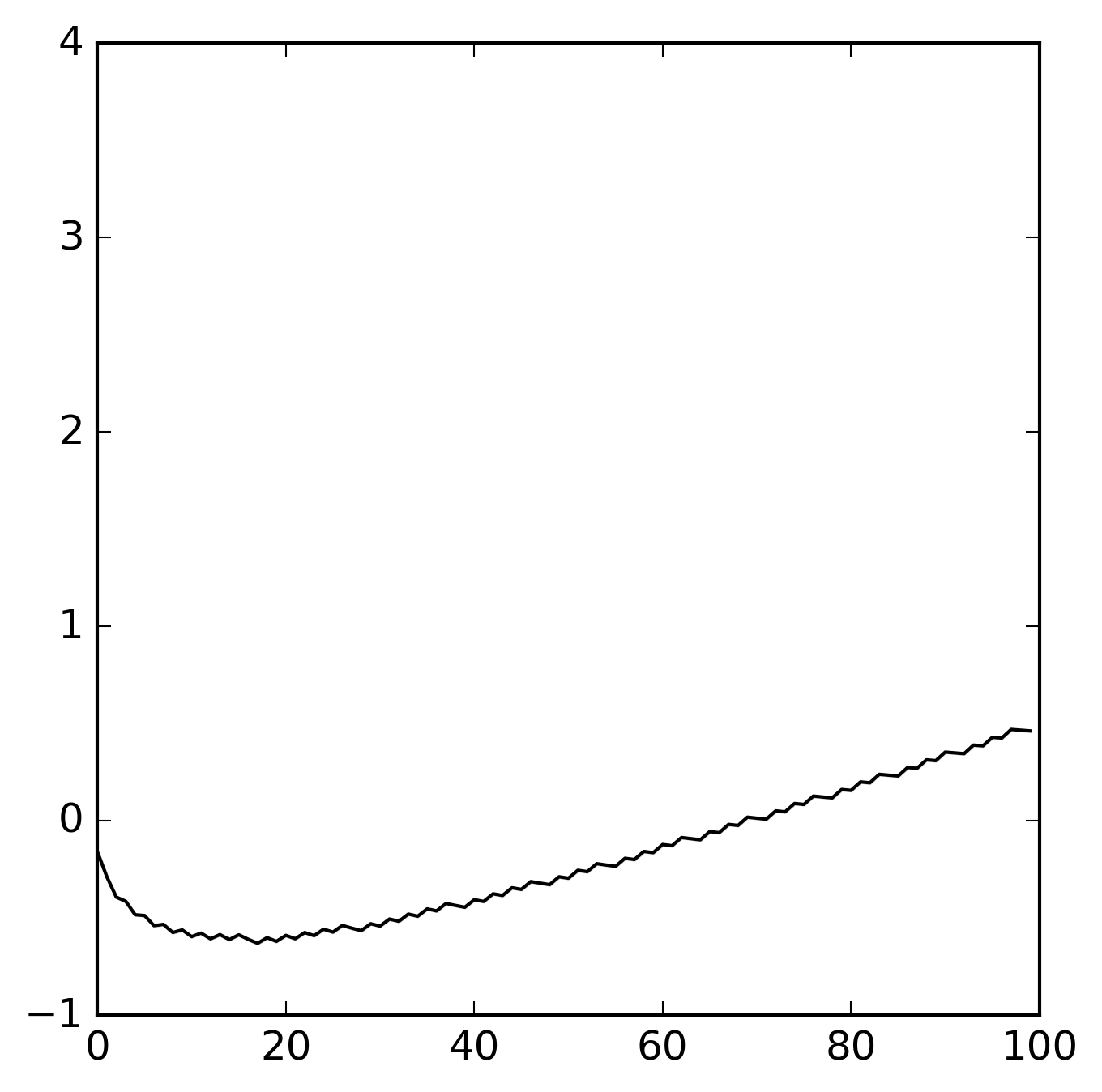}
    \end{minipage}
    \begin{minipage}[b]{0.23\textwidth}
        \centering
        \includegraphics[width=\textwidth]{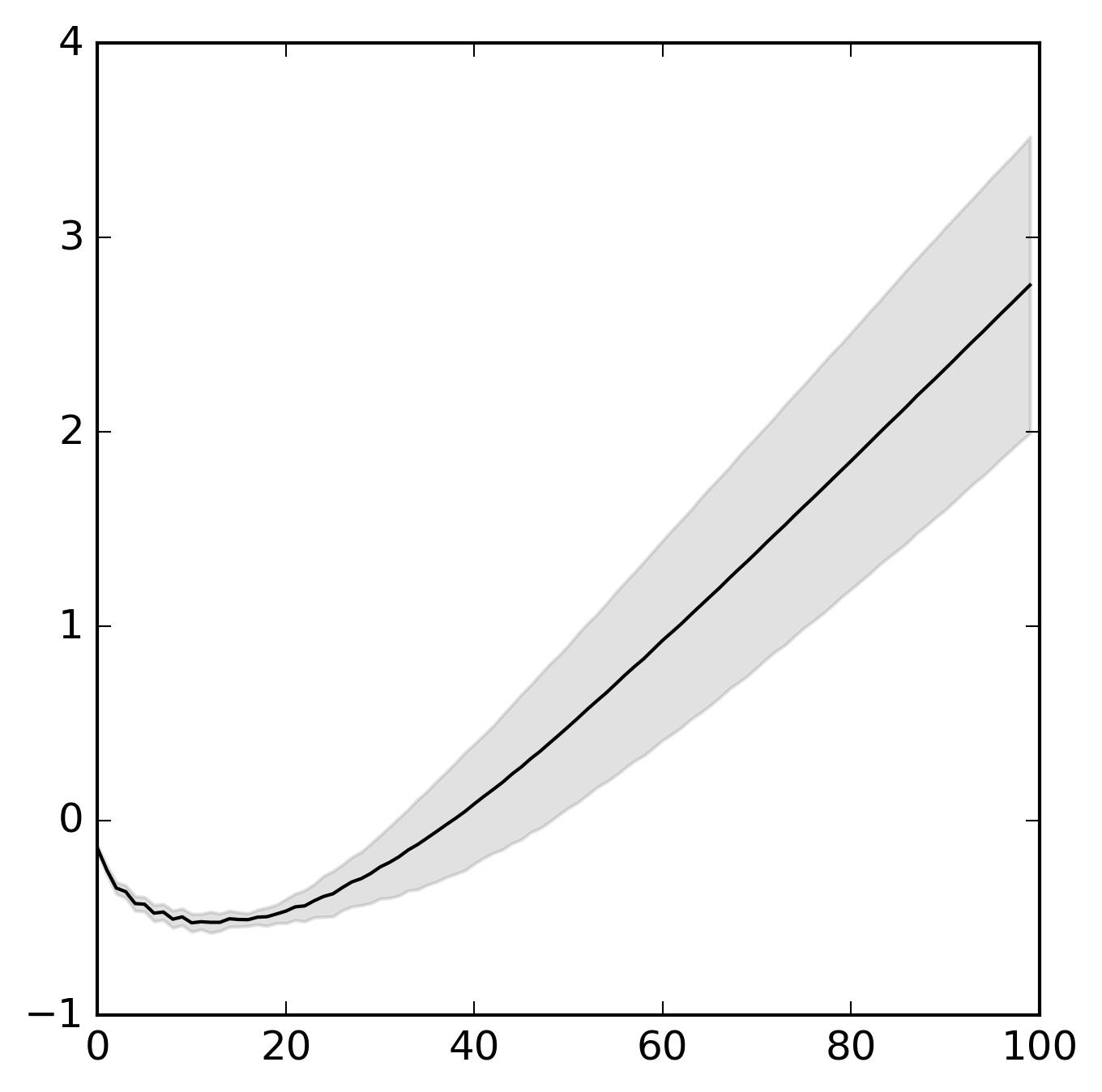}
    \end{minipage}
    \vfill
        \begin{minipage}[b]{0.23\textwidth}
    \centering
        \includegraphics[width=\textwidth]{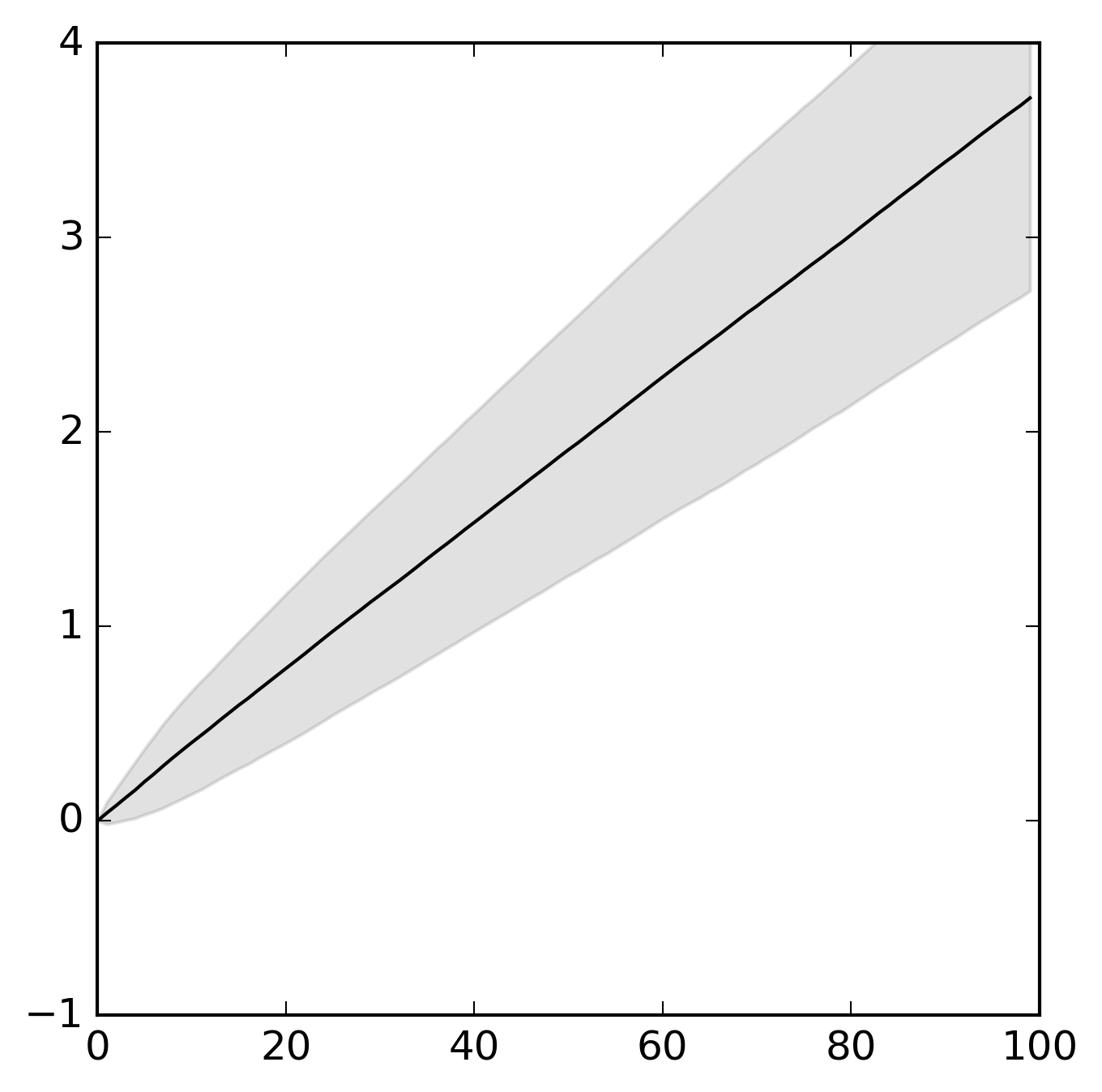}
    \end{minipage}
    \begin{minipage}[b]{0.23\textwidth}
        \centering
        \includegraphics[width=\textwidth]{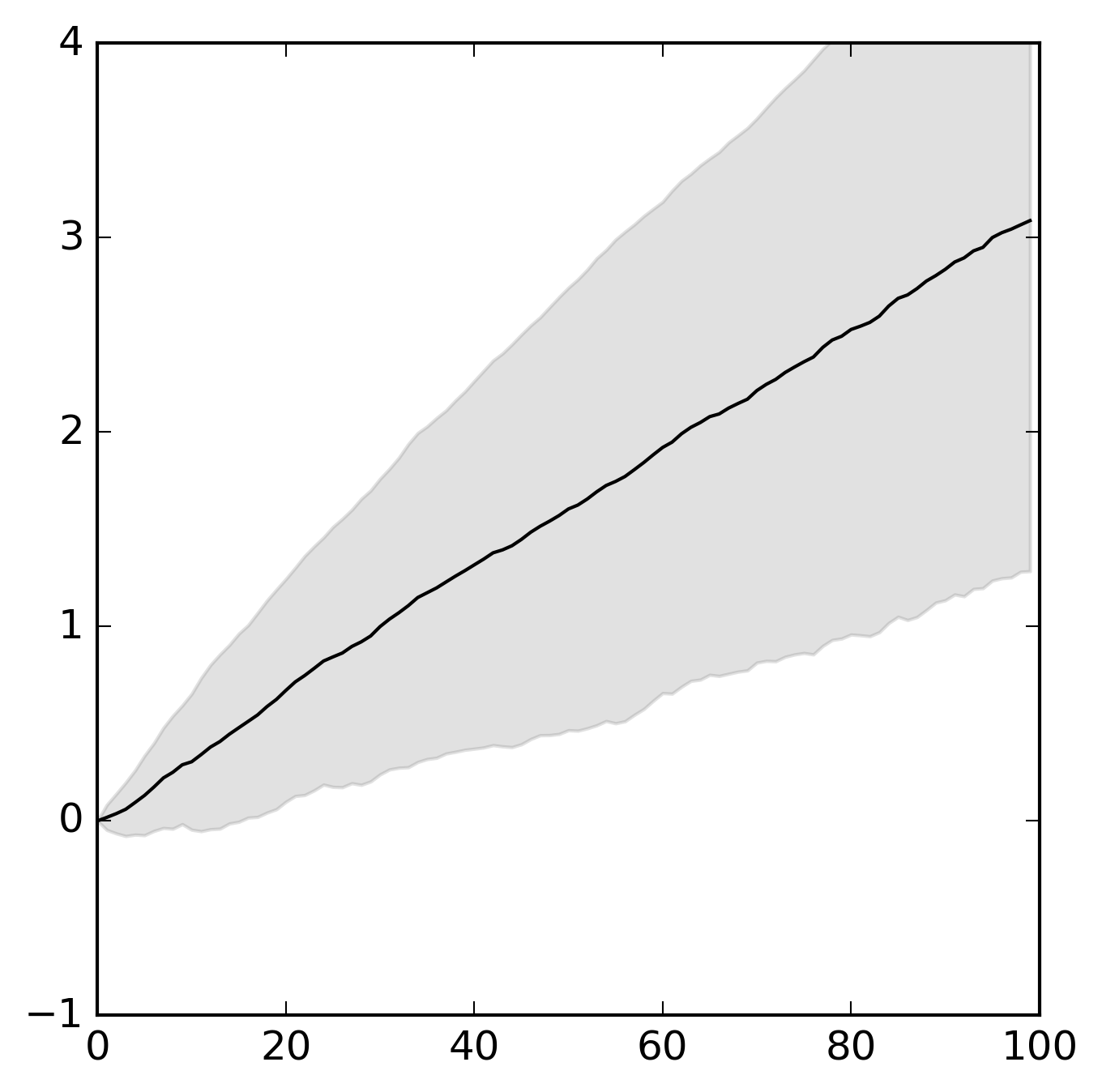}
    \end{minipage}
    \begin{minipage}[b]{0.23\textwidth}
        \centering
        \includegraphics[width=\textwidth]{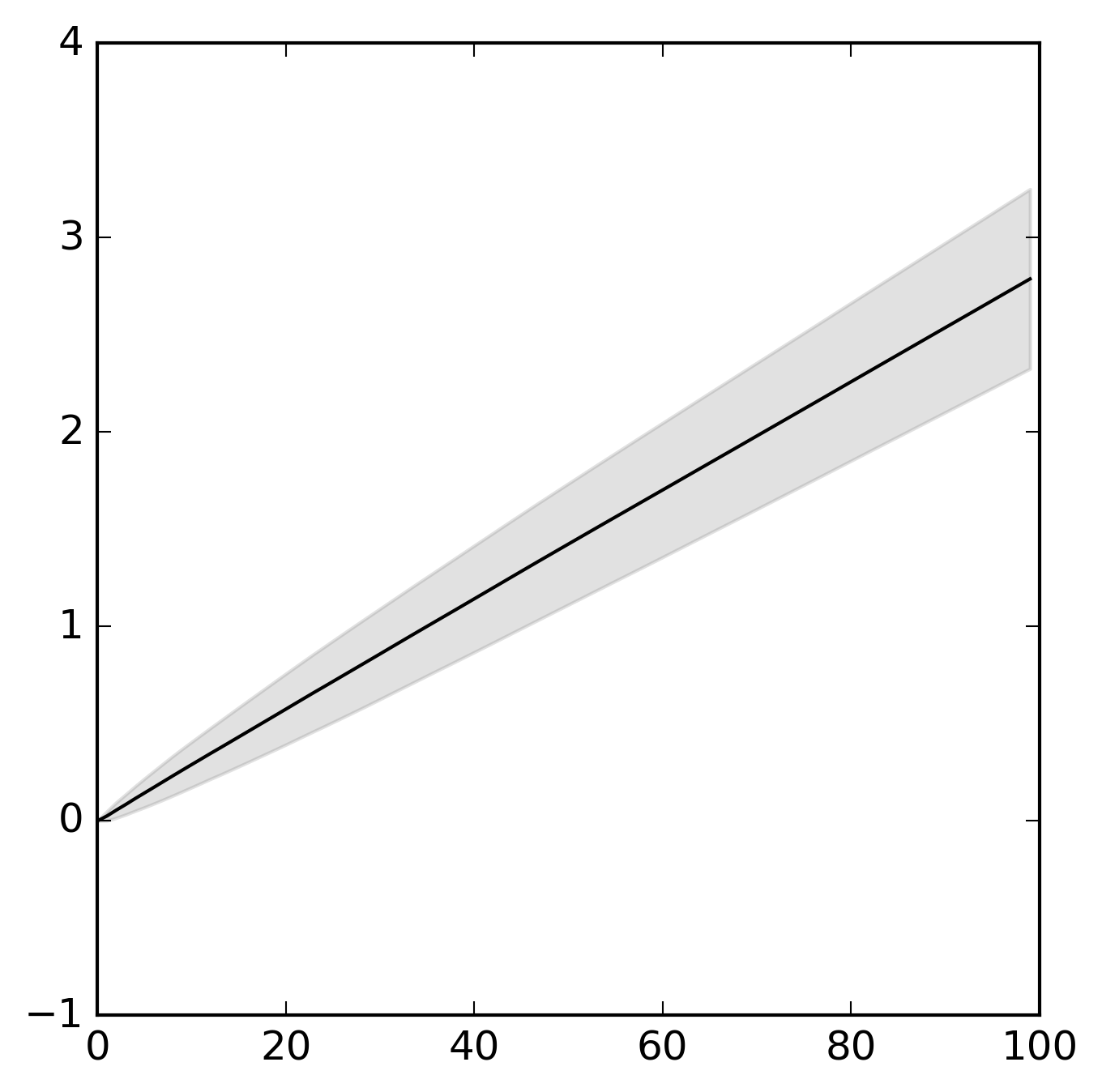}
    \end{minipage}
    \begin{minipage}[b]{0.23\textwidth}
        \centering
        \includegraphics[width=\textwidth]{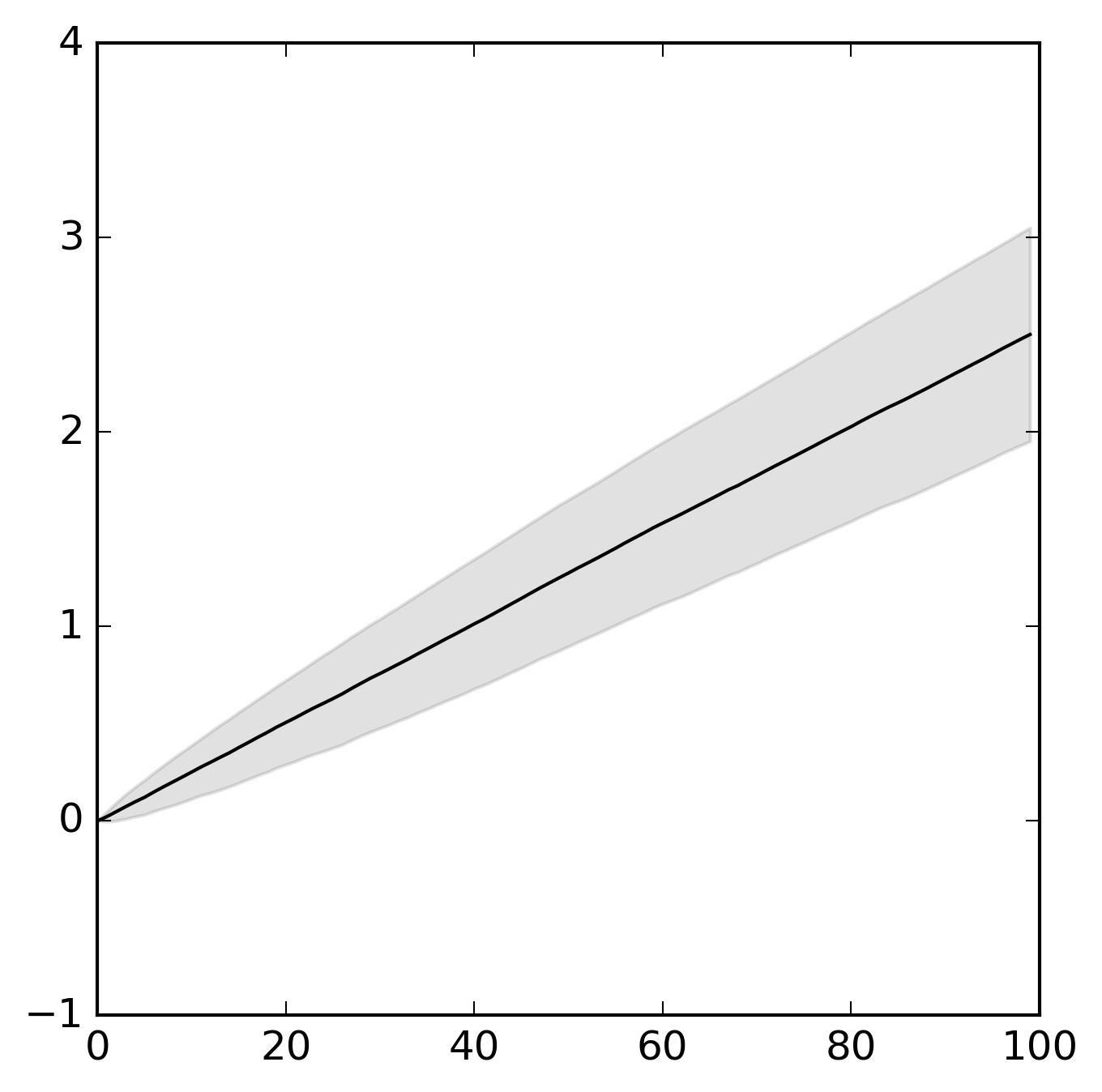}
    \end{minipage}
    \vfill
        \begin{minipage}[b]{0.23\textwidth}
    \centering
        \includegraphics[width=\textwidth]{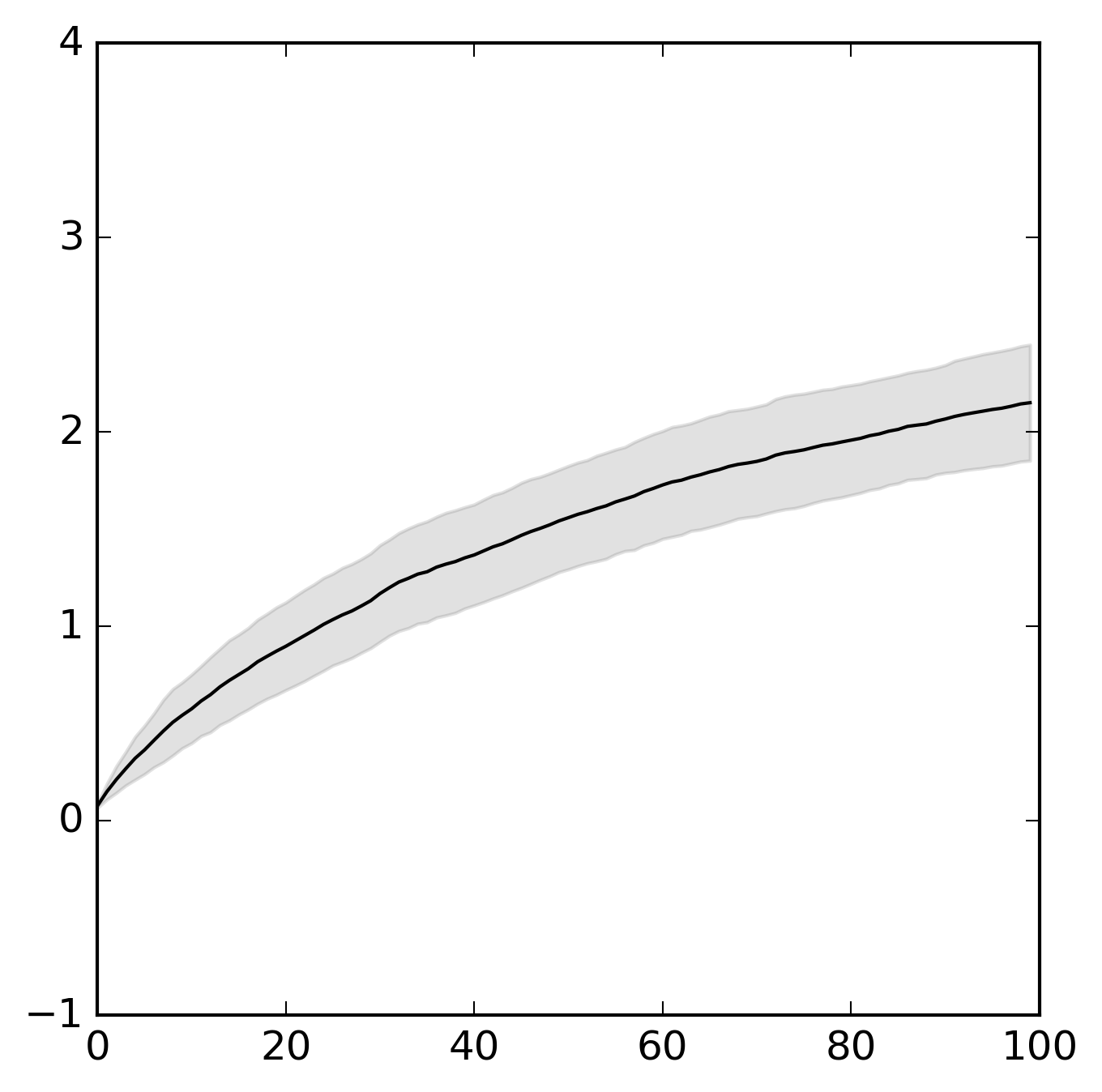}
    \end{minipage}
    \begin{minipage}[b]{0.23\textwidth}
        \centering
        \includegraphics[width=\textwidth]{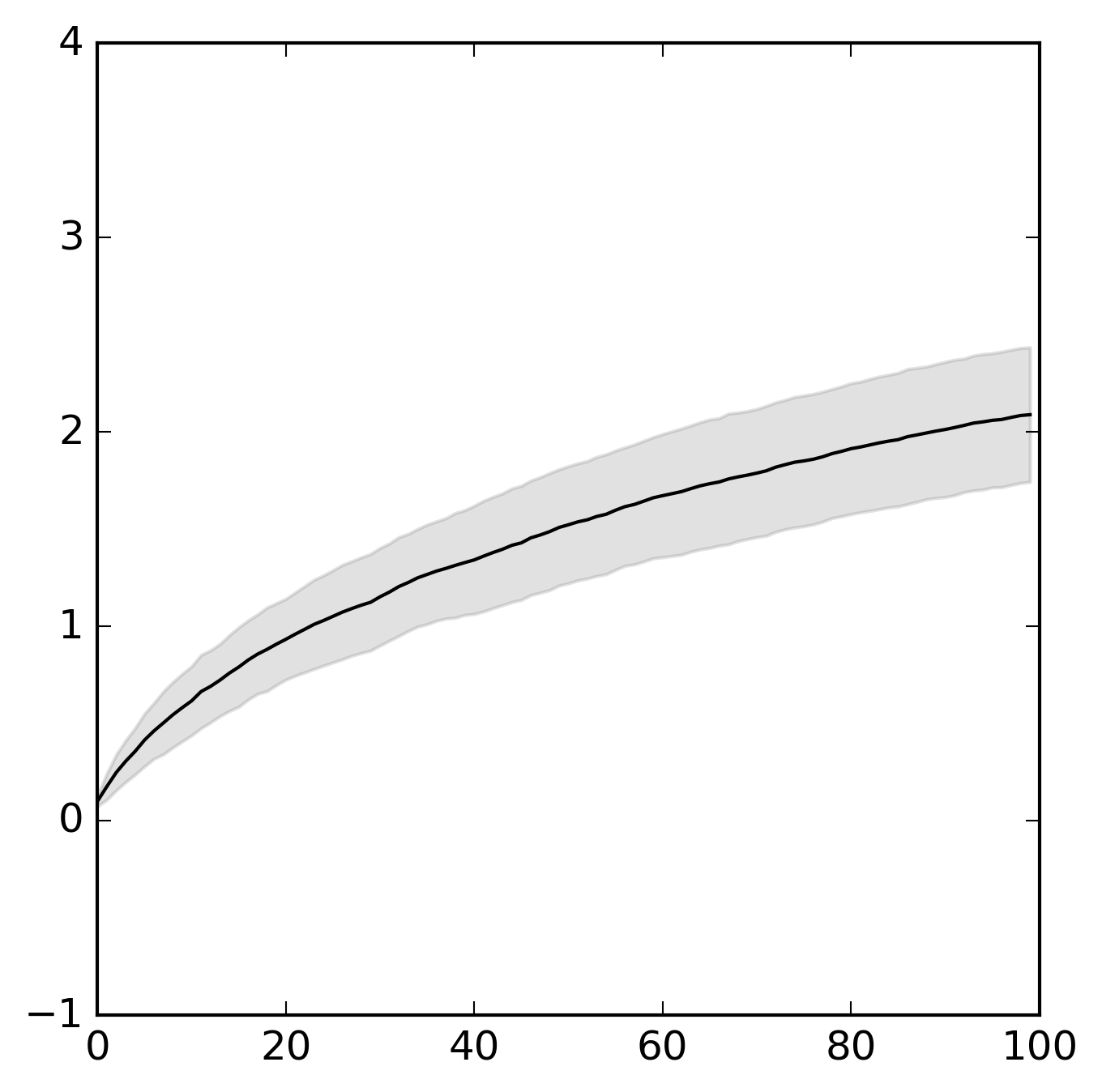}
    \end{minipage}
    \begin{minipage}[b]{0.23\textwidth}
        \centering
        \includegraphics[width=\textwidth]{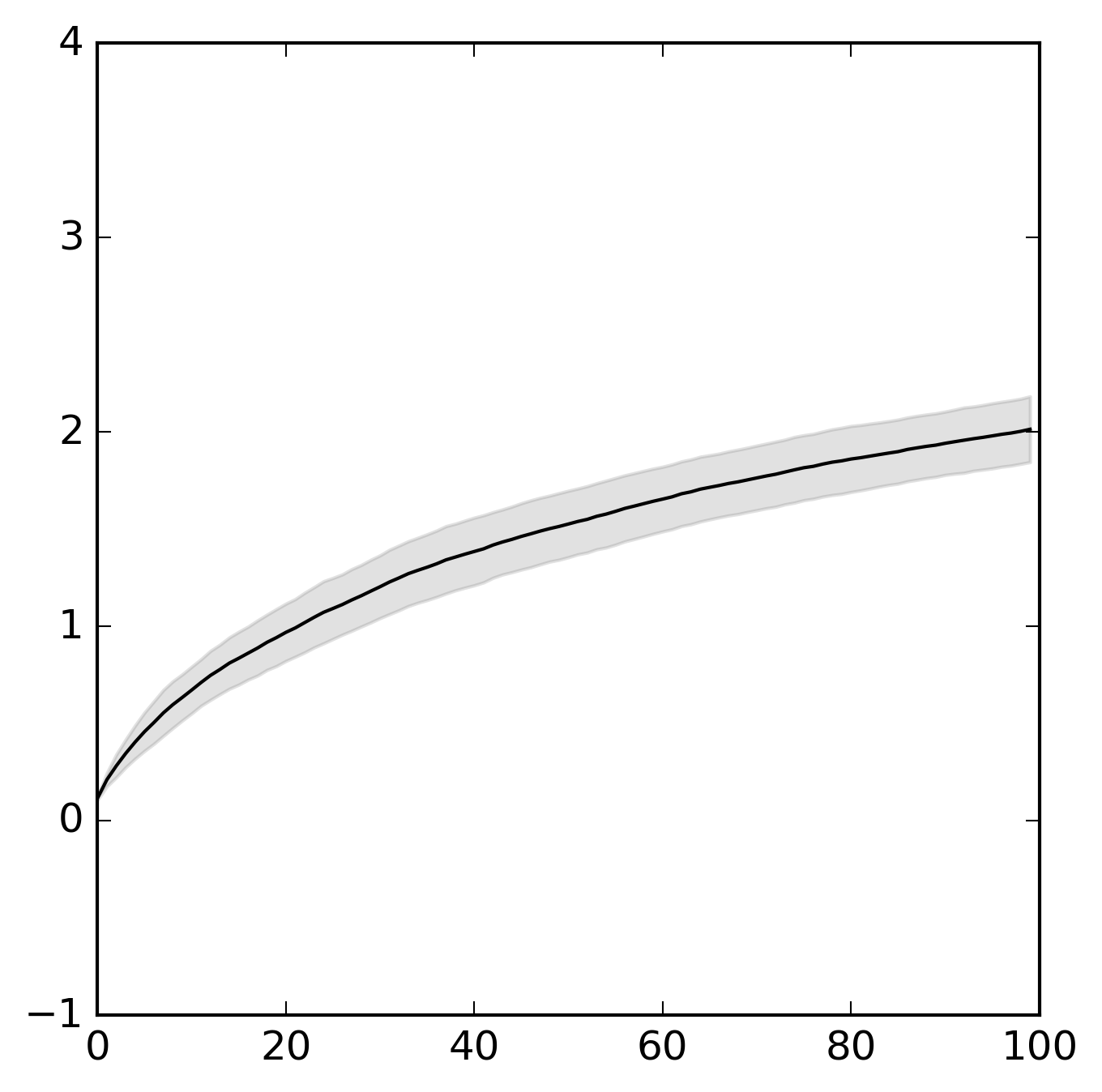}
    \end{minipage}
    \begin{minipage}[b]{0.23\textwidth}
        \centering
        \includegraphics[width=\textwidth]{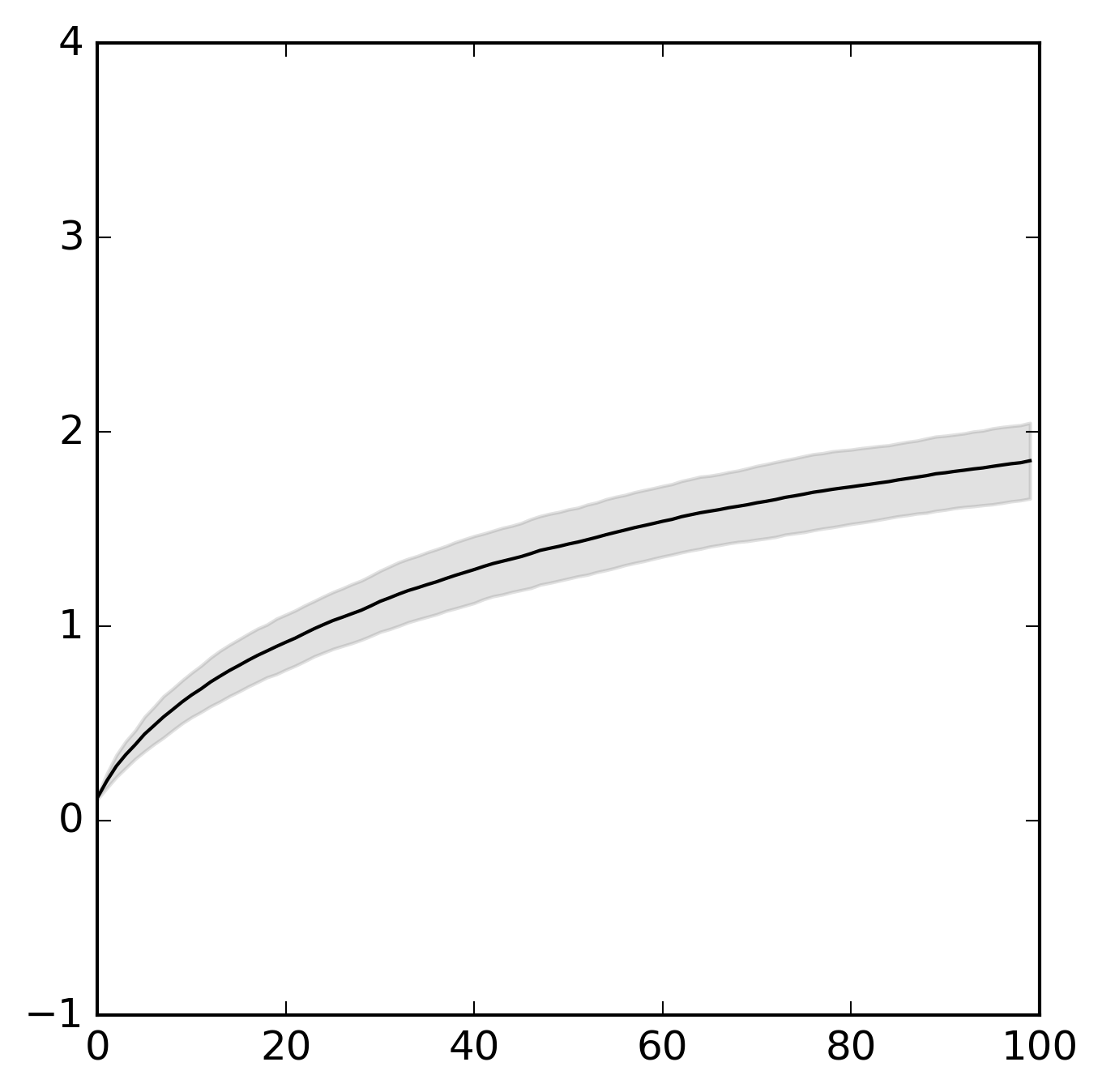}
    \end{minipage}
\label{fig:cumulative_regrets}
\end{figure}

\clearpage

\section{Implementation Details}\label{app:ImplementationDetails}
To replicate the study, follow  \href{https://github.com/IgorSadoune/algorithmic-collusion-and-the-minimum-price-markov-game}{GitHub}. 

\paragraph{E-Greedy Method.} Our \texttt{EpsilonGreedyMAB} class utilizes an \(\epsilon\)-greedy strategy with the following hyperparameters:
\begin{itemize}
    \item \texttt{n\_arms}: Number of possible actions (arms) for the bandit.
    \item \texttt{epsilon}: Exploration rate, set to 0.3 in the \texttt{initialize\_agents} function.
\end{itemize}

\noindent The agent selects a random action with probability \(\epsilon\) and the best known action with probability \(1 - \epsilon\). The action values are updated using an incremental average formula
\[
Q_{t+1}(a) = Q_t(a) + \frac{1}{N(a)} (R_t - Q_t(a))
\]
where \(Q_t(a)\) is the estimated value of action \(a\) at time \(t\), \(N(a)\) is the number of times action \(a\) has been selected, and \(R_t\) is the reward received at time \(t\).

\paragraph{Upper Confidence Bound.} Our \texttt{UCBMAB} class implements the UCB algorithm with the following hyperparameters:
\begin{itemize}
    \item \texttt{n\_arms}: Number of possible actions (arms) for the bandit.
\end{itemize}
\noindent The agent selects actions based on confidence bounds calculated as:
\[
\text{UCB}_t(a) = Q_t(a) + \sqrt{\frac{2 \log t}{N(a)}}
\]
where \(Q_t(a)\) is the estimated value of action \(a\) at time \(t\), \(N(a)\) is the number of times action \(a\) has been selected, and \(t\) is the total number of pulls. Action values are updated using the incremental average formula mentioned above.

\paragraph{Thompson Sampling.} Our \texttt{ThompsonSamplingMAB} class uses a Bayesian approach to action selection with the following hyperparameters:
\begin{itemize}
    \item \texttt{n\_arms}: Number of possible actions (arms) for the bandit.
    \item $\alpha$: Beta distribution parameter
    \item $\beta$: Beta distribution parameter
\end{itemize}
\noindent The agent maintains Beta distributions for each action, initialized with \(\alpha = 1\) and \(\beta = 1\) for each arm. Actions are selected based on samples drawn from these Beta distributions. The Beta distribution parameters are updated based on the received binary rewards as follows:
\begin{align*}
\alpha_{\text{new}} &= \alpha_{\text{old}} + \text{reward} \\
\beta_{\text{new}} &= \beta_{\text{old}} + 1 - \text{reward}
\end{align*}
This method effectively balances exploration and exploitation by sampling from the probability distributions of each action's expected reward.

\paragraph{Deep Dueling DQN (D3QN).}
The \texttt{DuelingDQN} class implements a Dueling Deep Q-Network with the following architecture:
\begin{itemize}
    \item \texttt{feature\_layer}: A fully connected layer with 128 units followed by a ReLU activation function.
    \item \texttt{value\_stream}: A fully connected layer with 128 units followed by a ReLU activation function, and another fully connected layer with 1 unit representing the state value.
    \item \texttt{advantage\_stream}: A fully connected layer with 128 units followed by a ReLU activation function, and another fully connected layer with \texttt{action\_dim} units representing the advantage values for each action.
\end{itemize}
\noindent The Q-values are computed as:
\[
Q(s, a) = V(s) + \left( A(s, a) - \frac{1}{|\mathcal{A}|} \sum_{a'} A(s, a') \right)
\]
where \( V(s) \) is the state value, \( A(s, a) \) is the advantage of action \( a \) in state \( s \), and \( \mathcal{A} \) is the set of all possible actions.

\paragraph{Multiagent Proximal Policy Optimization (MAPPO).} The \texttt{MLPNetwork} class implements a multi-layer perceptron network with the following architecture:
\begin{itemize}
    \item \texttt{state\_dim}: Dimensionality of the state space.
    \item \texttt{action\_dim}: Dimensionality of the action space.
    \item \texttt{layers}: A fully connected network with two hidden layers, each with 128 units followed by a ReLU activation function.
\end{itemize}

\noindent Our \texttt{MAPPOAgent} class implements a MAPPO agent with the following hyperparameters:
\begin{itemize}
    \item \texttt{state\_dim}: Dimensionality of the state space.
    \item \texttt{action\_dim}: Dimensionality of the action space.
    \item \texttt{lr}: Learning rate, set to $1 \times 10^{-3}$.
    \item \texttt{gamma}: Discount factor, set to 0.99.
    \item \texttt{clip\_param}: Clipping parameter for PPO, set to 0.2.
    \item \texttt{c1}: Coefficient for value function loss, set to 0.5.
    \item \texttt{c2}: Coefficient for entropy bonus, set to 0.01.
    \item \texttt{epsilon}: Initial exploration rate, set to 1.0.
    \item \texttt{epsilon\_decay}: Decay rate of epsilon, set to 0.995.
    \item \texttt{epsilon\_min}: Minimum exploration rate, set to 0.01.
\end{itemize}

The agent selects actions using an \(\epsilon\)-greedy strategy:
\begin{itemize}
    \item With probability \(\epsilon\), select a random action.
    \item With probability \(1 - \epsilon\), select the action with the highest policy probability:
    \[
    a_t = \arg\max_a \pi(a|s_t)
    \]
    \item The exploration rate \(\epsilon\) is decayed over time:
    \[
    \epsilon = \max(\epsilon \times \text{epsilon\_decay}, \text{epsilon\_min})
    \]
\end{itemize}

The agent stores experiences in a memory buffer for training:
\begin{itemize}
    \item Store experiences as tuples of \texttt{(state, action, reward, next\_state, done)}.
\end{itemize}

The Generalized Advantage Estimation (GAE) is computed using the following formula:
\[
\delta_t = r_t + \gamma V(s_{t+1}) - V(s_t)
\]
\[
\hat{A}_t = \sum_{l=0}^{\infty} (\gamma \lambda)^l \delta_{t+l}
\]
and the policy and value are updated as
\begin{itemize}
    \item \textbf{Policy Update (Actor)}:
    \[
    \text{ratio} = \frac{\pi(a_t|s_t)}{\pi_{\text{old}}(a_t|s_t)}
    \]
    \[
    \mathcal{L}^{\text{CLIP}}(\theta) = \mathbb{E}_t \left[ \min(\text{ratio} \cdot \hat{A}_t, \text{clip}(\text{ratio}, 1-\epsilon, 1+\epsilon) \cdot \hat{A}_t) \right]
    \]
    \item \textbf{Value Update (Critic)}:
    \[
    \mathcal{L}^{\text{VF}} = \left( \hat{R}_t - V(s_t) \right)^2
    \]
    \item \textbf{Total Loss}:
    \[
    \mathcal{L} = \mathcal{L}^{\text{CLIP}} + c_1 \mathcal{L}^{\text{VF}} - c_2 \mathcal{H}
    \]
    where \(\mathcal{H}\) is the entropy bonus.
\end{itemize}

\end{document}